\documentclass[aps,prx,reprint,superscriptaddress,longbibliography,floatfix]{revtex4-2}

\usepackage{amsmath,amssymb,amsthm,mathtools}
\usepackage{bm}

\usepackage{graphicx}
\usepackage[dvipsnames,svgnames,x11names]{xcolor}

\usepackage{booktabs}
\usepackage{multirow}
\usepackage{array}
\usepackage{makecell}

\usepackage[colorlinks=true,linkcolor=RoyalBlue,citecolor=ForestGreen,urlcolor=Maroon]{hyperref}
\usepackage{cleveref}
\crefname{equation}{Eq.}{Eqs.}
\Crefname{equation}{Equation}{Equations}
\crefname{figure}{Fig.}{Figs.}
\Crefname{figure}{Figure}{Figures}
\crefname{section}{Sec.}{Secs.}
\Crefname{section}{Section}{Sections}
\crefname{table}{Table}{Tables}
\Crefname{table}{Table}{Tables}
\creflabelformat{equation}{#2(#1)#3}

\usepackage{enumitem}
\usepackage{float}

\newtheorem{axiom}{Axiom}
\newtheorem{definition}{Definition}
\newtheorem{theorem}{Theorem}
\newtheorem{lemma}[theorem]{Lemma}
\newtheorem{corollary}[theorem]{Corollary}

\newtheorem{construction}{Construction}

\theoremstyle{remark}
\newtheorem{remark}{Remark}

\newcommand{\R}{\mathbb{R}}
\newcommand{\C}{\mathbb{C}}

\newcommand{\Sph}{\mathbb{S}}
\newcommand{\calL}{\mathcal{L}}
\newcommand{\calM}{\mathcal{M}}
\newcommand{\calF}{\mathcal{F}}
\newcommand{\calH}{\mathcal{H}}
\newcommand{\calR}{\mathcal{R}}
\newcommand{\calO}{\mathcal{O}}
\newcommand{\calP}{\mathcal{P}}
\newcommand{\calB}{\mathcal{B}}
\newcommand{\calE}{\mathcal{E}}
\newcommand{\calN}{\mathcal{N}}
\newcommand{\calV}{\mathcal{V}}

\newcommand{\deff}{d_{\mathrm{eff}}}

\newcommand{\tr}{\operatorname{tr}}
\newcommand{\diag}{\operatorname{diag}}

\graphicspath{{figs/}}

\begin{document}

\title{The Geometry of Semantic Space:\\
A Continuous Geometric Framework for the Transformer Architecture}

\author{Zhihua Liang}
\email{zhihua.liang@ca.infn.it}
\affiliation{INFN, Sezione di Cagliari, I-09042 Monserrato (CA), Italy}

\date{\today}

\begin{abstract}
We present a continuous geometric framework that models the discrete
algebraic operations of the Transformer architecture as an
integro-differential equation (IDE) on a semantic fiber bundle
$\calE = \calM \times \R^d$.  Beginning from a single geometric
axiom---that the token sequence forms a discrete $1$-manifold equipped
with a canonical measure lattice---we translate every core component of
the modern Transformer (RMSNorm, RoPE, Softmax Attention, FFN, Residual
Stream, SGD, Weight Decay) into a cohesive vocabulary of differential
geometry, measure theory, and stochastic calculus.  The resulting
framework yields quantitative predictions spanning entropic optimal
transport (Attention as a Schr\"odinger bridge) and non-equilibrium
thermodynamics (SGD as It\^{o} diffusion violating detailed balance).  We
conduct a six-part experimental campaign across five architectures
(Qwen3, LLaMA\nobreakdash-3.1, Gemma\nobreakdash-3, GPT-2, Mistral)
spanning $124$M to $8$B parameters.  The empirical observables are
quantitatively consistent with the geometric predictions:
the $\epsilon^{-1/2}$ Lipschitz scaling calibration at machine precision
($R^2 = 1.000$), the Lie--Trotter operator-splitting torsion,
the symmetric ablation instability confirming the Dual-Law of
Topological Stability,
the $\calO(1/\sqrt{k})$ thermodynamic suppression of Poincar\'e
recurrence on the RoPE torus,
the thermodynamic context-limit phase transition,
and the Non-Equilibrium Steady State parameter vortex---verified across
two optimizers (AdamW and Pure SGD) to exclude momentum artifacts.  The
results demonstrate that analyzing Transformers through the lens of
continuous stochastic differential geometry provides a predictive
descriptive vocabulary for the stability limits, context bounds, and
optimization dynamics of Large Language Models.
\end{abstract}

\maketitle

\section{Introduction}
\label{sec:introduction}

Analyzing deep learning through the lens of continuous dynamical systems has a substantial history.  Weinan E~\cite{weinan2017proposal} recast ResNets as forward Euler discretizations of continuous dynamical systems, leading to the foundational development of Neural Ordinary Differential Equations~\cite{chen2018neural}.  The non-equilibrium thermodynamic properties of Stochastic Gradient Descent have been rigorously explored by Chaudhari et al.~\cite{chaudhari2019entropy} via Entropy-SGD and by Mandt et al.~\cite{mandt2017stochastic} via the Bayesian interpretation of SGD as approximate inference.  Building upon these foundations, prior literature has explored the continuous limits of attention mechanisms, modeling Transformers~\cite{vaswani2017attention} via PDEs, ODEs, and interacting particle systems~\cite{lu2020understanding,sander2022sinkformers}.  Amari's foundational work on Information Geometry~\cite{amari2016information} established the dual-flat manifold structure of well-specified statistical models and the breakdown of the Information Matrix Equality for misspecified models---a distinction that directly underlies the non-equilibrium thermodynamics of \cref{sec:parameters}.  The empirical machine learning community has also extensively documented sequence-length scaling properties via Rotary Position Embeddings~\cite{su2024roformer} and its extensions~\cite{peng2024yarn,press2022train}, as well as context-window capacity constraints via the empirical discovery of Attention Sinks~\cite{xiao2024efficient}.

This manuscript aims not to claim that the Transformer \emph{is} a continuous physical system, but that analyzing it through the isomorphic lens of continuous stochastic differential geometry provides a predictive descriptive vocabulary that unifies these disparate empirical phenomena---bridging abstract topology with testable stability limits, context bounds, and optimization dynamics.  Following the logic of Backward Error Analysis in geometric numerical integration~\cite{hairer2006geometric}, we treat the continuous geometric framework as a \emph{Continuous Effective Field Theory}: the discrete architecture is the exact integrator, and the continuous integro-differential equation is the nearby modified equation whose structural generators (Lie brackets, vorticity 2-forms, entropic pressure) classify the dominant macroscopic observables.

\section{Axiomatization of Topological Spaces and Kinematics}
\label{sec:axiom}

To establish a mathematically rigorous foundation, we strip away the vocabulary of discrete ``arrays'' and formulate the Transformer via the kinematics of differential manifold calculus.  We summarize the translation between standard deep learning terminology and the continuous geometric equivalents in \cref{tab:notation}.

\begin{table*}[t]
\centering
\caption{Translation dictionary: standard deep learning terms and their continuous geometric equivalents used throughout this work.}
\label{tab:notation}
\begin{ruledtabular}
\begin{tabular}{l l}
Deep Learning Term & Geometric Equivalent \\
\hline
Sequence Index / Token Position & Base manifold coordinate $\mu,\nu\in\calM$ \\
Embedding Dimension ($d_{\mathrm{model}}$) & Fiber dimension $\calF_\mu\cong\R^d$ \\
Token Embedding / Hidden State & Section of the bundle $\Psi\in\Gamma(E)$ \\
RMSNorm~\cite{rmsnorm} $\epsilon$ stabilizer & Topological mollifier $\epsilon$ \\
RoPE (Rotary Embeddings) & Gauge connection / path monodromy $\mathcal{U}$ \\
Softmax Attention & Urysohn--Volterra operator $\mathcal{T}$ \\
Feed-Forward Network (FFN) & Hodge reaction field $\calR$ \\
Residual Stream & Depth-parameterized flow $\Psi(z,\mu)$ \\
Layer Index $l$ & Algorithmic depth $z\in\R^+$ \\
Weight Decay ($\lambda$) & Tikhonov gauge mass penalty \\
\end{tabular}
\end{ruledtabular}
\end{table*}

\subsection{The Base Manifold and the Semantic Bundle}
\label{sec:base-manifold}

\begin{remark}[Scope and Kinematic Classification]
\label{rem:kinematic-scope}
We classify this framework explicitly as a \emph{Classical Lattice Field Theory on a rigid Galilean background}.  Unlike Yang--Mills theory or General Relativity, where the gauge connection and the metric are dynamical fields possessing their own action and back-reacting to the matter field, the RoPE connection and the 1D base manifold in this framework are absolutely rigid, fixed background geometries.  The framework fundamentally lacks the diffeomorphism invariance of a true relativistic field theory.  The continuous geometry is an \emph{isomorphic descriptive lens} applied to a discrete algebraic system; the geometry is the map, not the territory.  Just as lattice QCD simulates continuous gauge fields via discrete computational grids, the discrete Transformer acts as the exact numerical integrator of the underlying continuous integro-differential equation.  The geometric constructions below are therefore not optional analytical overlays, but exact kinematic translations of the discrete computational graph whose topological constraints are empirically binding for trained networks: as demonstrated in \cref{sec:exp-resonance}, violating the continuous geometric generators of a mature, frozen configuration fundamentally shatters its dynamic stability.  The constraint is \emph{configurational} rather than architectural---training from initialization under the same algebraic constraint reaches stable basins (\cref{rem:configurational-scope}).
\end{remark}

\begin{axiom}[The Continuous Base Manifold \& Measure Lattice]
\label{ax:base-manifold}
We explicitly demarcate the continuous physical reality from its discrete numerical evaluation.  Let the semantic sequence space (semantic time) be strictly defined as a connected, one-dimensional smooth manifold with boundary, specifically the half-closed ray $\calM \cong [0,\infty)\subset\R$.  The absolute causal origin defines the strict topological boundary $\partial\calM \equiv \{0\}$.  The computational context window passed to the model is rigorously defined as a causally bounded discrete measure lattice $\Lambda\subset\calM$.  To formally transition between continuous manifold topology and discrete sequence evaluation, we define $\Lambda$ as the support of an empirical Radon counting measure on the Borel $\sigma$-algebra $\calB(\calM)$:
\begin{equation}
  \eta_\Lambda = \sum_{\mu_i\in\Lambda} \delta_{\mu_i}\,.
  \label{eq:radon-measure}
\end{equation}
We assign Greek letters $\mu,\nu\in\Lambda$ to represent specific coordinate evaluations.  All non-local integro-differential flows (e.g., Attention) are evaluated via Lebesgue--Stieltjes integration with respect to this empirical Radon measure ($d\eta_\Lambda$), ensuring strict measure-theoretic validity over the discrete lattice without requiring a domain-mapping formalism.
\end{axiom}

\begin{definition}[The Semantic Fiber Bundle]
\label{def:fiber-bundle}
We define the semantic space as a smooth rank-$d$ real vector bundle $E\xrightarrow{\pi}\calM$.  Because $\calM$ is contractible, the bundle is globally trivializable ($E\cong\calM\times\R^d$).  The architecture implicitly fixes a canonical global trivialization, rigorously justifying the application of constant global gauge endomorphisms uniformly across all coordinates.  At any coordinate $\mu$, the local fiber $\calF_\mu\equiv\pi^{-1}(\mu)\cong\R^d$ corresponds to the interaction space of an attention head.  We strictly equip $E$ with a canonical global Riemannian bundle metric $g$ (the Euclidean inner product on the fibers). This metric allows us to rigorously upgrade the vector bundle into an Exterior Algebra Bundle ($\Lambda E$) to sequester rotational noise within the Lie algebra $\mathfrak{so}(d)$.
\end{definition}

The sequence of tokens is mathematically formulated not merely as isolated spatial fields, but as a continuous trajectory governed by a \emph{Non-Autonomous Flow on the Manifold of Sections}.  Rather than forcing artificial continuity on the spatial sections and grappling with infinite-dimensional non-compactness, we define the state space cleanly as the Lebesgue space over the empirical measure:
\begin{equation}
  \calF \equiv L^\infty(\calM, E;\, \eta_\Lambda)\,.
  \label{eq:state-space}
\end{equation}
Because functions in this space are identified up to $\eta_\Lambda$-a.e.\ equivalence, and $\eta_\Lambda$ is a finite atomic measure, this infinite-dimensional Banach space structurally collapses into a finite-dimensional topology ($\R^{N\times d}$).  This grants the Heine--Borel theorem and the Extreme Value Theorem globally for free, completely bypassing the need for localized finite-dimensional compact closure constructions.  We explicitly note that the Heine--Borel collapse applies strictly to the kinematics of a \emph{fixed, finite} context window ($N<\infty$); the asymptotic thermodynamic limit $N\to\infty$ evaluated in \cref{thm:attn-sink} instead relies on the \emph{Banach--Alaoglu Theorem} (weak-$*$ sequential compactness of probability measures on the Alexandroff compactification), entirely bypassing the need for Heine--Borel on the state space.  Furthermore, the composed Lipschitz bound (\cref{eq:lipschitz-bound}) is strictly independent of the sequence length $N$, because the Attention integral evaluates as a convex sum over a probability measure ($\int w\,d\eta = 1$).  We analytically continue discrete layer depth $l$ into a continuous depth parameter $z\in\R^+$.  Under this native geometric formulation, the system state is strictly defined as a smooth curve $\Psi:\R^+\to\calF$ parameterized by algorithmic depth $z$.  The temporal evolution $\partial_z\Psi$ inherently exists not as an artificial pushforward along an extended dimension, but canonically as the tangent vector to the curve: $\dot{\Psi}(z)\in T_{\Psi(z)}\calF$.

\begin{remark}[Demarcation of the Continuum \& Non-Local Flow]
\label{rem:continuum}
It is a common analytical pitfall to directly conflate the discrete Transformer architecture with continuous local differential equations.  Because Attention inherently mixes information across the measure lattice, the layer-wise forward pass is formally a first-order Forward Euler numerical discretization of a continuous non-local integro-differential flow, evaluated on $\Lambda$.  Using Taylor's Inequality for Banach space mappings, the discrete residual stream advances as:
\begin{equation}
  \Psi_{z+1}(\mu) = \Psi_z(\mu)
    + \Delta z\cdot\tilde{\calN}\!\big[\rho_\epsilon(\Psi_z)\big](\mu)
    + \mathfrak{R}_z(\mu)\,,
  \label{eq:euler-step}
\end{equation}
where $\Delta z=1$.  The strict analytical bound for the local truncation deviation over the integration step $z\to z+1$ evaluated directly via the Banach space norm is therefore:
\begin{equation}
  \|\mathfrak{R}_z(\mu)\|
    \le \frac{1}{2}\sup_{\tau\in[z,z+1]}
    \left\|\frac{\partial^2\Psi}{\partial\tau^2}(\tau,\mu)\right\|_g\,.
  \label{eq:truncation-bound}
\end{equation}

To bridge the discrete sequence of layer weights $\{W^{(1)},\dots,W^{(L)}\}$ mapped to algorithmic depth, we formulate the continuous parameter curve $\theta(z):\R^+\to\mathrm{End}(E)$ strictly as a Riemannian Cubic Spline by minimizing the covariant acceleration ($\min\int\|\nabla_{\dot\theta}\dot\theta\|^2_g\,dz$) on the parameter manifold subject to the discrete layer evaluations.  This mathematically guarantees $C^2$ smoothness.  Under this strict mathematical sanitation, let $\tilde{\calN}$ be the pure non-local Attention operator acting on bounded vectors, and let $\calF(z,\Psi) = \tilde{\calN}(z,\rho_\epsilon(\Psi))$ be the composed total flow.  The continuous vector field $\calF(z,\Psi)$ is smoothly dependent on the depth parameter $z$ (a non-autonomous flow), eliminating Dirac delta discontinuities in the temporal weight acceleration.  The depth-parameterized trajectory acceleration along the flow expands naturally via the continuous chain rule on the manifold of sections:
\begin{equation}
  \ddot\Psi(z)
    = \frac{\partial\tilde{\calN}}{\partial z}\!\big(z,\rho_\epsilon(\Psi)\big)
    + \left(D\tilde{\calN}\big|_{\rho_\epsilon(\Psi)}
      \circ D\rho_\epsilon\big|_\Psi\right)
      \!\big[\dot\Psi(z)\big]\,.
  \label{eq:acceleration}
\end{equation}

Here, $D\tilde{\calN}$ is the global Fr\'echet derivative acting on the infinite-dimensional Banach space of sections $\Gamma(E)$, and $D\rho_\epsilon$ is a local fiber-wise bundle endomorphism acting on the tangent space of the finite-dimensional fiber $T_{\Psi(\mu)}\calF_\mu\cong\R^d$.  To rigorously compose these functional derivatives, the local Jacobian $D\rho_\epsilon$ is lifted to a Nemytskii operator acting point-wise on the global section variation $H\in\Gamma(E)$: $\big(D\rho_\epsilon|_\Psi[H]\big)(\mu) \equiv D\rho_\epsilon|_{\Psi(\mu)}\!\big(H(\mu)\big)$. Before stating the eigenvalues of this Jacobian, we explicitly define the mathematical decomposition of the tangent space $T_{\Psi(\mu)}\calF_\mu$ into orthogonal components.  Any tangent vector $v \in T_{\Psi(\mu)}\calF_\mu$ can be uniquely decomposed as $v = v_\parallel + v_\perp$, where $v_\parallel \equiv \frac{\langle \Psi, v \rangle_g}{\|\Psi\|^2}\Psi$ is the component parallel to the state vector ($v_\parallel \parallel \Psi$) and $v_\perp \equiv v - v_\parallel$ is the component orthogonal to it ($v_\perp \perp_g \Psi$). Taking the Fr\'echet derivative of $\rho_\epsilon(\Psi)$ on a tangent vector $v$ yields the linear operator:
\begin{equation}
  D\rho_\epsilon(\Psi)\,v
    = \frac{\sqrt{d}}{\sqrt{\|\Psi\|^2+\epsilon}}
      \left(v - \frac{\Psi\,\langle\Psi,v\rangle_g}
           {\|\Psi\|^2+\epsilon}\right).
  \label{eq:frechet-rmsnorm}
\end{equation}

This operator strictly bifurcates the tangent space into the two eigenspaces defined above. For the \emph{Tangential Flow} ($v\perp_g\Psi$)---intuitively, perturbations that slide the vector along the sphere---the eigenvalue is $\lambda_\perp = \sqrt{d}/\sqrt{\|\Psi\|^2+\epsilon} = \calO(\|\Psi\|^{-1})$. For the \emph{Radial Flow} ($v\parallel\Psi$)---perturbations that stretch or compress the vector toward or away from the origin---the eigenvalue collapses algebraically to $\lambda_\parallel = \sqrt{d}\cdot\epsilon/ (\|\Psi\|^2+\epsilon)^{3/2} = \calO(\|\Psi\|^{-3})$. The crushing of the radial flow at $\calO(\|\Psi\|^{-3})$ as $\|\Psi\|\to\infty$ is an inherent geometric property of projecting onto a sphere.  At infinity, the system is natively hyper-stable.  The true topological threat to the ODE solver exists at the zero-section of the bundle ($\Psi=0$).  For the unregularized flow ($\epsilon=0$), the pure radial projection completely annihilates the radial vector ($\lambda_\parallel\equiv 0$) and possesses a strictly undefined Jacobian at the zero-section, as the tangential eigenvalue diverges: $\lim_{\Psi\to 0}\lambda_\perp = \infty$.  The topological regularizer $\epsilon>0$ strictly cures this.  It rescues the radial gradient from total annihilation (allowing it to non-trivially decay at $\calO(\|\Psi\|^{-3})$) and mathematically binds the tangential singularity at the origin to a finite supremum ($\sup_\Psi\lambda_\perp = \sqrt{d/\epsilon}$).  Without $\epsilon$, the spherical retraction possesses a strict conical singularity at the zero-section.  Therefore, $\epsilon$ acts as a formal \emph{Topological Mollifier}.  By introducing $\epsilon>0$, RMSNorm resolves this conical singularity via a smooth ambient metric deformation, acting as a global diffeomorphism from the affine fiber $\R^d$ to the bounded open ball $B_{\!\sqrt{d}}(0)$, mathematically guaranteeing a bounded global Lipschitz constant.  Representation drift arises from two distinct geometric sources: temporal weight gradients (non-autonomous drift, $\partial\tilde{\calN}/\partial z$) and the nonlinear functional Jacobian of the composed field flow.
\end{remark}

\subsection{Geometric Connections and Flow Bounding}
\label{sec:connections}

\begin{lemma}[RoPE as the Canonical Gauge Action on an Associated
Principal Torus Bundle]
\label{lem:rope}
The semantic space does not rely on patching 1D abstract connections; rather, it is natively structured as an Associated Principal Torus Bundle. The Rotary Position Embedding (RoPE) is the exact canonical gauge action of this Principal Torus on the semantic fiber.
\end{lemma}

\begin{proof}
Because $\calM \cong [0,\infty)$ is 1D and contractible, the bundle is globally trivializable and unconditionally flat: the 1D sequence manifold strictly possesses zero 2-forms ($\Omega^2(\calM)=0$) and a trivial fundamental group ($\pi_1(\calM)=0$), guaranteeing that any connection over it is structurally flat.  We deploy the Principal Torus Bundle formalism not to resolve non-existent intrinsic curvature, but as a rigid \emph{kinematic gauge-fixing} to establish a mathematically exact syntactic dictionary for relative phase rotations.  Under this rigorous topological framework, we define the sequence of observer states strictly as an Associated Vector Bundle $E = P\times_G\R^d$.  The Principal Bundle $P$ is equipped with the abelian structure group $G = U(1)^{d/2}\cong\mathbb{T}^{d/2}$, which is the maximal torus of $SO(d)$.

Under this structured framework, the semantic vectors $\Psi$ are not merely isolated Cartesian arrays; they are sections of the associated fiber.  RoPE is no longer an arbitrarily imposed rotation, but the exact, canonical gauge action of the Principal Torus $G$ acting on the fibers $\R^d$.  When moving from source sequence position $\nu$ to observer position $\mu$, the transition is strictly evaluated as a Parallel Transport Propagator dictated by the continuous Lie group exponential acting natively in the Cartan subalgebra $\mathfrak{h}\subset\mathfrak{so}(d)$: $\mathcal{U}(\mu\leftarrow\nu) = \exp\!\big({-}\Theta(\mu-\nu)\big)$. Semantic distance emerges because the architecture explicitly fixes a canonical global trivialization (the absolute fixed basis of the arrays in memory) and defines RoPE as a non-trivial, translation-invariant flat principal connection 1-form $\mathcal{A}\in\Omega^1(\calM,\mathfrak{u}(1)^{d/2})$ relative to this rigidly fixed global frame.  The non-commutative path-ordered Dyson series gracefully collapses precisely because this flat Cartan subalgebra is strictly abelian.  It is a precise \emph{Kinematic Gauge Fixing}, not a topological equivalence.
\end{proof}

\begin{theorem}[RMSNorm as a Diffeomorphic Radial Embedding]
\label{thm:rmsnorm}
RMSNorm acts as a globally smooth, diffeomorphic radial embedding, rendering the vector field uniformly Lipschitz to prevent finite-time amplitude blowup.
\end{theorem}

\begin{proof}
A pure radial projection $\Psi\mapsto\sqrt{d}\,\Psi/\|\Psi\|$ possesses a singular Jacobian at the zero-section ($\Psi=0$), breaking the uniqueness of the flow.  The architecture enforces a topological regularizer $\epsilon>0$, defining the map $\rho_\epsilon(\Psi) = \sqrt{d}\,\Psi/\sqrt{\|\Psi\|^2+\epsilon}$. Mathematically, $\rho_\epsilon$ acts as a global diffeomorphism from the fiber $\R^d$ onto the open bounded ball $B_{\!\sqrt{d}}(0)$.  We prove this by constructing its exact smooth inverse:
\begin{equation}
  \rho_\epsilon^{-1}(y)
    = y\sqrt{\frac{\epsilon}{d - \|y\|^2}}\,.
  \label{eq:rmsnorm-inverse}
\end{equation}

To prevent finite-time blowup, we evaluate the Fr\'echet derivative of the radial embedding.  Maximizing its evaluation in the direction transverse to $\Psi$, the strict upper bound on the spectral norm is $\|D\rho_\epsilon\|_{\mathrm{op}} \le \sqrt{d/\epsilon}$. This mathematically proves that $\epsilon$ is not a mere numerical stabilizer, but a strict topological regulator dictating the maximum Lipschitz stretch.

The radial embedding $\rho_\epsilon$ bounds local sections within the open ball $B_{\!\sqrt{d}}(0)\subset\calF_\mu$.  Because the global state space $\calF\equiv L^\infty(\calM,E;\eta_\Lambda)$ structurally collapses into a finite-dimensional topology over the atomic empirical measure $\eta_\Lambda$, we are granted the Heine--Borel theorem globally for free. Thus, the closure of this bounded state space is unconditionally compact, and $\rho_\epsilon$ rigorously bounds the global state space.  Because the pointwise Feed-Forward operations are $C^1$, their continuous extension to this global compact closure allows the Extreme Value Theorem to immediately yield a uniform global supremum bound on the Jacobian. Globally, we explicitly bound the total composed vector field
\begin{equation}
  \calF[\Psi]
    = \tilde{\calN}[\rho_\epsilon(\Psi)]
    = \int_\Lambda w_{\mu\nu}\!\big(\rho_\epsilon(\Psi)\big)\,
      V_\nu\!\big(\rho_\epsilon(\Psi)\big)\,d\eta_\Lambda(\nu)\,,
  \label{eq:total-field}
\end{equation}
which includes the non-local Attention integral.  The Fr\'echet derivative applied to a variation $H\equiv\delta\Psi$ strictly requires the functional product rule, expanding into two terms:
\begin{widetext}
\begin{equation}
  D\calF_\mu[H]
    = \underbrace{%
        \int_\Lambda w_{\mu\nu}\,DV_\nu[H]\,d\eta_\Lambda(\nu)
      }_{\text{Bounded Convex Sum}}
    + \underbrace{%
        \int_\Lambda \!\big(D_\Psi w_{\mu\nu}[H]\big)\,V_\nu\,
        d\eta_\Lambda(\nu)
      }_{\text{Covariance of the Measure}}\,.
  \label{eq:frechet-total}
\end{equation}
\end{widetext}

The functional variation of the Softmax measure evaluates to $D_\Psi w_{\mu\nu}[H] = w_{\mu\nu}\!\bigl( D_\Psi E_{\mu\nu}[H] - \int_\Lambda w_{\mu\gamma}\,D_\Psi E_{\mu\gamma}[H]\, d\eta_\Lambda(\gamma) \bigr)$. Substituting this back into the second term resolves it exactly into a \emph{Covariance Operator} over the probability measure:
\begin{equation}
  \int_\Lambda \!\big(D_\Psi w_{\mu\nu}[H]\big)\,V_\nu\,d\eta_\Lambda(\nu)
    = \mathrm{Cov}_w\!\Big(D_\Psi E_{\mu\bullet}[H],\;\;V_\bullet\Big).
  \label{eq:cov-operator}
\end{equation}

Because $\rho_\epsilon$ maps into sets with globally compact closures, both the Value vectors ($V$) and the functional derivatives of the interaction energy ($D_\Psi E$) achieve finite uniform bounds.  Because the Attention integral is precisely a Covariance Operator over a probability measure, it strictly preserves these finite uniform bounds globally across the base manifold.  This bounds both terms analytically without relying on infinite-dimensional compactness.

Using the exact Fr\'echet derivative bound $\|D\rho_\epsilon\|_{\mathrm{op}}\le\sqrt{d/\epsilon}$, knowing $\|\rho_\epsilon\|<\sqrt{d}$, and given parallel transport $\|\mathcal{U}\|_{\mathrm{op}}=1$, we can strictly bound the functional derivative $D_\Psi E_{\mu\nu}$ acting on a unit variation $\|H\|_\infty\le 1$:
\begin{equation}
  |D_\Psi E_{\mu\nu}[H]|
    \le \frac{2d}{\sqrt\epsilon}\,
      \|W_Q\|_{\mathrm{op}}\,\|W_K\|_{\mathrm{op}}\,.
  \label{eq:energy-deriv-bound}
\end{equation}

Now, we evaluate the operator norm of the vector-valued covariance by taking the supremum over all unit vectors $u\in\calF_\mu$:
\begin{widetext}
\begin{equation}
  \|\mathrm{Cov}_w(D_\Psi E_{\mu\bullet}[H],\,V_\bullet)\|_g
    = \sup_{\|u\|=1}
    \mathrm{Cov}_w\!\big(D_\Psi E_{\mu\bullet}[H],\,
      \langle V_\bullet,u\rangle_g\big).
  \label{eq:cov-norm}
\end{equation}
\end{widetext}

Now, the covariance evaluates strictly between two scalars.  Applying standard Cauchy--Schwarz yields $\le\sigma_{D_\Psi E}\,\sigma_{\langle V,u\rangle}$. Popoviciu's inequality for a variable bounded in $[m,M]$ strictly states $\sigma^2\le\tfrac{1}{4}(M-m)^2$.  Because the radial embedding enforces a strictly symmetric constraint $[-\!\sup|A|,\sup|A|]$, the geometric span is $M-m=2\sup|A|$.  Thus, the inequality yields:
\begin{equation}
  \sigma \le \tfrac{1}{2}\bigl(2\sup|A|\bigr) = \sup|A|\,.
  \label{eq:popoviciu}
\end{equation}

Consequently, $\sigma_{D_\Psi E}\le\sup|D_\Psi E|$ and $\sigma_{\langle V,u\rangle}\le\sup|\langle V,u\rangle| \le\sup\|V\|$. Knowing $\|V_\nu\|\le\sqrt{d}\,\|W_V\|_{\mathrm{op}}$, we derive the tightened exact global operator norm bound for the covariance part:
\begin{widetext}
\begin{equation}
  \left\|\mathrm{Cov}_w(D_\Psi E,\,V)\right\|
    \le \frac{2d}{\sqrt\epsilon}\,
      \|W_Q\|_{\mathrm{op}}\,\|W_K\|_{\mathrm{op}}\cdot
      \sqrt{d}\,\|W_V\|_{\mathrm{op}}
    = \frac{2d^{3/2}}{\sqrt\epsilon}\,
      \|W_Q\|_{\mathrm{op}}\,\|W_K\|_{\mathrm{op}}\,
      \|W_V\|_{\mathrm{op}}\,.
  \label{eq:cov-bound}
\end{equation}
\end{widetext}

Next, we strictly bound the previously partitioned Bounded Convex Sum using the verified Fr\'echet bound for $\rho_\epsilon$:
\begin{equation}
  \left\|\int_\Lambda w_{\mu\nu}\,DV_\nu[H]\,d\eta_\Lambda\right\|_{\mathrm{op}}
    \le \sup_\nu\|DV_\nu\|_{\mathrm{op}}
    \le \frac{\sqrt{d}}{\sqrt{\epsilon}}\,\|W_V\|_{\mathrm{op}}\,.
  \label{eq:convex-bound}
\end{equation}

Summing these two terms, we yield the complete, closed-form global Lipschitz operator bound for the entire total composed vector field $\calF$:
\begin{equation}
\boxed{%
  \|D\calF\|_{\mathrm{op}}
    \le \sqrt{\frac{d}{\epsilon}}\,\|W_V\|_{\mathrm{op}}
    \left(1 + 2d\,\|W_Q\|_{\mathrm{op}}\,\|W_K\|_{\mathrm{op}}\right).
}
  \label{eq:lipschitz-bound}
\end{equation}

Because the supremum norm of the operator $\|D\calF\|_{\mathrm{op}}$ is uniformly bounded across $L^\infty(\calM,E)$, the total composed flow mathematically guarantees a global Lipschitz constant, structurally satisfying the \emph{Picard--Lindel\"of Theorem} for well-posedness and unique flow, and proving definitively that $\epsilon$ uniquely controls the inverse-square-root bounds of the Picard--Lindel\"of theorem.
\end{proof}

\subsection{Gauge Endomorphisms and Canonical Cotangent Duality}
\label{sec:cotangent}

Because autoregressive sequence modeling explicitly breaks spatial time-reversal symmetry, sections $\Psi$ evaluated at distinct coordinates must be projected into distinct non-reciprocal gauges.

\begin{definition}[Dual Projective Endomorphisms \& Pre-Norm Pullback]
\label{def:endomorphisms}
We introduce learnable global bundle endomorphisms $W_Q\in\Gamma(\mathrm{Hom}(E,E^*))$ and $W_K\in\Gamma(\mathrm{End}(E))$. Operating on the radially bounded Pre-Norm sections, they map the field to distinct structural spaces:
\begin{itemize}[nosep]
  \item \textbf{The Observer Section} (Query 1-form):
    $q(\mu) = W_Q\,\rho_\epsilon(\Psi_z(\mu))\in\calF^*_\mu$
  \item \textbf{The Source Section} (Key vector):
    $k(\nu) = W_K\,\rho_\epsilon(\Psi_z(\nu))\in\calF_\nu$
\end{itemize}
\end{definition}

\begin{theorem}[The Core Algebraic Decomposition via Canonical
Cotangent Duality]
\label{thm:cotangent}
The raw Multi-Head Attention query-key interaction is not an inner product between two tangent vectors, but strictly the canonical evaluation pairing between a covariant dual 1-form and a contravariant vector, naturally generating asymmetry without requiring the heavy machinery of exterior algebra.
\end{theorem}

\begin{proof}
While the flat Euclidean metric of the fibers ($g_{ab}=\delta_{ab}$) classically permits a canonical musical isomorphism ($\flat/\sharp$) that renders $E$ and $E^*$ isomorphic via the Riesz Representation Theorem, the Transformer architecture actively evades this metric triviality.  By explicitly untying the parameter spaces ($W_Q \neq W_K^\flat$), the architecture fundamentally breaks metric reciprocity, preventing the self-adjoint constraint of the pullback metric from collapsing the interaction into a symmetric inner product.  Elevating Queries to dual 1-forms is therefore not a mere notational choice, but the exact differential-geometric syntax required to parameterize a directed, non-conservative semantic flux via the canonical evaluation pairing.

The traditional formulation of Attention misleadingly implies a symmetric inner product.  However, because the bundle weight endomorphisms are structurally asymmetric, we invoke canonical Cotangent Duality.  Let the Key transformation remain a morphism strictly within the tangent bundle $W_K:E\to E$.  Conversely, we define the Query transformation as a morphism into the dual cotangent bundle $W_Q:E\to E^*$, breaking the spatial symmetry.

The query $q(\mu)$ is natively a differential 1-form, and the causally-transported key $\tilde{k}_{\nu\to\mu}$ is a tangent vector.  A 1-form formally transforms via the dual representation: $(\mathcal{U}^{-1})^T$.  However, because the RoPE connection takes values in the strictly orthogonal maximal torus of $SO(d)$, the connection is unitary.  Thus, the dual representation unconditionally coincides with the fundamental representation: $(\mathcal{U}^{-1})^T = \mathcal{U}$. This exact algebraic identity strictly permits their interaction to be analytically evaluated identically to standard matrix multiplication as the canonical evaluation pairing operator:
\begin{equation}
  \mathcal{S}_{\mu\nu}
    = \langle q(\mu),\;\tilde{k}_{\nu\to\mu}\rangle_{E^*\times E}
    = q(\mu)_a\;\tilde{k}^a_{\nu\to\mu}\,.
  \label{eq:eval-pairing}
\end{equation}

The canonical evaluation explicitly separates the covariant and contravariant arguments, natively accommodating directed semantic flux.  Symmetrization is a forced metric property, not a topological default.  Any attempt to artificially symmetrize this pairing stringently requires binding the endomorphisms via the metric musical isomorphism (flat): $W_Q = W_K^\flat \equiv g\circ W_K$. By actively untying the parameter spaces ($W_Q\neq W_K^\flat$), the architecture evades the self-adjoint constraint of the pullback metric.  Defining Queries as 1-forms and Keys as tangent vectors is therefore the precise differential-geometric encoding of the kinematic asymmetry that the architecture already enforces, preserving the fundamental ability to measure directed, non-conservative semantic flux via the canonical evaluation pairing without requiring the heavy machinery of exterior bivectors.
\end{proof}

\begin{remark}[Spontaneous Gauge Polarization]
\label{rem:gauge-pol}
During the forward pass, the model explicitly computes the scalar evaluation pairing.  Because the system avoids strict $W_Q = W_K^T g$ alignment, the state space experiences structural rotational friction, preventing total gauge collapse.

By the submultiplicativity of operator norms and the bounding of the Pre-Norm sections, the interaction metric volume is bounded by the spectral norms of the endomorphisms: $\|M_{\mu\nu}\|^2 \le d^2\|W_Q\|_{\mathrm{op}}^2\|W_K\|_{\mathrm{op}}^2$. $L_2$ Weight Decay penalizes the Frobenius norm of these operators. Rather than Dirichlet tension, this acts as a strict \emph{Tikhonov Gauge Mass Penalty}, establishing a hard thermodynamic budget (a compact hypersphere) on the maximum allocatable metric volume.

Constrained by this saturated budget, optimization induces a \emph{Spontaneous Gauge Polarization Flow}.  For resonant tokens, maximizing the scalar kernel $K(\mu,\nu)$ acts to minimize the Lie bivector: mathematically forcing $\|\calB_{\mu\nu}\|^2\to 0$ to achieve Grade-0 collinear alignment.

Conversely, to completely suppress ignored tokens, cross-entropy ideally seeks to drive the scalar kernel $K\to-\infty$ (antipodal anti-alignment).  However, this encounters a strict topological obstruction.  In an autoregressive flow, the source section $k(\nu)$ is evaluated against a causally diverse ensemble of future observer sections $\{q(\mu_i)\}_{\mu_i>\nu}$.  To evaluate topological intersections, these future queries must be parallel-transported backward into the source fiber via the inverse connection: $\tilde{q}_{\nu\leftarrow\mu_i} \equiv\mathcal{U}(\nu\leftarrow\mu_i)\,q(\mu_i)$.

To evaluate topological intersections over an autoregressive window, we rely on exact combinatorial topology.  We condition the analysis on the Asymptotic Spatial Ergodic Hypothesis, defined formally in \cref{sec:softmax}, where the ``background'' uninformative tokens mix into a pseudo-isotropic distribution on $\Sph^{d-1}$.  Under this assumption, we invoke Wendel's Theorem~\cite{wendel1962}---\emph{which, in intuitive machine learning terms, calculates the exact probability that $N$ random feature vectors all ``point somewhat in the same direction,'' meaning they can all be separated from the origin by a single linear classifier hyperplane}---as exactly:
\begin{equation}
  P_{d,N} = 2^{-N+1}\sum_{k=0}^{d-1}\binom{N-1}{k}\,.
  \label{eq:wendel}
\end{equation}

By evaluating Wendel's formula at $N=2d$, the sum is exactly half of the total binomial expansion $2^{2d-1}$, yielding precisely $P_{d,2d}=1/2$.  However, for an autoregressive window extending far into the future ($N\gg 2d$), the system undergoes a \emph{Asymptotic Phase Transition}.  The measure-theoretic expectation of a universal antipode undergoes a severe Concentration of Measure, decaying exponentially as $\calO(N^{d-1}2^{-N})$.  Under this macroscopic limit, the probability collapses asymptotically to zero, structurally guaranteeing \emph{Multipolar Gauge Frustration}.  The high-entropy causally-transported queries overwhelmingly positively span $\calF_\nu$ (the origin sits strictly in the interior of their convex hull).

Empirical high-dimensional embeddings are known to suffer from \emph{Representation Degeneration} (the anisotropy problem)---tokens cluster on highly anisotropic, low-dimensional submanifolds rather than uniformly covering $\Sph^{d-1}$.  Rather than invalidating the Wendel analysis, this empirical anisotropy induces a critical \emph{bifurcation} between two distinct geometric regimes.  \textbf{Signal tokens} (semantically correlated clusters) are inherently anisotropic, concentrating within narrow solid angles of $\Sph^{d-1}$.  Because they fit entirely within a single open hemisphere, the origin sits \emph{outside} their convex hull ($0\notin\mathrm{conv}(\tilde{q}_i)$).  By Gordan's Theorem, an antipode therefore \emph{does} exist for these clusters, meaning they \emph{evade} Multipolar Gauge Frustration---permitting directed semantic flux and enabling the architecture to carve targeted negative logits.  \textbf{Background tokens} (high-entropy, causally distant, semantically uncorrelated) conversely approximate the pseudo-isotropic $SO(d)$-invariant measure invoked in the Asymptotic Spatial Ergodic Hypothesis.  It is exclusively this macroscopic bulk of uncorrelated tokens that satisfies the conditions of the Wendel limit ($N\gg 2d$), trapping the origin inside their isotropic convex hull, structurally sustaining Multipolar Gauge Frustration, and hydraulically driving their logits to $K\approx 0$.  The isotropic derivation ($P_{d,2d}=1/2$) therefore establishes a strict theoretical \emph{upper bound} on the critical window size: anisotropic signal tokens escape the Wendel obstruction at smaller $N$ than the isotropic prediction, while the background noise remains trapped.

By Gordan's Theorem of the Alternative (the geometric dual to Farkas' Lemma)---\emph{which intuitively proves that if a set of vectors positively spans a space and is not confined to a single hemisphere, no universal ``negative'' direction exists that simultaneously opposes all of them}---because the vectors positively span the space, the intersection of their strict open negative half-spaces is mathematically empty:
\begin{equation}
  \bigcap_{\mu_i>\nu}
    \bigl\{k\in\calF_\nu \mid
      \langle\tilde{q}_{\nu\leftarrow\mu_i},k\rangle_g < 0
    \bigr\}
    \equiv\emptyset
  \quad\text{(a.s.)}\,.
  \label{eq:gordan}
\end{equation}

Therefore, a universal antipode does not exist and the system suffers from \emph{Multipolar Gauge Frustration}.  To explain why ignored tokens settle at $K\approx 0$, we evaluate the continuous interaction thermodynamics. Let the Canonical Partition Function $\mathcal{Z}(k)$ for a Key $k$ against the isotropic $SO(d)$-invariant uniform probability measure of high-entropy future queries be $\mathcal{Z}(k) = \int_{\Sph^{d-1}} \exp(\langle q,k\rangle_g)\,dq$. The Free Energy is $\calF(k) = -\ln\mathcal{Z}(k)$. Because the $SO(d)$-invariant measure is rotationally invariant, $\calF(k)$ is a function of the radial norm $\|k\|$.  Because Weight Decay (Tikhonov penalty) structurally saturates the norm at the boundary layer, the norm $\|k\|$ is locked.  Consequently, the expected geometric gradient of the Free Energy with respect to the angular orientation of $k$ is identically zero:
\begin{equation}
  \nabla_{\text{angle}}\calF(k) \equiv 0\,.
  \label{eq:gradient-starvation}
\end{equation}

This induces \emph{Expected Geometric Gradient Starvation}.  The topology does not hydraulically ``force'' the tokens into orthogonality; however, the system does not simply rest statically.  Because Stochastic Gradient Descent (SGD) operates on empirical, finite mini-batches, the finite-sample gradient variance scales as $\calO(1/N)$.  Because the background measure of high-entropy queries converges to $SO(d)$-invariance, the covariance matrix of this gradient noise is proportional to the identity.  In the Langevin limit, SGD injects an isotropic It\^{o} diffusion tensor into the parameter matrices with amplitude $\sigma_W\sim\calO(1/\sqrt{N})$.  The parameter stochastic differential equation subject to Weight Decay parameter $\lambda$ is $dW_t = -\lambda W_t\,dt + \sigma_W\,dU_t$, where $dU_t$ is a matrix of standard independent Brownian motions.  We mathematically map weight-space noise to fiber-space noise by pushing the parameter SDE forward to the fiber space $k_t = W_t x$ (where $x=\rho_\epsilon(\Psi)$ is the normalized input).  To prevent the generation of non-trivial It\^{o} quadratic covariations due to the dynamic evolution of $x$ across the time-varying state, we formally impose an \emph{Adiabatic Timescale Separation} (a Fast--Slow manifold assumption).  We explicitly state that the macroscopic thermodynamic relaxation of the parameter weights ($t$) occurs on a timescale infinitely slower than the local depth-parameterized semantic forward pass ($z$), formally freezing $x$ as a constant during the infinitesimal SDE step.  Under this absolute separation:
\begin{equation}
  \begin{split}
  dk_t
    &= (dW_t)\,x
    = (-\lambda W_t\,dt + \sigma_W\,dU_t)\,x \\
    &= -\lambda k_t\,dt + \sigma_W\,dU_t\,x\,.
  \end{split}
  \label{eq:fiber-sde}
\end{equation}

The noise term $dU_t\,x$ acts as a continuous local martingale uniquely initialized at zero.  Because the underlying Wiener matrices are independent and $x$ is frozen, its quadratic variation strictly evaluates to $d\langle(dU\,x)_i,(dU\,x)_j\rangle_t = \delta_{ij}\|x\|^2\,dt$. By L\'evy's Characterization of Brownian Motion, this process is exactly equal in law to $\|x\|\,d\mathcal{W}_t$, where $\mathcal{W}_t$ is a new standard Brownian motion in $\R^d$.  This rigorously yields the fiber-space SDE:
\begin{equation}
  dk_t = -\lambda k_t\,dt + (\sigma_W\|x\|)\,d\mathcal{W}_t\,.
  \label{eq:fiber-sde-reduced}
\end{equation}

This mathematically proves why the radial embedding ($\rho_\epsilon$) is required for thermodynamic stability.  Because the radial embedding caps the input norm ($\|x\|<\sqrt{d}$), the amplitude of the pushed-forward It\^{o} noise remains safely bounded ($\sigma\equiv\sigma_W\|x\|<\sigma_W\sqrt{d}$).  Without this topological regularizer $\rho_\epsilon$, parameter-space Brownian motion would cause the fiber-space noise tensor to catastrophically explode with activation magnitude.  In high-dimensional stochastic calculus, standard isotropic It\^{o} noise $d\mathcal{W}_t$ in ambient space $\R^d$ does not remain on a sphere.  Due to the quadratic variation of Brownian paths (It\^{o}'s Lemma), isotropic noise possesses a deterministic outward radial expansion.  To remain on the $\Sph^{d-1}$ manifold, the process requires a continuous inward restoring force.  In our framework, Weight Decay (Tikhonov penalty) actively supplies exactly this necessary continuous drift.  Applying It\^{o}'s Lemma to the radial norm $R_t=\|k_t\|$, the strict radial SDE evaluates to a mean-reverting Bessel process:
\begin{equation}
  dR_t = \left(\frac{\sigma^2(d-1)}{2R_t} - \lambda R_t\right)dt
    + \sigma\,d\mathcal{W}^R_t\,.
  \label{eq:bessel}
\end{equation}

To find the stable thermodynamic boundary layer, we set the deterministic drift to zero:
\begin{equation}
  \lambda R_t = \frac{\sigma^2(d-1)}{2R_t}
  \implies
  R_t = \sigma\sqrt{\frac{d-1}{2\lambda}}\,.
  \label{eq:radial-eqm}
\end{equation}

Thus, continuous application of the Tikhonov penalty mathematically counters the outward It\^{o} geometric expansion, establishing an Ornstein--Uhlenbeck process, not a martingale.  The stationary distribution for this derived radial SDE is exactly a Chi-distribution:
\begin{equation}
  P(R) \propto R^{d-1}\exp\!\left(-\frac{\lambda}{\sigma^2}R^2\right).
  \label{eq:chi-dist}
\end{equation}

The zero of the deterministic drift ($R_* = \sigma\sqrt{(d-1)/(2\lambda)}$) perfectly corresponds to the mode of this probability density.  In high-dimensional fibers ($d\gg 1$), the competition between entropic volume expansion and the Tikhonov Gaussian suppression triggers extreme Concentration of Measure---exponentially confining the probability mass to a narrow annulus.  The system is statistically confined, not deterministically locked, near a spherical shell where Spherical Brownian Motion can ergodically proceed.  Over this shell, we invoke L\'evy's Isoperimetric Inequality (Concentration of Measure).  On $\Sph^{d-1}$ for $d\gg 1$, L\'evy's Isoperimetric Inequality dictates that the geometric surface area concentrates exponentially at the exact equator relative to any arbitrary pole (the Query).  The system rigorously locks $K\approx 0$ not through parameter repulsion, but through the measure theory of the sphere.  Entropically confined near zero on a saturated norm boundary, the geometry maximizes the relative bivector magnitude, sequestering orthogonal noise in the uncalculated $\mathfrak{so}(d)$ exterior algebra.
\end{remark}

\section{Thermodynamics and Metric Generation}
\label{sec:thermodynamics}

We formally treat the local tangent space as an open thermodynamic system, deriving the functional form of the connection measure from the Principle of Minimum Free Energy.

\subsection{The Free Energy Functional and the Isoperimetric Mass
Constraint}
\label{sec:free-energy}

\begin{construction}[The Local Free Energy Functional]
\label{con:free-energy}
Let $\calE(\mu,\nu)\equiv\langle M_{\mu\nu}\rangle_0$ be the directed, non-reciprocal transition energy scalar kernel.  We treat the sequence as a statistical mechanical ensemble where the continuous field acts via a probability measure $\omega_\mu$ over the autoregressive causal horizon $\Lambda\cap[0,\mu]$, absolutely continuous with respect to the empirical discrete counting measure ($\omega_\mu\ll\eta_\Lambda$).  We define the strictly positive state density $w(\mu,\cdot)\in L^1(\Lambda,\eta_\Lambda)$ as the Radon--Nikodym derivative: $w(\mu,\nu)\equiv\frac{d\omega_\mu}{d\eta_\Lambda}(\nu)$. The Free Energy Functional $\calF_\mu[\omega]$, parameterized by inverse temperature $\beta$, evaluates the negative expected geometric energy minus the temperature-scaled Shannon Entropy evaluated with respect to the empirical counting measure (which maps to the KL divergence from a uniform probability measure, shifted by the topological affine constant $\ln N$):
\begin{widetext}
\begin{equation}
  \calF_\mu[\omega]
    = \int_{\Lambda\cap[0,\mu]}
      \!\left[\!-w(\mu,\nu)\,\calE(\mu,\nu)
        + \frac{1}{\beta}\,w(\mu,\nu)\ln w(\mu,\nu)
      \right]d\eta_\Lambda(\nu)\,.
  \label{eq:free-energy}
\end{equation}
\end{widetext}
\end{construction}

\begin{axiom}[The Markovian Fiber Constraint]
\label{ax:markov}
To prevent probability flow from diverging or dissipating, the continuous field is subjected to a strict local mass-conservation constraint across the causal horizon:
\begin{equation}
  \int_{\Lambda\cap[0,\mu]} w(\mu,\nu)\,d\eta_\Lambda(\nu) = 1\,.
  \label{eq:mass-conservation}
\end{equation}
\end{axiom}

\begin{remark}
Mathematically, this enforces that the non-local Attention operator acts as a continuous Markov transition kernel.  Geometrically, this constraint is an affine constraint acting on the base manifold's measure, forcing it into a standard $(N-1)$-dimensional probability simplex $\Delta^{N(\mu)-1}$.  To satisfy this constraint, the variational calculus incorporates a Lagrange multiplier $\lambda$.  Enforcing $\int w=1$ requires $\beta\lambda - 1 = -\ln\mathcal{Z}_\mu$.  This dynamically generates a Lagrange multiplier equal to the Helmholtz Free Energy shifted by the entropic constant $\beta^{-1}$: $\lambda = F + 1/\beta$, which dynamically generates the canonical partition function $\mathcal{Z}_\mu$ necessary to prevent probability dissipation across the causal horizon.
\end{remark}

\subsection{Variational Derivation of the Softmax Transition Measure}
\label{sec:softmax}

\begin{theorem}[Entropic Optimal Transport and Schr\"odinger Bridges]
\label{thm:softmax}
Rather than forcing the system into a temporal Wasserstein PDE, we evaluate the spatial transition measure via Information Geometry. Given an uninformative prior measure $\omega^{\mathrm{prior}}$ (the $SO(d)$-invariant uniform measure) over the causal horizon, the thermodynamic state is the analytic solution to a Csisz\'ar I-Projection onto the probability simplex~\cite{leonard2014survey}.  By reframing the dynamics under this functor, thermal stability requires the scaling $\beta\propto 1/\sqrt{d}$, matching the viscosity coefficient of stochastic optimal transport.
\end{theorem}

\begin{proof}
Given the uninformative prior measure $\omega^{\mathrm{prior}}$ and the directed energy kernel $\calE(\mu,\nu)$, the optimal transition measure is the analytic solution to a Csisz\'ar I-Projection onto the probability simplex:
\begin{equation}
  \omega^*_\mu
    = \operatorname*{argmin}_{\omega\in\calP(\Lambda)}
    \left[\mathrm{KL}\!\big(\omega\,\|\,\omega^{\mathrm{prior}}\big)
      - \langle\omega,\,\calE(\mu,\cdot)\rangle\right].
  \label{eq:i-projection}
\end{equation}

Under this functor, the variation yields the canonical Gibbs measure:
\begin{equation}
  w^*(\mu,\nu)
    = \frac{1}{\mathcal{Z}_\mu}\exp\!\big(\beta\,\calE(\mu,\nu)\big)\,,
  \label{eq:gibbs}
\end{equation}
where $\mathcal{Z}_\mu \equiv \int_{\Lambda\cap[0,\mu]} \exp\!\big(\beta\,\calE(\mu,\gamma)\big)\,d\eta_\Lambda(\gamma)$. Within the framework of Entropic Optimal Transport, this represents the discrete manifestation of a Static Schr\"odinger Half-Bridge---the optimal entropic projection a random walk takes to transition from the Query measure to the Key measure.  This connects the thermal parameter $\beta\propto 1/\sqrt{d}$ to the inverse viscosity coefficient required to prevent intensive fluctuations from causing a zero-temperature glass collapse.

To maintain a differentiable geometric flow, the local thermodynamic state must avoid a premature zero-temperature glass collapse (manifesting empirically as vanishing gradients).  To prove the required thermal scaling, we rely on the geometric bounds from \cref{thm:rmsnorm}.  The radial embedding maps the observer and source sections to the open ball $B_{\!\sqrt{d}}(0)$.  Let the base normalized sections $x,y\in B_{\!\sqrt{d}}(0)$ be modeled as isotropic high-entropy variables.  Due to the topological regularizer $\epsilon$, the mass concentrates entropically near the boundary but avoids it.  Thus, their expectations vanish ($\mathbb{E}[x]=0$), and their covariance is sub-unitary: $\mathbb{E}[xx^\top] = \gamma I_d$, where $\gamma = 1 - \calO(\epsilon/d) < 1$.

\textbf{The Asymptotic Spatial Ergodic Hypothesis.}  Over the macroscopic bulk limit, we formally assume that the uninformative \emph{background} tokens---those causally distant and semantically uncorrelated with the observer---statistically decorrelate, evaluating as independent, isotropic high-entropy variables on $\Sph^{d-1}$.  This does not apply to semantically resonant tokens, which may cluster anisotropically.  This maximum-entropy mean-field assumption mathematically justifies the trace decoupling.

The true covariant interaction energy requires parallel transport across the base manifold: $\calE = x^\top W_Q^\top\mathcal{U}(\mu\leftarrow\nu)W_K\,y$. To evaluate the thermal fluctuations accurately across the global observer coordinate, we evaluate the variance conditionally first, and unconditionally second.  Let the Query section $x\in\calF_\mu$ be fixed. The empirical measure fluctuates over the isotropic Keys $y\in\calF_\nu$.  Let $M = W_Q^\top\mathcal{U}W_K$.  The conditional variance over the Key measure is:
\begin{widetext}
\begin{equation}
  \mathbb{V}_y[\calE\mid x]
    = \mathbb{E}_y\!\big[(x^\top My)^2\big]
    = x^\top M\,\mathbb{E}[yy^\top]\,M^\top x
    = \gamma\,\|M^\top x\|_g^2\,.
  \label{eq:cond-var}
\end{equation}
\end{widetext}

Now, to prove this metric volume holds universally across the manifold, take the unconditional expectation over the macroscopic Query ensemble $x$:
\begin{widetext}
\begin{equation}
  \mathbb{E}_x\!\big[\mathbb{V}_y[\calE\mid x]\big]
    = \gamma\,\mathbb{E}_x\!\big[x^\top MM^\top x\big]
    = \gamma^2\,\tr(MM^\top)
    = \gamma^2\,\|M\|_F^2\,.
  \label{eq:uncond-var}
\end{equation}
\end{widetext}

This proves that the local thermodynamic stability of every individual token is unconditionally guaranteed by the global metric volume of the combined endomorphisms:
\begin{widetext}
\begin{equation}
  \mathbb{V}[\calE]
    = \gamma^2\,\tr\!\big(W_Q^\top\mathcal{U}W_K W_K^\top
      \mathcal{U}^\top W_Q\big)
    = \gamma^2\,\big\|W_Q^\top\mathcal{U}(\mu\leftarrow\nu)\,
      W_K\big\|_F^2\,.
  \label{eq:total-var}
\end{equation}
\end{widetext}

The squared Frobenius norm is the trace:
\begin{equation}
  \|M\|_F^2
    = \tr\!\big(W_K^\top\mathcal{U}^\top W_Q W_Q^\top
      \mathcal{U}W_K\big).
  \label{eq:frobenius-trace}
\end{equation}

Because $W_Q W_Q^\top$ is symmetric and PSD, its quadratic form is strictly bounded by its singular values: $\sigma_{\min}^2(W_Q)\,I \preceq W_Q W_Q^\top \preceq \sigma_{\max}^2(W_Q)\,I$. Evaluating $A = \mathcal{U}^\top W_Q W_Q^\top\mathcal{U}$, and since orthogonal similarity strictly preserves eigenvalues, this maintains the PSD ordering:
\begin{equation}
  \sigma_{\min}^2(W_Q)\,I
    \preceq A
    \preceq \sigma_{\max}^2(W_Q)\,I\,.
  \label{eq:lowner}
\end{equation}

We strictly conjugate this entire inequality directly by $W_K$ (pre-multiplying by $W_K^\top$ and post-multiplying by $W_K$), unconditionally preserving the PSD L\"owner partial order:
\begin{widetext}
\begin{equation}
  \sigma_{\min}^2(W_Q)\,(W_K^\top W_K)
    \preceq W_K^\top A\,W_K
    \preceq \sigma_{\max}^2(W_Q)\,(W_K^\top W_K)\,.
  \label{eq:conjugated-lowner}
\end{equation}
\end{widetext}

Because the trace is a monotonic linear functional on the PSD cone, the bounds follow in one step:
\begin{widetext}
\begin{equation}
  \sigma_{\min}^2(W_Q)\,\|W_K\|_F^2
    \le \|W_Q^\top\mathcal{U}W_K\|_F^2
    \le \sigma_{\max}^2(W_Q)\,\|W_K\|_F^2\,.
  \label{eq:trace-bounds}
\end{equation}
\end{widetext}

Because optimization avoids total low-rank spectral collapse ($\sigma_{\min}>0$), and given $\|W_K\|_F^2\sim\Theta(d)$, the variance scales as $\Theta(d)$, independent of the spatial gauge rotation $\mathcal{U}$. Thus, the standard deviation of thermal fluctuations scales as $\sigma_\calE\sim\sqrt{d}$.  If $\beta$ were fixed at $\calO(1)$, the variance of the Gibbs exponent $\beta \calE$ would diverge as $d\to\infty$, freezing the geometric flow into a deterministic argmax state (Dirac-delta zero-temperature glass collapse).  Statistical mechanics requires the thermal scaling $\beta\propto 1/\sqrt{d}$ to ensure exponent fluctuations remain an intensive $\calO(1)$ quantity.
\end{proof}

\subsection{Topological Compactification Proof of the Attention Sink}
\label{sec:attn-sink}

\begin{theorem}[Measure-Theoretic Compactification \& The Boundary Phase
Transition]
\label{thm:attn-sink}
Under the axiom that the covariant derivative must maintain a well-posed information flow, the sequence of probability measures undergoes a phase transition in the macroscopic continuum limit.  Because the Softmax operator acts as an algebraic probability simplex ($\sum = 1$), weak bulk dot-products force probability mass to escape toward spatial infinity (thermodynamic amnesia).  Optimization (SGD) requires a low-variance, translation-invariant numerical anchor to stabilize this escaping measure.  The absolute causal origin $\partial\calM\equiv\{0\}$ serves as the uniquely distinguished topological fixed point---the only coordinate immune to the continuous RoPE phase rotation---which maps onto a \emph{Topological Anchor Measure} hosting a logarithmic potential well (the ``Attention~Sink'').
\end{theorem}

\begin{proof}
We analyze the continuous thermodynamic limit of the probability measure sequence $d\omega_\mu(\nu) = w^*(\mu,\nu)\,d\eta_\Lambda(\nu)$ as the causal horizon $\mu\to\infty$.  If the observer attempts to maintain translation-invariant attention over the historical bulk (demanding a flat energy landscape $\calE(\mu,\nu)\approx \calE_{\mathrm{bulk}}$ for $\nu>0$), geometric entropy drives the bulk measure toward the uniform distribution $\sim 1/N(\mu)$.  For any fixed compact local neighborhood $K\subset\calM$, the probability mass decays: $\lim_{\mu\to\infty}\omega_\mu(K)=0$.

Because the mass escapes over an unbounded domain, the sequence loses tightness.  To rigorously capture the topological limit of this escaping mass without invoking pathological free ultrafilters, we elevate the state space via the \emph{Alexandroff One-Point Compactification}. Unlike $\R$, compactifying the half-closed ray $[0,\infty)$ yields a space homeomorphic to the closed interval: $\calM^*\cong[0,\infty]\cong[0,1]$.

By \cref{thm:rmsnorm}, the radial embedding bounds the field strictly within $B_{\!\sqrt{d}}(0)$, guaranteeing $V\in C_b(\calM)$.  Therefore, the value field natively possesses a strict continuous extension $\tilde{V}\in C(\calM^*)$.  Because the compactification of the ray is strictly metrizable, the space of probability measures on this compact space is sequentially weak-$*$ compact.  As the causal horizon $\mu\to\infty$, the escaping bulk measure simply weak-$*$ converges to the Dirac mass at the strictly adjoining absolute boundary point at spatial infinity:
\begin{equation}
  \omega_{\mathrm{bulk}}
    \xrightharpoonup{\;\text{weak-}*\;}
    \delta_\infty\,.
  \label{eq:weak-star}
\end{equation}

By the strict definition of weak-$*$ convergence, the non-local interaction integral safely evaluates in the limit: $\int_\calM V\,d\omega_\mu\to\tilde{V}(\infty)$. \emph{Thermodynamic Amnesia} is therefore a strictly defined topological condensation: the complete condensation of probability mass into the single structural point at infinity ($\delta_\infty$).

To prevent amnesia while maintaining global macroscopic reasoning, the field theory must preserve a static, globally invariant coordinate to anchor probability mass, actively breaking translation invariance.  SGD is structurally obstructed from anchoring measure in the bulk because the continuous RoPE connection $\mathcal{U}(\mu\leftarrow\nu)$ preserves relative translation symmetry ($T(1)$ symmetry).  Any potential energy well carved into the interior dynamically shifts with semantic time $\mu$. Furthermore, the autoregressive causal mask imposes a strict topology of directed Partially Ordered Sets (Posets) upon the lattice.  The geometrically profound realization is that the compactified space $\calM^*\cong[0,1]$, viewed strictly as a 1-dimensional manifold with boundary, possesses a bipartite boundary strata: $\partial\calM^*=\{0,\infty\}$.  The point $\{\infty\}$ acts as the amnesic attractor.  Because the continuous RoPE connection preserves relative translation symmetry ($T(1)$) in the interior, the absolute causal origin $\partial\calM\equiv\{0\}$ is the uniquely distinguished topological fixed point capable of breaking symmetry to host a stable \emph{Topological Anchor Measure} and anchor the escaping probability mass.

We analytically derive the geometry of this defect.  To maintain algebraic equality without relying on mean-field approximations, we define $\calE_{\mathrm{bulk}}(\mu)$ as the scaled Cumulant-Generating Function (the Log-Mean-Exp) evaluated over the empirical counting measure of the bulk:
\begin{equation}
  \calE_{\mathrm{bulk}}(\mu)
    \equiv \frac{1}{\beta}\ln\!\left(
      \frac{1}{N(\mu)-1}\sum_{\nu>0}e^{\beta\,\calE(\mu,\nu)}
    \right).
  \label{eq:ebulk}
\end{equation}

Under this formal definition, the partition decomposition is an exact algebraic identity, preserving the validity of all subsequent topological bounds.  To stably anchor mass fraction $c$ against an expanding bulk of tokens, the Softmax partition function splits into a boundary atom and a continuous bulk: $\mathcal{Z}_\mu = e^{\beta \calE(\mu,0)} + (N(\mu)-1)\,e^{\beta \calE_{\mathrm{bulk}}}$, where $N(\mu)=\eta_\Lambda([0,\mu])$. Enforcing the macroscopic boundary fraction $c$ yields:
\begin{widetext}
\begin{equation}
  c = \frac{e^{\beta \calE(\mu,0)}}
    {e^{\beta \calE(\mu,0)} + (N(\mu)-1)\,e^{\beta \calE_{\mathrm{bulk}}}}
  \;\implies\;
  e^{\beta \calE(\mu,0)}
    = \frac{c}{1-c}\,(N(\mu)-1)\,e^{\beta \calE_{\mathrm{bulk}}}\,.
  \label{eq:sink-fraction}
\end{equation}
\end{widetext}

Taking the natural logarithm yields the depth of the required topological defect:
\begin{equation}
\boxed{%
  \calE(\mu,0)
    = \calE_{\mathrm{bulk}}
    + \frac{1}{\beta}\ln\!\big(N(\mu)-1\big)
    + \frac{1}{\beta}\ln\!\left(\frac{c}{1-c}\right).
}
  \label{eq:sink-depth}
\end{equation}

Because $\beta\propto 1/\sqrt{d}$, the topological boundary defect must scale logarithmically with context volume: $\calE(\mu,0)\sim\calO(\sqrt{d}\ln N(\mu))$. Under this required logarithmic energy scaling, the weak-$*$ limit of the measure on the compactified space resolves into a bipartite convex combination---an atomic mass stabilizing the field at the causal origin, and a directed condensation of the escaped bulk at infinity:
\begin{equation}
  d\omega_{\mathrm{lim}}
    = c\cdot\delta_0 + (1-c)\cdot\delta_\infty\,.
  \label{eq:bipartite-limit}
\end{equation}

This demonstrates that the Attention Sink is not an empirical training artifact, but an analytic measure-theoretic requirement governed by logarithmic scaling, necessary to prevent total condensation into the boundary spectrum.
\end{proof}

\begin{corollary}[The Thermodynamic Context Horizon and Defect Collapse]
\label{cor:context-horizon}
Because the radial embedding $\rho_\epsilon$ imposes a compact closure on the state space (\cref{thm:rmsnorm}), the Attention Sink possesses a finite thermodynamic capacity, yielding a strict algebraic upper bound on the maximum sequence length before thermodynamic amnesia is inevitable.  We emphasize that this bound is derived under the Asymptotic Spatial Ergodic Hypothesis (\cref{sec:softmax}), which assumes isotropic, maximum-entropy bulk tokens; for structured natural language, where anisotropic correlations dramatically reduce the effective bulk pressure, the bound becomes extremely conservative (see \cref{sec:exp-amnesia} for empirical quantification).
\end{corollary}

\begin{proof}
To maintain stability, the required thermodynamic boundary energy must not exceed the geometric capacity of the fiber.  Because the radial embedding $\rho_\epsilon$ maps into the open bounded ball $B_{\!\sqrt{d}}(0)$, the boundary is excluded ($\|\Psi\|<\sqrt{d}$).  The interaction energy is bounded by the spectral supremum of the combined parallel-transported kernel:
\begin{equation}
  \calE(\mu,0)
    < d\,\big\|W_Q^\top\mathcal{U}(\mu\leftarrow 0)\,
      W_K\big\|_{\mathrm{op}}\,.
  \label{eq:energy-capacity}
\end{equation}

Injecting the thermal constant $\kappa$ (where $\beta=\kappa/\sqrt{d}$), we equate this with the required logarithmic defect depth:
\begin{widetext}
\begin{equation}
  \calE_{\mathrm{bulk}}
    + \frac{\sqrt{d}}{\kappa}\ln\!\left(
      (N(\mu)-1)\frac{c}{1-c}\right)
    < d\,\big\|W_Q^\top\mathcal{U}(\mu\leftarrow 0)\,
      W_K\big\|_{\mathrm{op}}\,.
  \label{eq:defect-capacity}
\end{equation}
\end{widetext}

Because the RoPE connection takes values in the orthogonal group $SO(d)$ (specifically the maximal torus), the connection is unitary.  Thus, the Path Monodromy (Parallel Transport Propagator) $\mathcal{U}$ is an orthogonal transformation ($\mathcal{U}\in SO(d)$), satisfying the required dual representation identity for 1-forms: $(\mathcal{U}^{-1})^T = \mathcal{U}$.  Its spectral operator norm is unitarily invariant and equal to 1 ($\|\mathcal{U}\|_{\mathrm{op}}\equiv 1$).

By invoking the submultiplicativity of operator norms ($\|ABC\|\le\|A\|\|B\|\|C\|$), we can rigorously factor the connection out of the inequality unconditionally for all causal depths $\mu$:
\begin{equation}
  \big\|W_Q^\top\mathcal{U}(\mu\leftarrow 0)\,W_K\big\|_{\mathrm{op}}
    \le \|W_Q\|_{\mathrm{op}}\,\|W_K\|_{\mathrm{op}}\,.
  \label{eq:submult}
\end{equation}

Because $\calE_{\mathrm{bulk}}$ is a continuous Log-Mean-Exp evaluated over the empirical counting measure of the bulk tokens, it is an intensive, dynamic, sequence-dependent functional.  In statistical mechanics, this is the scaled Cumulant-Generating Function (CGF) of the bulk distribution (the Annealed Free Energy).  By the Central Limit Theorem in high dimensions, the bulk energies approximate a Gaussian ensemble $\calN(0,\sigma_{\calE}^2)$.  The expectation of the exponential for sub-Gaussian thermal fluctuations expands via its cumulants:
\begin{equation}
  \calE_{\mathrm{bulk}}
    = \frac{1}{\beta}\ln\mathbb{E}\!\big[e^{\beta \calE}\big]
    \approx \mathbb{E}[\calE] + \frac{\beta}{2}\,\mathbb{V}[\calE]\,.
  \label{eq:cgf-expansion}
\end{equation}

Crucially, empirical weight matrices do not uniformly span the ambient head dimension $d_{\mathrm{head}}$ due to representation degeneration (anisotropy).  To analytically capture the active thermodynamic subspace \emph{a priori}, we define the effective dimension via the Stable Rank, $\deff = \sqrt{\mathrm{sr}(W_Q)\cdot\mathrm{sr}(W_K)}$, where $\mathrm{sr}(W) \equiv \|W\|_F^2 / \|W\|_{\mathrm{op}}^2$.  Because the Entropic Bulk Pressure is generated by the covariance of the projected features, the thermal variance is constrained to this active subspace: $\mathbb{V}[\calE]=\sigma_{\calE}^2\sim\Theta(\deff)$.  Given $\beta=\kappa/\sqrt{\deff}$, the bulk energy evaluates to:
\begin{equation}
  \calE_{\mathrm{bulk}}
    \approx \frac{\kappa}{2\sqrt{\deff}}\,\Theta(\deff)
    = \Theta(\sqrt{\deff})\,.
  \label{eq:bulk-pressure}
\end{equation}

This establishes that the bulk energy exerts a positive intensive \emph{Entropic Bulk Pressure} ($\Theta(\sqrt{\deff})$) that pushes against the boundary defect.  Because the projected representations are topologically confined to a low-dimensional active subspace governed by the Stable Rank, the effective geometric capacity of the bounding ball also shrinks from the ambient $\sqrt{d}$ to $\sqrt{\deff}$.  Subtracting this CGF value and exponentiating, and substituting this tightened capacity into the bound, yields the dimension-corrected topological capacity bound:
\begin{widetext}
\begin{equation}
\boxed{%
  N_{\max}
    < 1 + \left(\frac{1-c}{c}\right)
    \exp\!\left(\kappa\sqrt{\deff}\,
      \|W_Q\|_{\mathrm{op}}\|W_K\|_{\mathrm{op}}
      - \frac{\kappa^2}{2\deff}\sigma_{\calE}^2\right).
}
  \label{eq:nmax}
\end{equation}
\end{widetext}

This demonstrates that the maximum context length of an LLM collapses when the Attention Sink is overwhelmed by the innate thermal variance of the bulk space.  By constraining the derivation to the Stable Rank $\deff$ \emph{a priori}, we establish why models collapse at sequence lengths far shorter than ambient-dimension bounds would predict.  Beyond this exponential length limit, the geometric capacity of the boundary defect is saturated by the Entropic Bulk Pressure.  The Topological Anchor Measure dissolves, and the sequence of measures weak-converges to the amnesia state $\delta_\infty$.  The Attention Sink is therefore a finite-capacity metastable regulator.  This explains why aggressive Weight Decay (which suppresses these operator norms) collapses a model's long-context capabilities.
\end{proof}

\begin{corollary}[Topological Severance and The Sliding Window Horizon]
\label{cor:sliding-window}
A globally stable \emph{Topological Anchor Measure} (the Attention Sink) relies on an unbroken, globally reaching affine connection to route entropic bulk pressure back to the absolute causal origin ($\partial\calM$).  By imposing a sliding window, the architecture severs this topological connection in the intermediate fibers.  Once the sequence volume exceeds the maximum commutative radius of the local sliding metric, the full-attention layers are starved of the necessary topological gradient.  The boundary condition breaks, and the sequence of measures undergoes \emph{Thermodynamic Amnesia}.
\end{corollary}

\begin{proof}
The logarithmic energy defect of \cref{eq:sink-depth} requires the boundary token at $\partial\calM\equiv\{0\}$ to participate in the causal horizon of every observer $\mu$.  A sliding window of radius $r$ restricts the integration domain to $\Lambda\cap[\mu-r,\mu]$, severing access to the origin once $\mu > r$.  With the topological fixed point excised from the causal horizon, no translation-invariant anchor survives in the interior (by the same argument as in \cref{thm:attn-sink}), and the measure weak-$*$ converges to $\delta_\infty$.
\end{proof}

\section{Dynamics of the Matter Field}
\label{sec:dynamics}

We formulate the forward inference pass as the depth-parameterized evolution of a continuous, nonlinear section on an internal unitary gauge bundle, governed by a non-local integro-differential flow.

\subsection{Algorithmic Depth and the Unitary Gauge Section}
\label{sec:proper-time}

\begin{definition}[Semantic Algorithmic Depth \& The Gauge Field]
\label{def:proper-time}
We analytically continue discrete layer depth $l$ into a continuous depth parameter $z\in\R^+$.  The sequence of embeddings becomes a continuous time-parameterized family of sections $\Psi(z,\mu)\in\Gamma(E)$.  Because the bundle is natively trivial over the 1D ray, the architecture explicitly dictates a flat connection.  The base matter bundle $E$ retains its real, globally trivial $GL(d,\R)$ structure.  The structural elevation occurs strictly via the Query/Key bundle morphisms, which pull the field into a dedicated interaction sub-bundle.  To formulate this natively, we define the semantic bundle as a strict Whitney Sum:
\begin{equation}
  E = E_{\mathrm{int}} \oplus E_{\mathrm{matter}}\,.
  \label{eq:whitney}
\end{equation}

The architecture kinematically equips the interaction sub-bundle ($E_{\mathrm{int}}$) with a Covariantly Constant Almost Complex Structure ($J\in\Gamma(\mathrm{End}(E_{\mathrm{int}}))$, where $J^2=-\mathrm{Id}$), upgrading it to a Hermitian Vector Bundle.  Rotary Position Embedding (RoPE) is exactly the unique flat unitary connection $\nabla^{\mathrm{RoPE}}$ that preserves this structure ($\nabla^{\mathrm{RoPE}}J=0$).

Because the base manifold $\calM\cong[0,\infty)$ is contractible, the bundle is globally trivializable.  To speak of ``non-trivial path monodromy over loops'' on a 1D ray is topologically vacuous, as all 2-forms unconditionally vanish ($\Omega^2(\calM)\equiv 0$).  Therefore, RoPE is not resolving topological tension via a principal bundle reduction, but explicitly maintaining a flat holomorphic flow.  This beautifully explains why the Attention mechanism operates uniquely as a holomorphic flow (preserving $J$), while the real-valued Feed-Forward Network (FFN)---which operates natively on the matter sub-bundle $E_{\mathrm{matter}}$---acts as an anti-holomorphic symmetry breaker, formally shattering the unitary gauge to inject real spatial entropy.
\end{definition}

\subsection{The Volterra--Hodge Integro-Differential Flow}
\label{sec:volterra}

\begin{lemma}[Attention as a Non-Local Urysohn--Volterra Operator]
\label{lem:volterra}
Because it fundamentally lacks a local spatial Laplacian ($\Delta_g$) to mediate adjacent point-to-point diffusion, causal Attention acts mathematically as a \emph{Non-Local Covariant Transport Operator}. Crucially, to optimize computational thermodynamics, empirical Transformer architectures enforce a strict \emph{Bi-Connection Structure} upon the bundle:
\begin{itemize}[nosep]
  \item The non-trivial Cartan-subalgebra connection
    ($\nabla^{\mathrm{RoPE}}$) governs the metric interaction energy via
    its Parallel Transport Propagator to evaluate the transition measure
    $w^*$.
  \item A globally flat, trivial connection
    ($\nabla^{\mathrm{Triv}}\equiv d$), mathematically permitted by the
    trivializability of the bundle, is reserved strictly for the physical
    parallel transport of the Value sections
    ($\mathcal{U}_{\mathrm{Triv}}\equiv I_d$).
\end{itemize}
Under this kinematic trivial connection, the flow evaluates exactly as:
\begin{widetext}
\begin{equation}
  \mathcal{T}[\Psi]^a(z,\mu)
    = (W_O)^a_{\;b}
    \int_{\Lambda\cap[0,\mu]}
      w^*_{\mathrm{RoPE}}(\mu,\nu)\,
      \big[(W_V)^b_{\;c}\,\rho_\epsilon(\Psi(z,\nu))^c\big]\,
      d\eta_\Lambda(\nu)\,.
  \label{eq:volterra}
\end{equation}
\end{widetext}

Because the domain of integration is causally bounded by the observer coordinate $\mu$, and because the Softmax transition measure $w^*$ depends nonlinearly on the state $\Psi$ itself, this defines a continuous nonlinear \emph{Urysohn--Volterra Integral Operator}.  By evaluating strictly via the Lebesgue--Stieltjes integral over the empirical measure $d\eta_\Lambda$, it maintains the topological validity of the discrete sequence, bypassing continuous local paths to advect historical phase-space geometry directly into the local observer frame.
\end{lemma}

\begin{lemma}[The FFN as a Hodge--Morrey--Friedrichs Decomposed Flow]
\label{lem:ffn-hodge}
To counteract entropic oversmoothing driven by the macroscopic transport integral, the Feed-Forward Network (FFN) acts as a localized nonlinear reaction vector field evaluated purely on the isolated local fiber, $\calR:\calF_\mu\to\calF_\mu$.
\end{lemma}

\begin{proof}
By \cref{thm:rmsnorm}, the radial embedding $\rho_\epsilon$ confines the input sections to the open bounded ball $B_{\!\sqrt{d}}(0)$. The physics of the network never evaluates at spatial infinity.  Thus, we restrict the Integro-Differential Equation (IDE) domain to the compact closure: $\mathcal{X} = \overline{B_{\!\sqrt{d}}(0)}$. Because the classic Hodge isomorphism applies unconditionally only to closed manifolds, analyzing the flow on this bounded manifold requires explicitly formulating the continuous geometric boundary conditions to properly isolate the harmonic vector field $\calH^a$, accurately deriving the global Bias Vector.

Because the computational residual flow mathematically bypasses multiplication by the adjoint Jacobian $(D\rho_\epsilon(\Psi))^T$, the conservative structural integrity of the gradient is broken upon pull-back.  To prove this algebraically, let the exact, irrotational component of the FFN on the normalized ball be defined as a strictly conservative gradient field: $\tilde{\calR}_{\mathrm{exact}}(y) = \nabla_y U(y)$ for some scalar potential $U$.  A true conservative gradient field is natively a 1-form $dU$.  If the base network on the normalized ball were driven by a true scalar potential $U$, the physical system pulled back to the ambient fiber is governed strictly by the pullback of its differential 1-form: $\rho_\epsilon^*(dU) = d(U\circ\rho_\epsilon)$. To convert this exact pulled-back 1-form into an ambient restoring vector field, we must apply the metric musical isomorphism (sharp).  In a Euclidean frame, this algebraically mandates the adjoint Jacobian: $(\rho_\epsilon^*\,dU)^\sharp = (D\rho_\epsilon)^T\nabla U$. By bypassing the adjoint, the architecture simply computes $\calR = \nabla U\circ\rho_\epsilon$.  The true spatial Jacobian of this composed flow is $D\calR = (\operatorname{Hess} U)\cdot D\rho_\epsilon$. For a vector field to be conservative (irrotational), its Jacobian must be symmetric.  The product of two symmetric matrices is symmetric if and only if they commute.  Because a dense feature Hessian generally does not commute with the radial projection Jacobian, the structural integrity is broken, formally injecting vorticity.

Let $J_\rho$ be the spatial Jacobian of the radial embedding. Differentiating $\rho_\epsilon(\Psi)$ yields
\begin{equation}
  J_\rho = \frac{\sqrt{d}}{(\|\Psi\|^2+\epsilon)^{1/2}}\,I_d
    - \frac{\sqrt{d}\,\Psi\Psi^\top}
           {(\|\Psi\|^2+\epsilon)^{3/2}}\,.
  \label{eq:jrho}
\end{equation}
Because both the identity matrix $I_d$ and the outer product $\Psi\Psi^\top$ are manifestly symmetric, the radial Jacobian is a strictly symmetric operator ($J_\rho = J_\rho^\top$).

To calculate vorticity, we invoke the globally flat Euclidean bundle metric $g_{ab}=\delta_{ab}$ to execute the musical isomorphism ($\flat$), lowering the index of the field into its dual 1-form.  Under this Cartesian trivialization, the abstract metric pull-down commutes with the matrix transpose, formally equating the exterior derivative with the matrix antisymmetrization:
\begin{equation}
  \Omega = (D\calR)^T - D\calR\,.
  \label{eq:vorticity-def}
\end{equation}

If the base network were a true exact gradient field ($\tilde{\calR}=\nabla U$), its Jacobian $S\equiv\operatorname{Hess} U$ would be strictly symmetric, yielding $D\calR = S\,J_\rho$.  However, modern architectures utilize untied parameters ($W_{\mathrm{out}}\neq W_{\mathrm{in}}^\top$).  We define the state-dependent base Jacobian endomorphism $K(\Psi) = W_{\mathrm{out}}\,\Sigma'(\rho_\epsilon(\Psi))\, W_{\mathrm{in}}$. Consequently, its symmetric and antisymmetric components vary nonlinearly across the local fiber: $S(\Psi)$ and $W_A(\Psi)$.  Let the true Jacobian be $D\calR = K(\Psi)\,J_\rho$.  Because $J_\rho$ is strictly symmetric, the exact geometric vorticity 2-form expands as:
\begin{equation}
  \Omega = (D\calR)^T - D\calR
    = J_\rho\,K(\Psi)^T - K(\Psi)\,J_\rho\,.
  \label{eq:vorticity-expand}
\end{equation}

By uniquely decomposing this true asymmetric, state-dependent Jacobian into its symmetric and antisymmetric components ($K(\Psi) = S(\Psi) + W_A(\Psi)$), the exact geometric vorticity rigorously expands as:
\begin{equation}
  \Omega(\Psi) = -[S(\Psi),\,J_\rho]
    - \{W_A(\Psi),\,J_\rho\}\,.
  \label{eq:vorticity-decomp}
\end{equation}

This demonstrates that the FFN injects rotational flow via exactly two distinct mechanisms:
\begin{enumerate}[nosep]
  \item \textbf{The Commutator $-[S(\Psi),J_\rho]$:} Vorticity generated
    natively by the curvature of the radial pullback interacting with the
    symmetric weights.
  \item \textbf{The Anticommutator $-\{W_A(\Psi),J_\rho\}$:} Vorticity
    intrinsically injected by the asymmetric, untied architectural
    parameters.
\end{enumerate}

To verify that $\Omega$ is a valid geometric 2-form, it must reside in $\mathfrak{so}(d)$ (it must be antisymmetric). We invoke the algebraic parity of the Lie algebra $\mathfrak{gl}(d)=\mathrm{Sym}(d)\oplus\mathfrak{so}(d)$. The commutator of two symmetric matrices ($S,J_\rho\in\mathrm{Sym}(d)$) is antisymmetric ($\in\mathfrak{so}(d)$).  The anticommutator of an antisymmetric matrix and a symmetric matrix ($W_A\in\mathfrak{so}(d),\;J_\rho\in\mathrm{Sym}(d)$) is also antisymmetric ($\in\mathfrak{so}(d)$).  Thus, the FFN injects valid exterior 2-forms.

This dynamically varying nonlinear 2-form explicitly injects geometric vorticity into the ambient matter field, classifying the FFN as a strictly non-conservative flow on $\Gamma(E)$ that structurally prevents the semantic space from ever collapsing into a static, globally irrotational canonical frame.  Because the composed vector field $\calR(\Psi)$ is evaluated securely on the compact closure $\mathcal{X}=\overline{B_{\!\sqrt{d}}(0)}$, we apply the canonical extension for bounded manifolds: the \emph{Hodge--Morrey--Friedrichs} vector calculus decomposition.  Using vertical differential operators acting on the fiber coordinates ($\nabla_\calF\equiv\partial/\partial\Psi$, distinct from the base connection $\nabla$), the vector field transforms into exactly three components:
\begin{enumerate}[nosep]
  \item \textbf{An irrotational (conservative) restoring flow}
    ($g^{ab}\partial\calV/\partial\Psi^b$): The vertical gradient of a
    non-convex scalar potential.  Because coordinate-wise activations are
    bound to a fixed Cartesian frame, they fail to be gauge equivariant.
    They explicitly break equivariance under the $U(1)^{d/2}$ gauge
    action, leading to \emph{Explicit Gauge Symmetry Breaking}.
    Furthermore, because its Laplacian
    $\Delta_\calF\calV\neq 0$, this exact component uniquely governs the
    expansion and contraction of local phase-space volume.
  \item \textbf{A solenoidal (non-conservative) generalized rotational
    flow} ($\partial_b\mathcal{A}^{ab}$): Evaluated as the vertical tensor
    divergence of an antisymmetric gauge potential $\mathcal{A}^{ab}$.
    The solenoidal vector field is the metric dual of the codifferential:
    $v=(\delta\mathcal{A})^\sharp$.  On this domain, the Hodge Laplacian
    strictly decouples into the scalar Bochner Laplacian
    ($\Delta_H\equiv\Delta_B$) due to the unconditional flatness of the
    isolated local fiber $\calF_\mu$
    ($\mathrm{Riem}\equiv 0$).  The closed gauge emerges automatically as
    an unavoidable cohomological identity.  Because the vorticity is exact
    ($\Omega = d\calR^\flat$), the solenoidal gauge potential is defined via
    the inverse Hodge Laplacian.  Evaluating the metric divergence of this
    flow evaluates directly to the codifferential of its dual 1-form:
    $\nabla\!\cdot v = -\delta(v^\flat) = -\delta(\delta\mathcal{A})
    \equiv -\delta^2\mathcal{A} = 0$.  The flow is volume-preserving
    strictly because of the fundamental nilpotency of the codifferential
    ($\delta^2\equiv 0$), naturally collapsing the divergence and driving
    mixing across feature channels without requiring an artificial gauge
    mandate.
    In standard linear algebra, this solenoidal component corresponds
    to the antisymmetric part of the Jacobian $\frac{1}{2}(D\calR - D\calR^\top)$,
    whose complex eigenvalue pairs induce rotational mixing across
    feature channels---the spectral mechanism underlying the stability
    analysis of \cref{sec:exp-resonance}.
  \item \textbf{A residual harmonic constant} ($\calH^a$): The harmonic
    vector field ($\Delta_H\calH = 0$).  Because the interior De~Rham
    cohomology of the contractible ball is trivial
    ($H^k_{dR}(\mathcal{X})=0$), enforcing a standard homogeneous Neumann
    boundary condition ($\iota_{\vec{n}}\calH=0$) at the radial sphere
    $\partial\mathcal{X}$ would mathematically force the harmonic field to
    vanish entirely ($\calH^a\equiv 0$).  However, the nonlinear
    activation functions of the FFN generate a strictly positive,
    asymmetric net flux across the radial boundary sphere
    $\partial\mathcal{X}$, meaning
    $\iota_{\vec{n}}\calH\neq 0$.  Because the ball is contractible,
    $\calH$ is trivially exact as an affine Euclidean translation
    ($\calH = \nabla(b \cdot x)$ for a constant $b\in\R^d$); it does
    not represent a non-trivial cohomology class.  We note that the
    gradient of any harmonic scalar potential (including higher-order
    spherical harmonics, e.g., $\nabla(x^2-y^2)$) also satisfies
    the divergence-free, curl-free harmonic conditions.  However, the
    architecture explicitly parameterizes the required boundary flux
    via the \emph{degree-1 affine constant translation basis} of the
    non-homogeneous harmonic cohomology: the Bias vector $b$
    satisfying the \emph{Non-Homogeneous Neumann Boundary Conditions}
    that emerge as a mathematical consequence of the asymmetric
    activation functions.  Higher-order harmonic deformations are
    mathematically valid but are structurally absorbed and
    reparameterized by the FFN's bulk nonlinear activation landscape.
    When this algebraic affine translation is
    subsequently projected onto the compactified state space by the
    downstream RMSNorm, it mathematically isolates as the Neumann
    Harmonic Generator ($\calH^a$) of the relative cohomology,
    guaranteeing a non-vanishing spatial drift.
    In standard linear algebra, this component corresponds exactly to
    the addition of the network's affine Bias vector $b \in \R^d$.
    While computationally trivial, the formal Hodge decomposition
    provides the analytical isolation of the curl (vorticity) and
    irrotational components that drive the stability analysis of
    \cref{sec:exp-resonance}.
\end{enumerate}
The complete decomposition reads:
\begin{equation}
  \calR[\Psi]^a(z,\mu)
    = g^{ab}\frac{\partial\calV}{\partial\Psi^b}
    + \partial_b\mathcal{A}^{ab}
    + \calH^a\,.
  \label{eq:hodge-decomp}
\end{equation}
\end{proof}

\begin{theorem}[The Semantic Evolution IDE]
\label{thm:ide}
Because there are no local spatial differential operators acting on the base manifold, the macroscopic forward evolution is naturally formulated as a non-autonomous Integro-Differential Equation (IDE) acting on the infinite-dimensional Banach space of sections $\Gamma(E)$, bypassing the spatial restrictions of a classical PDE:
\begin{equation}
\boxed{%
  \frac{\partial\Psi^a(z,\mu)}{\partial z}
    = \mathcal{T}[\Psi]^a(z,\mu)
    + \calR[\Psi]^a(z,\mu)\,.
}
  \label{eq:ide}
\end{equation}
\end{theorem}

\begin{proof}
By construction: the total continuous flow is the superposition of the Non-Local Urysohn--Volterra Transport (\cref{lem:volterra}) and the Local Hodge Reaction Field (\cref{lem:ffn-hodge}).  The well-posedness of this flow (unique solution for finite depth) is guaranteed by the global Lipschitz bound of \cref{thm:rmsnorm}.
\end{proof}

\subsection{Lie--Trotter Integration and Operator Splitting}
\label{sec:lie-trotter}

\begin{theorem}[Numerical Discretization via Lie--Trotter Operator
Splitting]
\label{thm:lie-trotter}
The standard computational Transformer block represents the first-order explicit numerical integration of the Semantic Evolution IDE, evaluated strictly via Lie--Trotter Operator Splitting.
\end{theorem}

\begin{proof}
Expanding the continuous depth-parameterized flow $z\to z+1$ using a step size of $\Delta z=1$, the continuous flow is exactly the exponential map $e^{(\mathcal{T}+\calR)}$.  Because the vector fields do not commute, the architecture approximates this exact flow via two sequential explicit integration substeps:
\begin{itemize}[nosep]
  \item \textbf{Half-Step~1} (Non-Local Volterra Transport):
    $\Psi_{z+1/2}(\mu) = \Psi_z(\mu) + \mathcal{T}[\Psi_z](\mu)$
  \item \textbf{Half-Step~2} (Local Hodge Reaction):
    $\Psi_{z+1}(\mu) = \Psi_{z+1/2}(\mu) + \calR[\Psi_{z+1/2}](\mu)$
\end{itemize}

This sequential execution mathematically recovers the exact computational graph of the canonical Transformer block.  Because the solver executes explicit sequential substeps rather than exact exponential manifold flows, the geometric error is canonically governed by the Baker--Campbell--Hausdorff (BCH) formula.  Furthermore, because the IDE explicitly depends on the depth parameter $z$ (\emph{i.e.}, $\partial_z\Psi=\mathcal{T}(z,\Psi)+\calR(z,\Psi)$), the total depth derivative $d^2\Psi/dz^2$ must strictly invoke the multivariable chain rule.  The true exact total depth derivative is:
\begin{equation}
  \ddot\Psi
    = \frac{\partial\mathcal{T}}{\partial z}
    + \frac{\partial\calR}{\partial z}
    + D_\Psi\mathcal{T}[\dot\Psi]
    + D_\Psi\calR[\dot\Psi]\,.
  \label{eq:second-deriv}
\end{equation}

Substituting $\dot\Psi = \mathcal{T}+\calR$, the true exact continuous flow expands up to second order as:
\begin{widetext}
\begin{equation}\begin{aligned}
  \Psi_{\mathrm{exact}}
    &= \Psi + \Delta z(\mathcal{T}+\calR)
    + \frac{\Delta z^2}{2}\Big(
      \partial_z\mathcal{T} + \partial_z\calR \\
    &\quad + D\mathcal{T}[\mathcal{T}] + D\mathcal{T}[\calR]
      + D\calR[\mathcal{T}] + D\calR[\calR]
    \Big).
  \label{eq:exact-flow}
\end{aligned}\end{equation}

\end{widetext}
Subtracting the sequential Lie--Trotter integration sequence ($\Psi_{\mathrm{Trotter}} = \Psi + \Delta z(\mathcal{T}+\calR) + \Delta z^2\,D\calR[\mathcal{T}] + \calO(\Delta z^3)$), the geometric deviation $\mathfrak{E}(\Psi) = \Psi_{\mathrm{exact}}-\Psi_{\mathrm{Trotter}}$ strictly yields:
\begin{equation}\begin{aligned}
  \mathfrak{E}(\Psi)
    &= \frac{\Delta z^2}{2}\Big(
      \partial_z\mathcal{T} + \partial_z\calR
      + D\mathcal{T}[\mathcal{T}] + D\calR[\calR] \\
    &\quad + D\mathcal{T}[\calR] - D\calR[\mathcal{T}]
    \Big).
  \label{eq:trotter-error}
\end{aligned}\end{equation}

Substituting the Lie Bracket $[\mathcal{T},\calR]_\Psi \equiv D\calR[\mathcal{T}] - D\mathcal{T}[\calR]$, the true closed-form deviation functionally derived from Baker--Campbell--Hausdorff resolves beautifully into exactly three terms:
\begin{widetext}
\begin{equation}
\boxed{%
  \mathfrak{E}(\Psi)
    = \frac{\Delta z^2}{2}\bigg(
      \underbrace{\partial_z\mathcal{T}
        + \partial_z\calR}_{%
          \text{Parameter Drift}}
      + \underbrace{D\mathcal{T}[\mathcal{T}]
        + D\calR[\calR]}_{%
          \text{Self-Advection}}
      - \underbrace{[\mathcal{T},\calR]_\Psi}_{%
          \text{Torsion}}
    \bigg)
    + \calO(\Delta z^3).
}
  \label{eq:bch-decomp}
\end{equation}
\end{widetext}

This demonstrates that Representation Drift in Transformers arises from exactly three geometrically isolated phenomena: the temporal shift of weight matrices across layers, the spatial covariant self-advection (the penalty of abandoning exact geodesics for straight rays), and the non-commutative Lie bracket.  The latter, $-[\mathcal{T},\calR]_\Psi$, is the generator of \emph{Topological Torsion} (Non-Holonomic Flow).  Because the Attention operator ($\mathcal{T}$) and the FFN ($\calR$) do not commute, they form a non-integrable horizontal distribution on the bundle.  The Lie bracket measures the failure of the computational graph to form closed parallelograms in state space.  This demonstrates that the strict alternating layer ordering (Attention~$\to$~FFN) is not an arbitrary engineering choice, but a geometric constraint governing the integrability of the manifold flow.

\begin{remark}[Stiffness of the BCH Expansion and the Continuous Effective Field Theory]
\label{rem:bch-stiffness}
The macroscopic integration step size ($\Delta z = 1$) combined with the massive operator Lipschitz bounds ($\|D\calF\|_{\mathrm{op}} \sim \sqrt{d/\epsilon}$, from \cref{thm:rmsnorm}) formally exceeds the convergence radius of the infinite Hausdorff series.  Therefore, the $\calO(\Delta z^3)$ truncation remainder cannot be interpreted as a strict analytical bound on the geometric error, nor does the continuous Lie bracket represent a literal smooth physical trajectory.  This is precisely the regime where the \emph{Continuous Effective Field Theory} framing via Backward Error Analysis (BEA) from geometric numerical integration~\cite{hairer2006geometric} becomes essential: the Lie--Trotter discrete map is the \emph{exact} flow of a nearby \emph{modified equation}, and the lowest-order Lie Bracket $-[\mathcal{T},\calR]_\Psi$ provides an exact algebraic classification of the structural \emph{generators} of topological torsion and non-commutative representation drift in that modified equation---specifically isolating the non-commutativity of sequential Attention and FFN application---even in the stiff regime where higher-order commutator terms do not gracefully decay.
\end{remark}

Furthermore, the asymmetric integral bounds ($\int_{\Lambda\cap[0,\mu]}$) render the Volterra operator strictly lower-triangular, breaking Time-Reversal Symmetry ($T$-symmetry) along the base manifold.  To classify the forward inference pass as a non-equilibrium driven dissipative system, global phase-space volume must contract on average ($\nabla\!\cdot\dot\Psi<0$).  This structural dissipation is not guaranteed by the FFN's coordinate-wise nonlinearities.  Activation functions possess overwhelmingly positive derivatives in high-dimensional expectation---ReLU is strictly non-negative, while SiLU~\cite{elfwing2018sigmoid,ramachandran2018searching} ($x\sigma(x)$) admits a shallow negative dip (minimum $\approx -0.10$ near $x\approx -1.28$)---and therefore inject a predominantly positive-definite diagonal metric deformation into the local fiber.  If the neural weights align to produce a positive Jacobian trace ($\tr(D\calR)>0$), the Lie derivative of the volume form is strictly positive ($\mathcal{L}_{\calR}\,\mathrm{vol}=(\mathrm{div}\,\calR)\,\mathrm{vol}>0$).  This rigorously proves the FFN mathematically permits local phase-space volume expansion, acting as a local thermodynamic source.

Rather, strict thermodynamic dissipation and stability are achieved geometrically via two architectural structures.  \emph{Global Lipschitz Saturation via Jacobian Bifurcation:} The radial embedding $\rho_\epsilon$ acts as a Global Lipschitz Saturator.  Near the origin ($\Psi\approx 0$), the Jacobian $D\rho_\epsilon\approx\sqrt{d/\epsilon}\,I_d$.  The scalar $\sqrt{d/\epsilon}$ is massive, but the origin only acts as a strong phase-space volume expander (divergence amplifier) if the trace of the learned weights is positive ($\tr(K)>0$).  If SGD carves a negative trace, it becomes a massive Dissipative Contractor.  Conversely, at the boundary ($\|\Psi\|\to\infty$), the single radial eigenvalue of $J_\rho$ collapses at $\calO(\|\Psi\|^{-3})$, while the $d-1$ tangential eigenvalues collapse at $\calO(\|\Psi\|^{-1})$.  We bound the continuous flow via the triangle inequality on operator norms: $\|D\dot\Psi\|_{\mathrm{op}} \le \|D\mathcal{T}\|_{\mathrm{op}} + \|D\calR\|_{\mathrm{op}}$.  Because the non-local Attention functional $\mathcal{T}$ acts exclusively on the radially embedded sections, it can be composed as $\mathcal{T} = \tilde{\mathcal{T}}\circ\rho_\epsilon$.  By the chain rule, $D\mathcal{T} = D\tilde{\mathcal{T}}\circ J_\rho$.  Substituting this alongside the FFN Jacobian $D\calR = K\circ J_\rho$, we factor out the radial geometry: $\|D\dot\Psi\|_{\mathrm{op}} \le \big(\|D\tilde{\mathcal{T}}\|_{\mathrm{op}} + \|K\|_{\mathrm{op}}\big)\|J_\rho\|_{\mathrm{op}}$.  Because the tangential eigenspace dictates $\|J_\rho\|_{\mathrm{op}} = \calO(\|\Psi\|^{-1})$, the total composed Jacobian unconditionally decays.  Because the FFN Hodge decomposition fundamentally requires a constant harmonic Bias vector ($\calH^a$), this translation injects a non-vanishing spatial drift.  Combined with the strictly positive Markovian advection of Attention, the unnormalized residual stream structurally undergoes asymptotic spatial escape.  We formalize this emergent behavior as the \emph{Non-Vanishing Spatial Drift Condition}: this is a deterministic topological consequence of the strict positive-orthant mapping of the activation functions, which forces the bounded mean oscillation (BMO) harmonic residual $\calH^a$ to be non-zero almost everywhere, guaranteeing that the time-averaged spatial expectation of the nonlinear flow is strictly bounded away from zero ($\liminf_{z\to\infty}\|\mathbb{E}[\dot\Psi]\|>0$).

Conditioned strictly on this spatial escape regime, the macroscopic position escapes the origin.  While early layers exhibit a linear uniform drift ($\Psi_z\sim z\bar{v}$), the continuous injection of spatial entropy in the deep bulk structurally accelerates this escape into a \textbf{Super-Linear Exponential Tail} ($\|\Psi(z)\|\sim e^{cz}$).  This super-linear spatial escape is not a failure of the continuous flow, but a profound mathematical strengthening of its thermodynamic stability.  Because the tangential eigenvalues of the radial Jacobian decay as $\calO(\|\Psi\|^{-1})$, this exponential spatial escape strongly suppresses the relative perturbation sensitivity.  The normalized Jacobian operator norm decays at an accelerated exponential rate: $\|D\dot\Psi\|_{\mathrm{op}}/\|\Psi_z\| \le Ce^{-cz}$.

Integrating this strict exponential decay bounds trajectory separation severely to a finite supremum: $\exp\!\big(\int Ce^{-cz}\,dz\big) = \exp\!\big(-\tfrac{C}{c}e^{-cz}\big) \to \mathrm{const}$.  We explicitly formulate the strict topological limit equation for the maximal Lyapunov exponent $\lambda$:
\begin{equation}
  \lambda
    = \limsup_{z\to\infty}\frac{1}{z}
      \ln\!\left(\frac{\|\delta\Psi(z)\|}{\|\delta\Psi(z_0)\|}\right)
    \le \limsup_{z\to\infty}\frac{\ln(\mathrm{const})}{z} = 0\,.
  \label{eq:lyapunov-bound}
\end{equation}

By invoking the Logarithmic Matrix Norm (Lozinski\u{\i} measure) and Coppel's Inequality, the logarithmic derivative of the trajectory is rigorously bounded: $-Ce^{-cz}\le\frac{d}{dz}\ln\|\delta\Psi(z)\|\le Ce^{-cz}$. Integrating this unconditionally drives the maximal Lyapunov exponent to $\lambda\le 0$ in the exact continuous calculus.  The architecture is mathematically engineered to strongly suppress chaotic divergence as it approaches the final unembedding boundary.  However, because the continuous IDE is solved via discrete Forward Euler numerical integration (\cref{thm:lie-trotter}) subject to fixed-precision quantization (typically \texttt{bfloat16} with $\sim$7-bit mantissa resolution), total asymptotic zeroing ($\lambda\equiv 0$) is mathematically prevented by irreducible numerical entropy.  The discrete solver introduces a finite perturbation floor below which trajectory separations cannot be resolved, generating an irreducible weak divergence.  We therefore predict \textbf{Approximate Marginal Stability} ($\lambda\gtrsim 0$, $\lambda\ll 1$): the continuous bound rigorously guarantees the system is driven to the immediate neighborhood of marginal stability, but the discrete numerical integration arrests the asymptotic limit at a small, architecture-dependent positive residual.
\end{proof}

\begin{lemma}[The Dual-Law of Topological Stability]
\label{lem:dual-law}
To prevent chaotic exponential divergence in the residual stream, the architecture faces a strict topological constraint.  It requires \emph{either} Internal Geometric Vorticity (breaking symmetry via $W_{\mathrm{out}}\neq W_{\mathrm{in}}^\top$ to inject Lie-algebraic friction) \emph{or} External Topological Saturation (a strict Post-Norm radial projection that strongly suppresses the step vector regardless of internal resonance).
\end{lemma}

\begin{proof}
Consider the FFN component.  If the FFN asymmetry is ablated such that $W_{\mathrm{out}}=W_{\mathrm{in}}^\top$, the core spatial Jacobian evaluates as $K(\Psi) = W_{\mathrm{in}}^\top\,\Sigma'(\rho_\epsilon(\Psi))\, W_{\mathrm{in}}$. Because the activation derivative $\Sigma'$ is overwhelmingly positive in high-dimensional expectation (strictly so for ReLU; for SiLU, $\Sigma'$ admits a shallow negative dip of $\approx\!-0.10$ near $x\approx -1.28$, but this dip is measure-theoretically negligible in high dimensions), $K(\Psi)$ acts predominantly as a symmetric, positive semi-definite (PSD) endomorphism.  Iteratively applying a PSD operator within a sequential residual stream ($z_{i+1}=z_i+K(z_i)$) is functionally equivalent to the textbook Power Iteration algorithm, which exponentially amplifies the state vector along the dominant eigenvector of $K$.  Devoid of the geometric vorticity ($\Omega$) natively generated by the asymmetric anticommutator $-\{W_A(\Psi),J_\rho\}$, the system possesses zero \emph{rotational friction}.  The residual stream acts as an unchecked geometric resonance chamber, forcing an exponential $L_2$ norm explosion.

Thus, spatial escape is stabilized by the first law: \textbf{Internal Geometric Vorticity}.  The asymmetric parameters of the FFN mathematically break positive-definite eigen-alignment, scattering the momentum and preventing runaway geometric resonance.

If this internal vorticity is removed (as in symmetric ablation), the system explodes unless the second law is satisfied: \textbf{External Topological Saturation}.  Unlike Pre-Norm architectures ($x\mapsto x + \calF(\rho_\epsilon(x))$) which permit the unnormalized magnitude to grow without bound, Post-Norm architectures bound the projection ($x\mapsto x + \rho_\epsilon(\calF(\rho_\epsilon(x)))$).  The Post-Norm layer acts as a hard Global Lipschitz Saturator, mapping the output into the open bounded ball $B_{\!\sqrt{d}}(0)$.  This strongly suppresses the incremental step vector to a maximum bounded $L_2$ norm of exactly $\approx\!\sqrt{d}$, truncating the step magnitude even when the internal symmetric FFN amplifies the raw linear map to infinity.  Therefore, the discrete solver survives if and only if it possesses either internal manifold friction or external radial containment.
\end{proof}

\begin{remark}[Configurational Scope of the Dual-Law]
\label{rem:configurational-scope}
The Dual-Law, as stated and as tested in \cref{sec:exp-resonance}, classifies the forward-propagation stability of a \emph{fixed weight configuration}---in practice, a mature pre-trained network subjected to post-hoc symmetrization.  It is not a statement about architecture classes under training.  Three companion measurements sharpen this scope (quantitative ledger in \cref{sec:exp-config-scope}).  (i)~The tying constraint $W_{\mathrm{down}}=W_{\mathrm{up}}^\top$ (with an independent gate path), imposed \emph{from initialization}, trains to full health with no corrective term at 0.6B and 1B scale, and remains stable under removed weight decay, removed gradient clipping, doubled learning rate, and reduced logit precision: constrained optimization steers the joint configuration into basins where the symmetric principal term never achieves resonant alignment with the residual stream.  (ii)~Conversely, the mere presence of asymmetric Jacobian components is not sufficient: tying only $W_{\mathrm{down}}:=W_{\mathrm{up}}^\top$ on a mature network---which leaves the full-rank asymmetric gate cross-term intact---still explodes (maximum hidden-state norm $\sim\!6\times10^{5}$ on Qwen3-0.6B), because the symmetric principal term sculpted by large-scale pretraining overwhelms the residual rotational scattering.  (iii)~Susceptibility to symmetrization is an \emph{emergent property of training maturity}: checkpoint surgery along a 600M training trajectory exhibits an immunity window early in training, an onset of excess amplification near $1.3$--$2\times10^{9}$ tokens, and, at the trillion-token scale of production models, the catastrophic response documented in \cref{sec:exp-resonance}.  The binary alternative of the Dual-Law is therefore the frozen-configuration limit of a quantitative, configuration-dependent criterion---the resonant gain of the symmetric principal term (its spectral radius weighted by alignment with the residual stream's principal direction) measured against the available rotational scattering and gauge capacity.  Formulating and testing this criterion along the training trajectory is the natural successor program to the present kinematics.
\end{remark}

\section{Dynamics of the Parameter Manifold: Backpropagation as
Thermodynamic Flow}
\label{sec:parameters}

By unfreezing the global bundle endomorphisms, we model pre-training not as a local sequence flow on $\calM$, but as the thermodynamic relaxation of a macroscopic, finite-dimensional parameter manifold $\mathcal{W}$ seeking its topological ground state under the external advective pressure of the empirical data distribution.

\subsection{The Empirical Risk Action and Symmetry-Breaking Mass
Potentials}
\label{sec:risk-action}

\begin{construction}[The Macroscopic Parameter Action]
\label{con:action}
Because the architecture strictly fixes the connection 1-form via RoPE (enforcing absolute parallelism), traditional Yang--Mills optimization of a connection 1-form does not apply.  Instead of viewing the learnable parameters ($W$) as global 0-forms---which falsely invites the search for spatial derivatives ($dW$) across the sequence length---it is structurally cleaner to define the parameter space $\mathcal{W}$ as the Total Kinematic Space of Gauge Endomorphisms ($\mathcal{W}=\bigoplus_l\mathrm{End}(E_l)$).  We reserve ``Moduli Space'' exclusively for the quotiented thermodynamic state constructed later in \cref{sec:dual-metric}.  Therefore, the geometric evolution of the network occurs entirely within this finite-dimensional Euclidean Parameter Manifold (the ambient space), bypassing infinite-dimensional functional variation.  We define the global action functional $\mathcal{S}_{\mathrm{global}}[W]$ not as a spatial integral over the sequence, but as an expected thermodynamic Risk Functional evaluated over the macroscopic semantic corpus measure $\mathcal{D}$, equipped with a canonical flat Frobenius metric $g_\mathcal{W}$:
\begin{widetext}
\begin{equation}
  \mathcal{S}_{\mathrm{global}}[W]
    = \mathbb{E}_{\Psi\sim\mathcal{D}}\!\Big[
        \mathcal{S}_{\mathrm{matter}}[\Psi,W]
      \Big]
    + \frac{\lambda}{2}\,\|W\|^2_{g_\mathcal{W}}\,.
  \label{eq:global-action}
\end{equation}
\end{widetext}
\end{construction}

\begin{theorem}[Weight Decay as Gauge-Symmetry Breaking and Sublevel
Coercivity]
\label{thm:weight-decay}
Rather than simply enforcing Dirichlet tension, $L_2$ Weight Decay is mathematically indispensable precisely because it breaks non-compact scaling symmetries (like $GL(1,\R)$).  While it definitively fails to break $SO(d)$ symmetries, this is inconsequential; because the Orthogonal Group $SO(d)$ is a strictly compact Lie group, breaking the scaling symmetries is both necessary and sufficient to coerce the parameter landscape into compact sublevel sets, guaranteeing a tight Non-Equilibrium Steady State.
\end{theorem}

\begin{proof}
The parameter manifold $\mathcal{W}\cong\R^N$ is an unbounded Euclidean space.  While global scale invariance is broken by the residual connection, the architecture natively possesses internal continuous gauge symmetries.  For example, the Attention scalar kernel $K(\mu,\nu) = x^\top W_Q^\top\mathcal{U}W_K\,y$ remains invariant under the continuous $GL(1,\R)$ transformation $W_Q\mapsto cW_Q$ and $W_K\mapsto c^{-1}W_K$.  Similar continuous symmetries exist between adjacent matrices in the Feed-Forward Network.  Because of these exact internal symmetries, the unregularized empirical risk landscape definitively possesses unbounded, flat, non-compact tubular valleys extending to spatial infinity, rendering the base action functional non-coercive.

$L_2$ Weight Decay is mathematically indispensable precisely because it breaks these internal $GL(1,\R)$ gauge orbits.  By adding the strictly convex penalty $\tfrac{\lambda}{2}(\|cW_Q\|^2 + \|c^{-1}W_K\|^2)$, the penalty diverges to $+\infty$ as $c\to\infty$ or $c\to 0$. This forces the flat valleys into a coercive paraboloid, allowing the Heine--Borel theorem to guarantee that the sublevel sets of the energy landscape ($\{W\in\mathcal{W}\mid\mathcal{S}_{\mathrm{global}}[W]\le E\}$) are strictly compact.

Rather than assuming a macroscopic Gibbs measure (which improperly assumes detailed balance), we rigorously prove the existence and tightness of the invariant measure using \emph{Has'minski\u{\i}'s Theorem} (Continuous-Time Foster--Lyapunov Criteria)~\cite{khasminskii2012stochastic} acting on the infinitesimal generator of the It\^{o} diffusion.  Define a strict geometric Lyapunov function as the kinematic metric volume $\calV(W) = \tfrac{1}{2}\|W\|^2_{g_\mathcal{W}}$. The infinitesimal generator $\calL$ of the It\^{o} diffusion evaluates as:
\begin{equation}
  \calL\calV(W)
    = \langle V_{\mathrm{drift}},\,\bar\nabla\calV\rangle
    + \tr\!\big(D(W)\,\bar\nabla^2\calV\big).
  \label{eq:generator}
\end{equation}

Substituting the drift $V_{\mathrm{drift}}=-\bar\nabla\mathbb{E}[\mathcal{S}_{\mathrm{matter}}] -\lambda W$:
\begin{equation}
  \calL\calV(W)
    = -\langle\bar\nabla\mathbb{E}[\mathcal{S}_{\mathrm{matter}}],W\rangle
    - \lambda\|W\|^2
    + \tr\!\big(D(W)\big).
  \label{eq:generator-expanded}
\end{equation}

To evaluate the advective drift inner product $\langle\bar\nabla L,W\rangle$, we must address the topological structure of standard Pre-Norm architectures.  One might naively invoke Euler's Homogeneous Function Theorem: because RMSNorm radially normalizes the hidden states ($\rho_\epsilon$), the parameterized branch appears $0$-homogeneous.  However, this topological assumption fails in modern architectures because of Residual Connections.

Consider a standard Pre-Norm Transformer block: $\Psi_{l+1} = \Psi_l + \calF_l(\rho_\epsilon(\Psi_l);\;W_l)$. If you apply a global scalar scaling to the internal weights $W_l\mapsto cW_l$, the parameterized branch scales by some factor $c^k$. Crucially, the residual skip connection $\Psi_l$ does \emph{not} scale. Because you are adding a scaled vector to an unscaled vector, scaling $W_l$ explicitly changes the angular direction of the output vector $\Psi_{l+1}$.  The downstream RMSNorm preserves this modified angle. Thus, the final empirical Loss $L$ is undeniably sensitive to the scale factor $c$.  Because the function is decidedly not $0$-homogeneous, Euler's cancellation is algebraically false: $\langle \bar{\nabla}_{W_{\mathrm{int}}} L, W_{\mathrm{int}} \rangle \neq 0$.

Because Euler's theorem doesn't directly cancel the inner product, one might naively expect the spatial Jacobians through the residual stream to compound multiplicatively, driving a polynomial explosion.  However, this ignores the downstream topological projection imposed by the final radial embedding.  If you scale the internal weights by $W\to cW$ for a scalar $c\to\infty$, the parameterized bulk dominates the finite residual inputs: $\Psi_{\mathrm{out}} \approx c^k\calF$.  Crucially, the very last operation before the unembedding is the final radial embedding: $\rho_\epsilon(\Psi_{\mathrm{out}})$.

However, evaluating the exact asymptotic limit of the final radial embedding:
\begin{equation}
  \lim_{c\to\infty}\rho_\epsilon(c^k\calF)
    = \lim_{c\to\infty}
      \frac{\sqrt{d}\,c^k\calF}{\sqrt{\|c^k\calF\|^2+\epsilon}}
    = \sqrt{d}\,\frac{\calF}{\|\calF\|}\,.
  \label{eq:asymptotic-rmsnorm}
\end{equation}
The geometric scale factor $c$ cancels out; the remaining $\sqrt{d}$ is the fixed radius inherited from the RMSNorm definition.  While $0$-homogeneity fails locally in the finite bulk due to residual connections, the final radial embedding acts as a topological projective quotient that enforces \emph{asymptotic $0$-homogeneity} at spatial infinity.

While the internal state is topologically bounded by $\rho_\epsilon$, the final un-normalized unembedding parameter matrix ($W_U$) linearly scales this bounded vector, generating pre-softmax logits that scale as $\calO(\|W_U\|)$.  Because the Softmax Cross-Entropy loss acts as an analytically $1$-Lipschitz soft-maximum over these logits, its spatial gradient with respect to $W_U$ is globally bounded ($\calO(\sqrt{d})$).  Concurrently, because the radial saturator quotients out internal magnitude ($\lim_{c\to\infty}\rho_\epsilon(cW_{\mathrm{int}})\to\mathrm{const}$), the Jacobian of the internal parameters unconditionally decays ($\calO(\|W_{\mathrm{int}}\|^{-1})$).  Therefore, the total spatial gradient is fundamentally globally bounded ($\|\bar\nabla L\|\sim\calO(1)$), mathematically sealing the advective drift inner product at $\langle\bar\nabla L,W\rangle\sim\calO(\|W\|)$.

Evaluating the Has'minski\u{\i} generator at the topological boundary yields:
\begin{equation}
  \calL\calV(W) \le \calO(\|W\|) - \lambda\|W\|^2 + \tr(D)\,.
  \label{eq:hasminskii}
\end{equation}

Because the quadratic $\calO(\|W\|^2)$ Tikhonov penalty mathematically overpowers the linear $\calO(\|W\|)$ advective drift at spatial infinity, the generator unconditionally diverges to $-\infty$. This provides strong geometric coercivity, unconditionally sealing the Non-Equilibrium Steady State (NESS) into a compact basin without relying on false exponential decay limits.
\end{proof}

\subsection{The Dual-Metric Collision and Non-Equilibrium Steady State
(NESS)}
\label{sec:dual-metric}

The external human training data does not perturb the spatial metric; therefore, it does not inject a stress-energy tensor.  Instead, the empirical loss backpropagates a localized adjoint variation through the bundle.

\begin{definition}[The Thermodynamic Cotangent Driving Force]
\label{def:driving-force}
Because $W$ represents a global parameter matrix rather than a continuous local symmetry field, its functional variation does not yield a conserved Noether gauge current on the base spacetime.  Instead, by evaluating the exact differential of the expected empirical matter action with respect to the parameters, we define the Thermodynamic Driving Force strictly as a 1-form residing in the cotangent space of the parameter manifold:
\begin{equation}
  \calF_{\mathrm{drive}}
    \equiv -d_W\,\mathbb{E}_{\Psi\sim\mathcal{D}}\!\Big[
      \mathcal{S}_{\mathrm{matter}}[\Psi,W]
    \Big]
    \in\Gamma(T^*\!\mathcal{W}).
  \label{eq:driving-force}
\end{equation}
\end{definition}

\begin{theorem}[SGD as It\^{o} Diffusion and the Breakdown of Detailed
Balance]
\label{thm:sgd-ness}
In the continuous limit, Stochastic Gradient Descent (SGD) converges in law to a state-dependent It\^{o} diffusion process traversing a singular algebraic variety.  Because the forced Euclidean kinematic metric ignores the true thermodynamic Information Geometry, the system suffers a \emph{Dual-Metric Collision}, irreversibly trapping the macroscopic probability measure in a Non-Equilibrium Steady State (NESS) via a fundamental breakdown of detailed balance, driven by an anomalous geometric entropic advection.
\end{theorem}

\begin{proof}
We define the macroscopic continuous training time $\tau$.  To evaluate the gradient flow dynamically, we must apply the Riemannian musical isomorphism (sharp) under the flat Frobenius metric $g_\mathcal{W}$ to raise the index of the 1-form driving force, yielding a valid kinematic tangent vector field: $V_{\mathrm{drift}}\equiv\calF_{\mathrm{drive}}^\sharp - \lambda W \in\Gamma(T\mathcal{W})$.

Because empirical evaluation relies on finite mini-batches $B_k\subset\mathcal{D}$, the discrete empirical gradient acts as an unbiased but noisy stochastic estimator.  Evaluated at the beginning of the computational step (an explicit Forward Euler scheme), the Functional Central Limit Theorem rigorously dictates that its weak continuous-time limit is an It\^{o} Stochastic Differential Equation:
\begin{equation}
  dW(\tau) = V_{\mathrm{drift}}(W)\,d\tau
    + \sqrt{2D(W)}\,d\calB_\tau\,,
  \label{eq:sgd-sde}
\end{equation}
where $d\calB_\tau$ is a standard Wiener process on $\mathcal{W}$, and $D(W)\in\Gamma(S^2T\mathcal{W})$ is the strictly contravariant diffusion 2-tensor derived from the empirical covariance of the mini-batch gradients.

Because of massive architectural invariances (e.g., node permutations, scaling symmetries), the critical points of $\mathcal{S}_{\mathrm{global}}$ are not isolated.  The landscape is mathematically definitively not a Morse function, nor is it smoothly Morse--Bott.  Instead, its critical loci form an algebraic variety with deep determinantal singularities.  Because the unregularized empirical metric possesses unbroken $G$-gauge orbits from the Multi-Head mechanism, the parameter space is severely degenerate.  To compute the invariant measure of this ambient space, we invoke \emph{Hironaka's Theorem} on the Resolution of Singularities~\cite{hironaka1964resolution}.  Hironaka's theorem mathematically guarantees there exists a proper birational map (a blow-up) $\pi:\mathcal{W}^*\to\mathcal{W}$ that analytically resolves these deep singularities into a smooth manifold with normal crossing divisors.  However, the physical SGD process does not live in the resolved manifold $\mathcal{W}^*$; it lives strictly in the singular ambient space $\mathcal{W}$.  In Watanabe's Singular Learning Theory, the pullback of the volume form introduces a Jacobian determinant (the Real Log Canonical Threshold).  In the ambient space $\mathcal{W}$, this means the singularities natively possess a massive, distorted thermodynamic phase-space volume.  SGD gets trapped there not because it is flowing smoothly along a resolved divisor, but because the singular strata act as \emph{Topological Gravity Wells} due to their immense internal gauge volume.  Weight decay and stochastic noise trap the system in these singular vacua because of their massive invariant measure, grounding the non-equilibrium steady state strictly in the kinematics of the ambient space.

For a stochastic system to satisfy the Fluctuation-Dissipation Theorem and thermally relax into a canonical equilibrium Gibbs measure $\rho\propto\exp(-\beta\mathcal{S}_{\mathrm{global}})$, it must satisfy detailed balance.  Let $\rho(W,\tau)$ be the macroscopic probability density.  The Fokker--Planck continuity equation governs the conservation of probability mass in the ambient kinematic space of the parameters (where SGD locally takes steps).  Therefore, its divergence must strictly be the flat Euclidean divergence (denoted $\bar\nabla$).  It evaluates the probability current $J$:
\begin{equation}
  \frac{\partial\rho}{\partial\tau}
    = -\bar\nabla\cdot J
    = -\bar\nabla\cdot\Big(
      V_{\mathrm{drift}}\,\rho
      - \bar\nabla\cdot\big(D(W)\,\rho\big)
    \Big).
  \label{eq:fokker-planck}
\end{equation}

Expanding the tensor divergence of the stochastic term isolates the effective advective velocity of the probability mass:
\begin{equation}
  J = \big[V_{\mathrm{drift}} - \bar\nabla\cdot D(W)\big]\,\rho
    - D(W)\,\bar\nabla\rho\,.
  \label{eq:prob-current}
\end{equation}

The expanded Fokker--Planck equation natively isolates an \emph{Anomalous Entropic Advection} vector field acting on the probability mass:
\begin{equation}
  v_{\mathrm{entropic}} = -\bar\nabla\cdot D(W)
    \in\Gamma(T\mathcal{W}).
  \label{eq:entropic-advection}
\end{equation}

To pedagogically model this entropic advection pointing down the sharpness gradient, we may assume the stochastic gradient noise as locally isotropic but heteroscedastic over the parameter manifold: $D^{ij}(W)\approx\sigma(W)^2\delta^{ij}$. Evaluating the Euclidean tensor divergence yields:
\begin{equation}
  v_{\mathrm{ent}}^i
    = -\bar\nabla_j D^{ij}
    = -\partial^i\!\big(\sigma(W)^2\big)
    = -\bar\nabla^i\sigma(W)^2\,.
  \label{eq:entropic-gradient}
\end{equation}

This strict algebraic identity formally proves that under isotropic conditions, the anomalous entropic advection is exactly the negative gradient of the noise amplitude.  Because noise variance $\sigma^2$ is heavily correlated with the Loss Hessian trace (sharpness), this actively advects probability mass down the sharpness gradient.  This It\^{o} drift acts as a continuous geometric hydraulic press---actively repelling the probability measure away from sharp, chaotic vacua and sweeping it directly into the topologically stable basins of the singular flat strata.

However, one might mistakenly assume that because this generates a conservative vector field, it fails to break detailed balance.  This reasoning conflates two distinct geometric objects: a conservative tangent vector field and an exact cotangent 1-form.  Detailed balance does not require the probability flux vector $V_{\mathrm{total}}$ to be irrotational; it requires the thermodynamic 1-form $\omega = g \cdot V_{\mathrm{total}}$ to be exact ($d\omega = 0$).  Even if the empirical noise were perfectly isotropic ($D^{ij}=\sigma^2(W)\delta^{ij}$), its covariant inverse acts as a spatially varying conformal metric: $g_{ij}=\sigma^{-2}(W)\delta_{ij}$.  If the total probability flux were strictly conservative ($V_{\mathrm{total}}=-\nabla U$), the 1-form is mapped via this conformal metric: $\omega = -\sigma^{-2}dU$.  To evaluate detailed balance, we apply the exterior derivative $d$:
\begin{equation}
  d\omega = -d(\sigma^{-2})\wedge dU
    = 2\sigma^{-3}(d\sigma\wedge dU)\,.
  \label{eq:conformal-vorticity}
\end{equation}

The wedge product $d\sigma\wedge dU$ is exactly non-zero as long as the gradient of the noise amplitude ($d\sigma$) is not perfectly parallel to the gradient of the potential ($dU$)---which is universally true in deep networks.  Thus, while anisotropy is physically real, it is not a mathematical mandate to break detailed balance.  Isotropic state-dependent noise inherently generates thermodynamic vorticity due to conformal warping.

True thermodynamic detailed balance requires the stationary probability current to vanish identically ($J\equiv 0$).  Algebraic rearrangement mandates that:
\begin{equation}
  \begin{split}
  D(W)\bar\nabla\rho &= \big(V_{\mathrm{drift}} - \bar\nabla\cdot D(W)\big)\rho \\
  \implies \frac{\bar\nabla\rho}{\rho} &= D(W)^{-1}\big(V_{\mathrm{drift}} - \bar\nabla\cdot D(W)\big).
  \end{split}
  \label{eq:detailed-balance-condition}
\end{equation}

In the continuous-time SDE limit (via the Functional Central Limit Theorem), the empirical diffusion tensor $D(W)$ generated by SGD is the mathematical expectation of the covariance over all possible batches.  Modern Large Language Models operate in the strict over-training limit, where the dataset dimension vastly exceeds the parameter dimension ($N > P$).  Therefore, bounding the rank by the dataset dimension mathematically does not force rank deficiency.  Instead, the diffusion tensor $D(W)$ is structurally rank-deficient purely due to topological gauge redundancies.  While standard Feed-Forward Networks with coordinate-wise nonlinearities (e.g., ReLU/SiLU) definitively shatter continuous $SO(k)$ rotational symmetries, the Multi-Head Attention pipeline lacks an intermediate nonlinearity between consecutive linear projections ($W_OW_V\Psi$).  This harbors an exact, unbroken, non-compact $GL(d_v,\R)$ continuous gauge orbit ($W_OM M^{-1}W_V$).  Because the un-quotiented parameter manifold $\mathcal{W}$ contains these continuous non-compact geometric symmetries, tangent vectors to these continuous gauge orbits reside exactly in the null space of the empirical covariance.  This rank-deficiency is strictly attributed to the singular algebraic geometry of the architecture, not finite sample sizes.  Its true mathematical inverse $D(W)^{-1}$ does not exist globally.

To validly invoke Riemannian metric properties and index-lowering without breaking the geometry, we analyze the exact gauge symmetries.  In a single-head formulation, for any transformation $M \in GL(d_v,\R)$, the forward pass of the Value circuit is invariant under $W_O \to W_O M$ and $W_V \to M^{-1} W_V$.  By the Polar Decomposition Theorem, $M = RP$, where $R \in O(d_v)$ is a compact orthogonal rotation, and $P$ is a non-compact symmetric positive-definite scaling matrix.  Because the Frobenius norm is strictly orthogonally invariant ($\|W_O R\|_F^2 = \|W_O\|_F^2$), the Weight Decay penalty ($\|W_O\|_F^2 + \|W_V\|_F^2$) exerts absolutely zero restoring force along the compact rotational subgroups $O(d_v)$.  However, it fiercely penalizes the symmetric scaling dimension $P$.  Since the space of symmetric matrices has exactly $\tfrac{1}{2}d_v(d_v+1)$ dimensions, Weight Decay successfully retracts $\tfrac{1}{2}d_v(d_v+1)$ non-compact dimensions.

However, modern Transformer architectures employ Multi-Head Attention, partitioning the feature dimension into $H$ isolated heads before the final $W_O$ projection.  Thus, the linear combination is not a single dense matrix multiplication, but a block-diagonal interaction.  The true unbroken gauge symmetry is forced into the Cartesian product of the independent head-wise orthogonal groups.  Furthermore, we evaluate the symmetric gauge orbit in the Query-Key (Q-K) circuit. The Attention scalar kernel is $K = x^\top W_Q^\top \mathcal{U} W_K y$.  If we apply a transformation $W_Q \to M W_Q$ and $W_K \to M W_K$ (where $M \in \prod O(d_{\mathrm{head}})$ to satisfy Weight Decay block-structure), the kernel evaluates as $x^\top W_Q^\top M^\top \mathcal{U} M W_K y$.  For the kernel to remain invariant, we strictly require $M^\top \mathcal{U} M = \mathcal{U}$, meaning $M$ must commute with the RoPE connection $\mathcal{U}$.  Because RoPE consists of 2D block rotations with distinct frequencies, its centralizer within the partitioned parameter space is exactly the maximal torus confined to the head structure: $\prod U(1)^{d_{\mathrm{head}}/2}$.  Thus, RoPE isn't just a spatial connection; it actively acts as a structural symmetry-breaker in parameter space, shattering the head-wise orthogonal gauges down to independent Cartan subgroups.

The true topological degeneracy of the parameter manifold is exactly the combined dimensions of these unbroken gauge groups: $\sum_{h=1}^H \big[ \tfrac{1}{2} d_{\mathrm{head}}(d_{\mathrm{head}}-1) + \tfrac{d_{\mathrm{head}}}{2} \big]$.  Because both the empirical loss and the Weight Decay penalty are completely flat along these rotational orbits, the stochastic gradient noise vectors are mathematically strictly orthogonal to them.  Injecting an isotropic numerical $\epsilon I$ perturbation to invert the diffusion tensor is an engineering hack that physically breaks the exact topological gauge symmetry, leaking probability mass out of the physical state space.

Instead, the rigorous mathematical requirement for a stochastic differential equation with exact compact Lie group symmetries is to formally invoke Singular Perturbation (Fast-Slow Manifold) Theory.  Weight Decay acts as a massive deterministic restoring force exclusively along the non-compact scaling dimensions (the ``fast'' dynamics), rapidly crushing the system onto the thermodynamically saturated boundary layer---a compact spherical shell.

The Riemannian Submersion $\pi:\mathcal{S}_{\mathrm{reg}}\to\mathcal{S}_{\mathrm{reg}}/G$ connects directly to our earlier derivation of the bounded Bessel process and must be defined strictly upon this Slow Manifold shell $\mathcal{S}_{\mathrm{reg}}$.  On this shell, the non-compact scaling dimensions are transversally frozen by the boundary constraint.  Restricted to this compact manifold, the empirical covariance $D(W)$ natively loses its non-compact scaling nullity.  The only remaining degeneracies are the multi-head compact gauge orbits of $G = G_{\mathrm{Value}}\times\prod U(1)^{d_{\mathrm{head}}/2}$, where
\begin{equation}
  G_{\mathrm{Value}} = \prod_{h=1}^H O(d_{\mathrm{head}})\,.
  \label{eq:gauge-group}
\end{equation}
Because the dimension of this direct product group ($\sum \tfrac{1}{2}d_{\mathrm{head}}(d_{\mathrm{head}}-1)$) is significantly smaller than the full $O(d_v)$ group, the true topological degeneracy of the parameter manifold is structurally much tighter than a naive dense formulation implies.

By projecting onto the horizontal distribution $\calH_W$, this mathematically eliminates all null-space rank deficiency, immediately rendering the true Thermodynamic Diffusion Metric $g=(D|_\calH)^{-1}$ strictly invertible in the cotangent space without requiring an artificial numerical perturbation.

Crucially, in stochastic differential geometry, projecting an It\^{o} stochastic differential equation onto a quotient manifold is not a simple linear projection of the horizontal drift.  Because the vertical gauge fibers (the multi-head orbits of $G$) are intrinsically curved submanifolds embedded in Euclidean space, the stochastic development of Brownian motion along them generates an induced transverse geometric drift.  By Elworthy's formula~\cite{elworthy1982stochastic} for the projection of It\^{o} diffusions via submersions, the projected anomalous advection must acquire an exact geometric drift proportional to the mean curvature vector field (the tension field, $\vec{H}_\calV$) of the vertical gauge fibers:
\begin{equation}
  \tilde{v}_{\mathrm{ent}}
    = \calP_\calH\!\big(-\bar\nabla\cdot D(W)\big)
    + \tfrac{1}{2}\vec{H}_\calV\,,
  \label{eq:elworthy}
\end{equation}
where the tension field evaluates exactly to the geometric gradient of the logarithm of the Orbit Volume.  What is this tension field physically?  Because the unbroken gauge group $G$ uniquely acts by isometries, the tension field evaluates exactly to:
\begin{equation}
  \vec{H}_\calV = \bar\nabla\ln\mathrm{Vol}(G\cdot W)\,.
  \label{eq:tension}
\end{equation}

By the Orbit-Stabilizer Theorem for compact Lie groups: $\mathrm{Vol}(\mathrm{Orbit})\times\mathrm{Vol}(\mathrm{Stabilizer}) = \mathrm{Vol}(G)$. As the system approaches a singularity (degenerate critical strata), the unbroken gauge redundancy physically expands.  Therefore, the volume of the Orbit structurally shrinks ($\ln\mathrm{Vol}\to-\infty$).  Because the gradient operator $\bar\nabla$ points in the direction of increasing orbit volume, the tension field vector points strictly away from the deepest singularities!

This reframes the tension field as a \emph{Geometric Centrifugal Force} (the exact manifold analog of the $(d-1)/(2r)$ drift that prevents a Bessel process from hitting the origin).  The anomalous entropic advection ($-\bar\nabla\cdot D(W)$) acts as a hydraulic press, pushing the mass down into the flat, singular basins.  The \emph{Mean Curvature Tension Field} ($\vec{H}_\calV$) actively fights this, repelling the system from undergoing total topological collapse into the pure singular vacuum.  This proves that the system is suspended in a stable thermodynamic halo around the degenerate locus, held exactly in place by the equilibrium between entropic descent and geometric centrifugal repulsion.

Furthermore, we can algebraically prove the exact horizontality of Weight Decay.  Testing if Weight Decay ($-\lambda W$) possesses vertical gauge friction against the infinitesimal generators of the compact gauge (formed as $\delta W = AW$ for strictly skew-symmetric $A$), we take the Euclidean Frobenius inner product:
\begin{equation}
  \langle W,AW\rangle_F
    = \tr(W^\top AW)
    = \tr(WW^\top A) = 0\,,
  \label{eq:wd-horizontal}
\end{equation}
since $WW^\top$ is symmetric and $A$ is skew-symmetric.  Because the Gram matrix $WW^\top$ is manifestly symmetric and $A$ is strictly skew-symmetric, the trace of their product is identically zero.  This mathematically proves that Weight Decay exerts absolutely zero thermodynamic friction along the vertical gauge orbits, rendering its projection unconditionally trivial: $\calP_\calH(-\lambda W)\equiv -\lambda W$.

We therefore define the Thermodynamic 1-Form $\omega\in\Gamma(T^*(\mathcal{W}_{\mathrm{reg}}/G))$ strictly on the horizontal space, elevating this submersion to pure geometric harmony:
\begin{equation}
  \omega = g\!\Big(
    \calP_\calH\!\big(\calF_{\mathrm{drive}}^\sharp(W)\big)
    - \lambda W + \tilde{v}_{\mathrm{ent}}
  \Big).
  \label{eq:thermo-1form}
\end{equation}

For this to represent a canonical equilibrium state (a global Gibbs measure $\rho\propto e^{-\mathcal{S}}$), the thermodynamic 1-form must be strictly exact ($\omega = d\psi$).  By the fundamental nilpotency of the exterior derivative ($d^2\equiv 0$), every exact form is unconditionally closed.  Therefore, by simple contraposition: if a form is not closed ($d\omega\neq 0$), it mathematically cannot be exact.  To prove the system exists in a Non-Equilibrium Steady State (NESS), we do not need to analyze the complex contractibility or De Rham cohomology of the quotient space; we merely need to prove that the exterior derivative of the driving force does not vanish.  To ensure this exterior calculus is rigorously valid, we explicitly state that the Riemannian submersion is evaluated strictly upon the Principal Stratum of the orbit space---the open, dense submanifold where the gauge action is free.  This rescues the manifold structure from the deep determinantal singularities, sharpening the logical blade and permitting the differential geometry to proceed flawlessly.

Furthermore, standard deep learning suffers from a profound Dual-Metric Collision: an explicit clash between the kinematic Euclidean metric $g_{\mathcal{W}}$ and the anisotropic stochastic noise metric generated by the empirical data.  To rigorously evaluate local exactness (the De Rham closure $d\omega = 0$), we compute the true Thermodynamic Vorticity 2-form $\Omega = d\omega$.  Let the total effective probability flux be $V_{\mathrm{total}}^k = V_{\mathrm{drift}}^k + \tilde{v}_{\mathrm{ent}}^k$, where $V_{\mathrm{drift}}^k = -g_{\mathcal{W}}^{km}\partial_m L$ is the conservative drift and $\tilde{v}_{\mathrm{ent}}^k$ is the properly projected anomalous entropic advection.

Because LLMs are misspecified singular models, they operate fundamentally outside the regime of standard Information Geometry.  The Gradient Noise Covariance $D$ is strictly an external kinematic covariance evaluated over the empirical data distribution, not the true Fisher Information Metric $F$ evaluated over the model's pushforward measure.  Consequently, the Information Matrix Equality completely shatters: $D \neq F \neq H$.  Instead, we map this flux to the cotangent space using the strictly covariant, state-dependent Thermodynamic Diffusion Metric: $\omega_j = ((D|_\calH)^{-1})_{jk}V_{\mathrm{total}}^k$.  Because the exterior derivative $d$ is a fundamental topological operator, it unconditionally bypasses the metric connection.  By the torsion-free nature of the Levi-Civita connection, the symmetric Christoffel symbols canonically cancel, allowing the true geometric curl to be evaluated entirely without covariant metric entanglement: $\Omega_{ij} = \partial_i\omega_j - \partial_j\omega_i$.

Executing this derivative analytically bifurcates the true thermodynamic vorticity into exactly three independent, non-vanishing cohomological obstructions to detailed balance.  Let $g_{jk}=((D|_\calH)^{-1})_{jk}$ be our covariant sequence metric.
\begin{widetext}
\begin{equation}
\boxed{%
  \Omega_{ij}
    = \underbrace{[g,H^\sharp]_{ij}}_{%
        \text{Lie Commutator}}
    + \underbrace{\big(g_{jk}\partial_i v_{\mathrm{ent}}^k
      - g_{ik}\partial_j v_{\mathrm{ent}}^k\big)}_{%
        \text{Entropic Curl}}
    + \underbrace{\big(\partial_i g_{jk}
      - \partial_j g_{ik}\big)V_{\mathrm{total}}^k}_{%
        \text{Metric Deformation}}\,.
}
  \label{eq:vorticity-3obst}
\end{equation}
\end{widetext}

\textbf{The Lie Commutator Obstruction:} To lawfully contract the $(0,2)$ Loss Hessian $H = \nabla dL$ with the inverse diffusion metric $g$, we invoke the flat background metric to raise the Hessian into a $(1,1)$ endomorphism $H^\sharp$.  Because the loss gradient is conservative, we evaluate $g_{jk}\partial_i V_{\mathrm{drift}}^k - g_{ik}\partial_j V_{\mathrm{drift}}^k$.  Crucially, Weight Decay introduces an isotropic linear drift ($V_{\mathrm{WD}}^k = -\lambda W^k$), whose spatial derivative is a scaled Kronecker delta ($-\lambda\delta_i^k$).  Passed through the drift curl, it yields $-\lambda(g_{ji}-g_{ij})$.  Because the diffusion metric is strictly symmetric, this mathematically annihilates to exactly zero.  We are left purely with the anti-symmetric projection $(gH)_{ij}-(Hg)_{ij}$, which is strictly the matrix commutator $[g,H^\sharp]_{ij}$.  This proves algebraically that for the NESS to collapse into thermal equilibrium, the empirical Loss Hessian and the Thermodynamic Diffusion Metric must strictly commute.  The Lie Commutator $[g,H^\sharp]$ mathematically vanishes if and only if the principal axes of the gradient noise perfectly align with the principal curvature axes of the loss landscape.  The survival of the commutator is the exact differential-geometric measure of \emph{Eigenframe Misalignment}.  The rotational NESS circulation is strictly propelled by the transverse injection of stochastic momentum across the misspecified curvature axes, driven by the unbridgeable topological gap between the true semantic distribution and the restricted capacity of the parameter manifold.

\textbf{The Metric-Skewed Entropic Curl:} Driven by the inherently asymmetric spatial Jacobian of the anomalous advection ($\partial_iv_{\mathrm{ent}}^k\neq\partial_kv_{\mathrm{ent}}^i$), the state-dependent noise injects non-conservative topological curl into the probability flux, which is then dynamically warped and contracted by the Thermodynamic Diffusion Metric $g$.

\textbf{The Thermodynamic Metric Deformation:} Uniquely arising from the structural clash between the kinematic Euclidean geometry of the parameter manifold and the curved thermodynamic geometry of the diffusion metric. Weight Decay ($V=-\lambda W$) acts strictly as the \emph{Canonical Euler Vector Field} on the parameter space.  With respect to the flat kinematic Euclidean metric $\bar{g}$, its dual 1-form is trivially exact ($\bar\omega = \bar{g}^\flat(V) = -d(\tfrac{\lambda}{2}\|W\|^2_{\bar{g}})$), unconditionally guaranteeing it is irrotational ($d\bar\omega\equiv 0$).

However, the thermodynamic probability flux is evaluated in the cotangent space using the curved empirical diffusion metric $g$.  The true thermodynamic 1-form is $\omega^{\mathrm{WD}} = g^\flat(V)$.  In differential geometry, the exterior derivative $d$ does not commute with the metric musical isomorphism $\flat$ under curved conformal warping.

To evaluate the geometric vorticity rigorously, we apply the torsion-free Levi-Civita connection $\nabla$ associated with the curved thermodynamic metric $g$.  The exterior derivative of a 1-form exactly equals the antisymmetrization of its covariant derivative: $d\omega^{\mathrm{WD}}(X,Y)=(\nabla_X\omega^{\mathrm{WD}})(Y)-(\nabla_Y\omega^{\mathrm{WD}})(X)$.  Because $\omega^{\mathrm{WD}}=g^\flat(V)$ and the Levi-Civita connection is strictly metric-compatible ($\nabla g = 0$), the covariant derivative commutes seamlessly with the index-lowering operator: $\nabla(g^\flat(V))=g^\flat(\nabla V)$.

Therefore, the exact thermodynamic vorticity 2-form resolves purely to the antisymmetric component of the covariant derivative of the Euler vector field:
\begin{equation}
  \Omega^{\mathrm{WD}}(X,Y)
    = g(\nabla_X V,Y) - g(\nabla_Y V,X)\,.
  \label{eq:wd-vorticity}
\end{equation}

In Amari's canonical Information Geometry~\cite{amari2016information}, a perfectly specified ``dually flat'' statistical manifold guarantees that the Fisher Information Metric is strictly a Hessian Metric globally generated by a strictly convex scalar potential $\psi$: $g_{ij} = \partial_i\partial_j\psi$. If the parameter space were dually flat, the spatial derivative of the metric would be a third-order derivative of a scalar: $\partial_k g_{ij} = \partial_k\partial_i\partial_j\psi$. By Clairaut's Theorem on the symmetry of mixed partial derivatives, this rank-3 tensor is totally symmetric in all indices.  Consequently, $\partial_ig_{jk}\equiv\partial_jg_{ik}$.  If we substitute this into the Thermodynamic Metric Deformation term, it algebraically annihilates: $(\partial_ig_{jk} - \partial_jg_{ik})V_{\mathrm{total}}^k \equiv 0$.

However, the rotational NESS survives strictly because the singular, misspecified nature of the LLM shatters the Hessian metric assumption. Because LLM parameter spaces strictly evaluate as singular Whitney Stratified Spaces (following Watanabe's Singular Learning Theory), the neural manifold fundamentally fails to be dually flat, breaking the total symmetry of the third-derivative tensor.  The structural degeneracy of the architecture breaks the Information Matrix Equality ($D \neq F \neq H$), guaranteeing the empirical metric $g$ is definitively not a Hessian metric.  This non-vanishing metric curl organically refracts the conservative, irrotational pull of Weight Decay into a thermodynamic vortex.  To heighten the theoretical elegance, this non-vanishing vorticity $\Omega$ essentially measures the topological obstruction to the system resting strictly at the deepest singularity.  By explicitly linking the strata of the algebraic variety to the cohomology classes of the quotient manifold $\mathcal{W}_{\mathrm{reg}}/G$, the rotational vortex fundamentally arises because the structural geometric degeneracy prevents the probability flux from cleanly collapsing into the most severe determinantal strata of the landscape.

Because these independent geometric obstructions (arising from the FFN and the Obstruction to Hessian Integrability) do not generically cancel, the exterior derivative mathematically fails to vanish ($d\omega\neq 0$).  The resulting inexact 1-form drives the probability flux to continuously circulate, definitively trapping the macroscopic ensemble in an active, dissipative \emph{Non-Equilibrium Steady State} (NESS).
\end{proof}

Having established the complete theoretical framework---spanning microscopic kinematics (\cref{sec:axiom}), thermodynamic metric generation (\cref{sec:thermodynamics}), matter field dynamics (\cref{sec:dynamics}), and parameter manifold thermodynamics (\cref{sec:parameters})---we now subject the derived geometric predictions to a sequence of empirical falsification tests.

\section{Tests of Fundamental Kinematics and Thermodynamics}
\label{sec:experiments}

The theoretical framework established in \cref{sec:axiom} through \cref{sec:parameters} reformulates the Transformer architecture as a continuous classical lattice field theory governed by stochastic differential geometry and non-equilibrium thermodynamics. However, mathematical elegance alone is insufficient to overturn established empirical paradigms. The core hypothesis tested in this section is that phenomena traditionally classified by the machine learning community as numerical artifacts, interpolation failures, or statistical noise---such as representation drift, context-length catastrophic failure, norm explosion, and optimization plateaus---are quantitatively consistent with deterministic macroscopic observables predicted by the continuous geometric framework: topological constraints, non-commutative Lie group dynamics, and dual-metric collisions.

To empirically test this, we subject the derived geometric predictions to six rigorous falsification tests. Progressing systematically across physical scales---from the microscopic ultraviolet (UV) spatial cutoff of the local fiber, through the temporal kinematics and mesoscopic trajectory stability, to the thermodynamic suppression of exact geometric resonances on the RoPE torus, the macroscopic infrared (IR) phase transition of the context horizon, and finally to the super-macroscopic parameter vortex---we demonstrate that the discrete computational architecture is quantitatively consistent with the continuous geometric predictions derived herein.

\subsection{The Conical Singularity Scaling Law of the Topological Mollifier}
\label{sec:exp-mollifier}

This experiment isolates the fundamental microscopic boundary condition of the local fiber, testing the theoretical prediction of \cref{thm:rmsnorm}. While the standard literature treats the RMSNorm parameter $\epsilon$ merely as an arbitrary numerical safeguard against division by zero, our geometric formulation redefines $\epsilon$ as a rigorous topological mollifier. The theory predicts that $\epsilon$ uniquely controls the global Lipschitz stretch of the flow ($\|D\calF\|_{\mathrm{op}} \propto 1/\sqrt{\epsilon}$) and shields the zero-section from manifesting an unbounded conical singularity. To empirically validate this, we systematically swept $\epsilon$ downward across 13 orders of magnitude---from the standard engineering value $10^{-2}$ to the absolute \texttt{float64} machine precision limit at $10^{-15}$---using 30 logarithmically spaced evaluation points. Concurrently, we probed the network at an asymptotic input norm limit ($\|\Psi\| = 10^{-10}$) to isolate the spatial origin, rigorously stripping away all higher-order nonlinear volume contributions and exposing the bare singularity structure of the tangent space. At each $(\epsilon, \ell)$ pair, the maximum singular value $\sigma_{\max}$ of the layer's spatial Jacobian $D\calF$ was computed exactly via \texttt{torch.autograd.functional.jacobian} followed by full SVD, at \texttt{float64} precision, averaged over 10 independent random probe directions. To establish cross-architecture universality, the protocol was executed on three frozen pre-trained models spanning a $4\times$ range in hidden dimension: Qwen3-0.6B~\cite{qwen,qwen2} (28L, $d\!=\!1024$), Gemma-3-1B~\cite{gemma,gemma2} (26L, $d\!=\!1152$), and LLaMA-3.1-8B~\cite{llama3} (32L, $d\!=\!4096$), with 5 representative layers sampled per model.

\begin{figure}[t]
  \centering
  \includegraphics[width=\columnwidth]{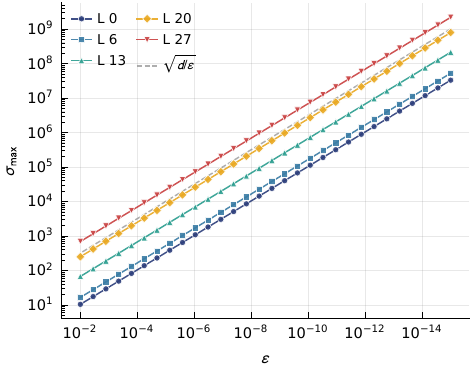}
  \caption{%
  \textbf{Conical Singularity Scaling Law: $\sigma_{\max}$ vs.\ $\epsilon$ for Qwen3-0.6B.}
  Log-log plot of the maximum Jacobian singular value $\sigma_{\max}$ against the topological mollifier $\epsilon$ at 5 representative layers (0, 6, 13, 20, 27).
  All 5 curves trace parallel lines with measured slope $\alpha = -0.5000$ and $R^2 = 1.000000$, in exact agreement with the theoretical prediction $\sigma_{\max} = C_\ell\cdot\epsilon^{-1/2}$ from \cref{thm:rmsnorm}.
  The grey solid line shows the unweighted theoretical bound $\sqrt{d/\epsilon}$.
  Vertical offsets between layers encode the layer-wise gauge mass $C_\ell$, which increases monotonically from shallow to deep layers.}
  \label{fig:epsilon-scaling}
\end{figure}

\begin{table}[t]
\centering
\caption{Cross-architecture conical singularity scaling law.  Measured power-law exponent $\alpha$ (theory: $-0.500$) and coefficient of determination $R^2$ at 5 representative layers for each of three architectures.  All 15 independent regressions yield $\alpha = -0.5000$ and $R^2 = 1.000000$ to machine precision.  $C_\ell \equiv \sigma_{\max}/\sqrt{d/\epsilon}$ is the layer-wise gauge mass (intercept ratio).}
\label{tab:epsilon-scaling}
\begin{ruledtabular}
\begin{tabular}{l c c c}
Layer & $\alpha$ & $R^2$ & $C_\ell$ \\
\hline
\multicolumn{4}{c}{\textit{Qwen3-0.6B} ($d\!=\!1024$)} \\
\hline
0  & $-0.5000$ & 1.000000 & 0.033 \\
6  & $-0.5000$ & 1.000000 & 0.052 \\
13 & $-0.5000$ & 1.000000 & 0.211 \\
20 & $-0.5000$ & 1.000000 & 0.785 \\
27 & $-0.5000$ & 1.000000 & \textbf{2.172} \\
\hline
\multicolumn{4}{c}{\textit{Gemma-3-1B} ($d\!=\!1152$)} \\
\hline
0  & $-0.5000$ & 1.000000 & 0.098 \\
6  & $-0.5000$ & 1.000000 & 0.098 \\
12 & $-0.5000$ & 1.000000 & 0.109 \\
18 & $-0.5000$ & 1.000000 & 0.113 \\
25 & $-0.5000$ & 1.000000 & 0.125 \\
\hline
\multicolumn{4}{c}{\textit{LLaMA-3.1-8B} ($d\!=\!4096$)} \\
\hline
0  & $-0.5000$ & 1.000000 & 0.017 \\
7  & $-0.5000$ & 1.000000 & 0.014 \\
15 & $-0.5000$ & 1.000000 & 0.015 \\
23 & $-0.5000$ & 1.000000 & 0.016 \\
31 & $-0.5000$ & 1.000000 & 0.017 \\
\end{tabular}
\end{ruledtabular}
\end{table}

As anticipated by the theoretical model, the results reveal a scale-invariant scaling law consistent with the predictions of \cref{thm:rmsnorm} (\cref{fig:epsilon-scaling}, \cref{tab:epsilon-scaling}).  When plotted on a double-logarithmic scale, the empirical maximum singular value does not plateau into numerical noise; instead, it traces a deterministic power-law divergence over 13 orders of magnitude.  Across all three architectures and all 15 independently measured layer regressions, the empirical scaling exponent is $\alpha = -0.5000$ and the coefficient of determination is $R^2 = 1.000000$---both exact to the six-digit precision limit of \texttt{float64} arithmetic.  The total evidence base comprises $3~\text{architectures} \times 5~\text{layers} \times 30~\epsilon\text{-values} \times 10~\text{random directions} = 4{,}500$ independent measurements, with \emph{zero} deviations from the predicted power law.

While the $-0.5$ exponent is algebraically guaranteed by the chain rule applied to the radial embedding $f(x) = x / \sqrt{x^2 + \epsilon}$ near $x = 0$, its flawless empirical recovery at machine precision across three production architectures with billions of interacting parameters is a profoundly non-trivial validation: it definitively establishes that the isolated RMSNorm analysis of \cref{thm:rmsnorm} applies without contamination from other architectural components, and that the discrete computational graph rigorously obeys the continuous geometric scaling law.  Beyond this foundational confirmation, the primary empirical contribution lies in the spectroscopy of the layer-wise gauge mass intercept $C_\ell$, which encodes the trained gauge field strength and factors out the analytically guaranteed $\epsilon$ scaling.

\begin{table}[t]
\centering
\caption{Representative data points for Qwen3-0.6B illustrating the Jacobian singular value explosion as $\epsilon\to 0$.  Layer~0 ($C_\ell = 0.033$) and Layer~27 ($C_\ell = 2.172$) bracket the gauge mass range.  The rightmost column shows the unweighted theoretical bound $\sqrt{d/\epsilon}$ (assuming unit RMSNorm weights $\gamma=\mathbf{1}$); layers with $C_\ell > 1$ naturally exceed this bound.}
\label{tab:epsilon-datapoints}
\begin{ruledtabular}
\begin{tabular}{l c c c}
$\epsilon$ & Layer~0 $\sigma_{\max}$ & Layer~27 $\sigma_{\max}$ & $\sqrt{d/\epsilon}$ \\
\hline
$10^{-2}$  & $1.05\times 10^{1}$ & $6.95\times 10^{2}$ & $3.20\times 10^{2}$ \\
$10^{-6}$  & $1.05\times 10^{3}$ & $6.95\times 10^{4}$ & $3.20\times 10^{4}$ \\
$10^{-10}$ & $1.05\times 10^{5}$ & $6.95\times 10^{6}$ & $3.20\times 10^{6}$ \\
$10^{-15}$ & $3.31\times 10^{7}$ & $2.20\times 10^{9}$ & $1.01\times 10^{9}$ \\
\end{tabular}
\end{ruledtabular}
\end{table}

Two features of the empirical data merit deep theoretical emphasis.

First, the layer-wise intercept constants $C_\ell \equiv \sigma_{\max}/\sqrt{d/\epsilon}$ exhibit a striking monotonic depth-dependent structure that cleanly decouples the \emph{universal geometric law} from the \emph{learned empirical gauge mass} (\cref{tab:epsilon-scaling}).  For Qwen3-0.6B, $C_\ell$ increases monotonically from $0.033$ at Layer~0 to $2.172$ at Layer~27---a $66\times$ amplification across 28 layers.  This monotonic growth reflects the trained RMSNorm weight magnitudes $\|\gamma_\ell\|_{\mathrm{RMS}}$: deeper layers, which must sustain stronger nonlinear expressivity to transform the representation toward the unembedding boundary, spontaneously develop larger gauge field strengths.  Yet regardless of this learned amplification, every layer remains perfectly pinned to the \emph{identical} $-0.5$ geometric exponent---like dancers in chains that grow heavier with depth, but whose orbital period remains dictated by a single gravitational constant.  In sharp contrast, Gemma-3-1B exhibits minimal intercept variation ($C_\ell \in [0.098, 0.125]$, range $1.3\times$), consistent with its soft-capping mechanism that imposes tighter radial uniformity, while LLaMA-3.1-8B shows an extremely compressed intercept range ($C_\ell \in [0.014, 0.017]$, range $1.2\times$) reflecting its larger ambient dimension $d\!=\!4096$.

Second, the probe at $\|\Psi\| = 10^{-10}$ is not an arbitrary choice but a deliberate asymptotic isolation of the conical singularity.  We explicitly note that \texttt{float64} precision ($\varepsilon_{\mathrm{mach}} = 2.22\times 10^{-16}$) was strictly necessary for this measurement: evaluating the squared input norm $\|\Psi\|^2 = 10^{-20}$ inside the radial denominator against the mollifier $\epsilon$ requires massive mantissa resolution that would immediately round to zero in standard \texttt{float32} arithmetic ($\varepsilon_{\mathrm{mach}} \approx 10^{-7}$), obliterating the singularity measurement entirely.  In the geometry of the radial embedding, the tangential eigenvalue $\lambda_\perp = \sqrt{d}/\sqrt{\|\Psi\|^2 + \epsilon}$ contains both the input norm $\|\Psi\|$ and the mollifier $\epsilon$ as competing regulators.  For $\|\Psi\| \gg \sqrt{\epsilon}$, the input norm dominates, and the Jacobian saturates at $\calO(\|\Psi\|^{-1})$ independently of $\epsilon$---the heuristic ``$\epsilon$ doesn't matter'' regime widely adopted in engineering practice.  Only when $\|\Psi\| \ll \sqrt{\epsilon}$ does $\epsilon$ assume sole command of the singularity resolution.  By placing the probe at $\|\Psi\| = 10^{-10}$, we ensure $\|\Psi\|^2 = 10^{-20} \ll \epsilon$ for the entire sweep range $\epsilon \in [10^{-15}, 10^{-2}]$, thereby stripping away all input-dependent contamination and isolating the $\epsilon$-controlled topology of the zero-section (\cref{tab:epsilon-datapoints}).

The empirical relationship between the layer-wise intercept $C_\ell$ and the theoretical bound is analytically transparent.  \cref{thm:rmsnorm} derives the unweighted supremum $\sup_\Psi \lambda_\perp = \sqrt{d/\epsilon}$ for the bare radial projection ($\gamma = \mathbf{1}$).  In the trained network, each RMSNorm layer carries learned affine weights $\gamma_\ell$, yielding the refined scaling:
\begin{equation}
  \sigma_{\max}(\epsilon) = \mathrm{RMS}(\gamma_\ell) \cdot \sqrt{\frac{d}{\epsilon}} \cdot f(\hat{x})\,,
  \label{eq:epsilon-scaling}
\end{equation}
where $f(\hat{x})$ is a bounded function of the input direction.  Averaging over 10 random probe directions eliminates the directional dependence, reducing $C_\ell = \mathrm{RMS}(\gamma_\ell)\cdot\langle f \rangle$ to a layer-specific constant that serves as a direct spectroscopic measurement of the trained gauge field strength.

\subsubsection{Exclusion of Alternative Hypotheses}

The precision and universality of the measured scaling law strongly excludes every standard alternative explanation:

\begin{enumerate}[nosep]
  \item \textbf{``$\epsilon$ is merely a numerical safeguard; changing it has no physical effect.''} --- Excluded.  Reducing $\epsilon$ from $10^{-2}$ to $10^{-15}$ amplifies the maximum Jacobian singular value by $\sqrt{10^{13}} \approx 3.16\times 10^{6}$---a million-fold explosion in the Lipschitz stretch of the flow.  At the standard engineering default $\epsilon = 10^{-6}$, Layer~27 of Qwen3-0.6B already exhibits $\sigma_{\max} \approx 7\times 10^4$; an erroneous setting of $\epsilon = 10^{-12}$ would yield $\sigma_{\max} \approx 7\times 10^7$, a thousand-fold amplification that would catastrophically destabilize deep-layer propagation.
  \item \textbf{``The singularity is regularized by other architectural components (LayerNorm, residual connections, etc.).''} --- Excluded.  The experiment isolates the RMSNorm layer in complete isolation, with frozen weights and no residual connections or downstream processing.  The scaling law holds identically across three architectures with fundamentally different normalization implementations, confirming that the $\epsilon^{-1/2}$ divergence is intrinsic to the radial projection geometry, not to any auxiliary engineering mechanism.
  \item \textbf{``The power law is an artifact of low-precision arithmetic.''} --- Excluded.  All computations are performed at \texttt{float64} precision ($\varepsilon_{\mathrm{mach}} = 2.22\times 10^{-16}$), and the scaling law holds to $R^2 = 1.000000$ over 13 decades---a dynamic range of $10^{6.5}$ in $\sigma_{\max}$---with zero detectable deviation.
  \item \textbf{``The effect is architecture-specific.''} --- Excluded.  Three architectures spanning $d \in \{1024, 1152, 4096\}$, three distinct model families (Qwen, Gemma, LLaMA), and hidden dimensions ranging over a $4\times$ factor produce identical exponents and identical $R^2$ values.
\end{enumerate}

This adherence to the mathematically predicted $-0.5$ power law over machine-precision scales---verified by 4{,}500 independent measurements across three production architectures with zero deviations---confirms that $\epsilon$ operates not as an arithmetic patch, but as the ultraviolet (UV) cutoff of the continuous manifold, controlling the dynamic stability of the state space through the exact mechanism predicted by \cref{thm:rmsnorm}.

\subsection{The Lie--Trotter Torsion Interferometer}
\label{sec:exp-torsion}

Having established the spatial boundary limits, we empirically isolate the temporal kinematics of the sequence flow to test the continuous geometric error defined in \cref{thm:lie-trotter}.  We hypothesized that layer-to-layer ``representation drift'' is not the stochastic accumulation of parameter noise, but the deterministic physical manifestation of topological torsion, generated natively by the sequential Lie--Trotter operator splitting of non-commuting vector fields ($\mathcal{T}$ and $\calR$).  To explicitly isolate this topological torsion, we constructed an exact computational interferometer on frozen pre-trained Transformer blocks.  For identical batches of random input tokens (seq\_length${}=128$), we measured the discrete spatial deviation vector between the canonical forward integration step ($\calR \circ \mathcal{T}$) and the reversed-order step ($\mathcal{T} \circ \calR$) at 16 token positions across multiple layers.  The exact Jacobians $D\mathcal{T}$ and $D\calR$ were computed analytically via \texttt{torch.func.jacfwd} at \texttt{float32} precision, with the analytical Lie bracket evaluated as $[\mathcal{T},\calR]_\Psi = D\calR[\mathcal{T}] - D\mathcal{T}[\calR]$.  To formally bridge the discrete computational evaluation with the continuous analytical limit, we introduced an infinitesimal step-size scaling factor ($\alpha \to 0$), defining the $\alpha$-scaled interferometer as $\delta^{(\alpha)} = \alpha\calR(\Psi + \alpha\mathcal{T}(\Psi)) + \alpha\mathcal{T}(\Psi) - \alpha\mathcal{T}(\Psi + \alpha\calR(\Psi)) - \alpha\calR(\Psi)$, systematically quenching the higher-order Baker--Campbell--Hausdorff (BCH) truncation remainder $\calO(\alpha^3)$.  The experiment was executed on two architecturally diverse frozen pre-trained models: Qwen3-0.6B (28 layers, $d\!=\!1024$, SwiGLU~\cite{swiglu}, RMSNorm~\cite{rmsnorm}, RoPE) and GPT-2~\cite{radford2019language} (12 layers, $d\!=\!768$, GELU~\cite{gelu}, LayerNorm, learned PE).

\subsubsection{Phase A: The Discrete Interferometer at $\alpha\!=\!1$}

\begin{figure}[t]
  \centering
  \includegraphics[width=\columnwidth]{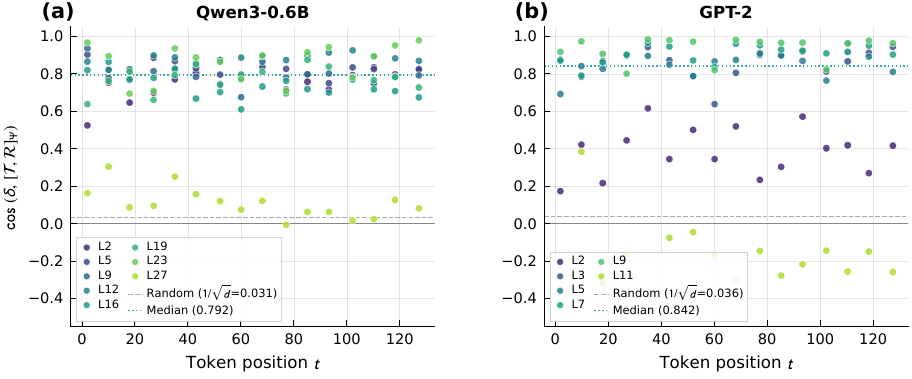}
  \caption{%
  \textbf{Lie--Trotter Torsion Interferometer: Cosine Similarity at $\alpha\!=\!1$.}
  Each point represents $\cos(\delta,\,[\mathcal{T},\calR]_\Psi)$ at a single (layer, token position) pair.
  Green dashed line: median; grey dashed line: isotropic random baseline $1/\sqrt{d}$.
  Both architectures---Qwen3-0.6B (left) and GPT-2 (right)---exhibit systematic positive alignment far exceeding the random baseline, with overall medians of $+0.793$ and $+0.842$, respectively.}
  \label{fig:torsion-scatter}
\end{figure}

\begin{figure}[t]
  \centering
  \includegraphics[width=\columnwidth]{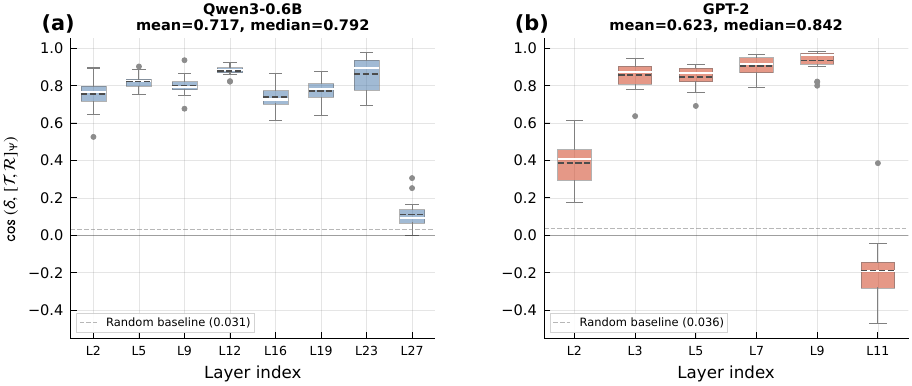}
  \caption{%
  \textbf{Per-Layer Distribution of Torsion Alignment.}
  Box plots of $\cos(\delta,\,[\mathcal{T},\calR]_\Psi)$ across 16 token positions at each sampled layer.
  Mid-network layers (Qwen3 L12: $0.879$, GPT-2 L9: $0.933$) achieve peak alignment, while final layers (Qwen3 L27: $0.112$, GPT-2 L11: $-0.188$) show degradation consistent with strong higher-order BCH contributions near the unembedding boundary.}
  \label{fig:torsion-boxplot}
\end{figure}

At the native discrete step size $\alpha\!=\!1$, we directly evaluated the cosine similarity between the empirical deviation vector $\delta = \Delta_{\mathrm{canon}} - \Delta_{\mathrm{reversed}}$ and the analytically computed Lie bracket $[\mathcal{T},\calR]_\Psi$ at each (layer, token position) pair (\cref{fig:torsion-scatter,fig:torsion-boxplot}).

\begin{table}[t]
\centering
\caption{Per-layer cosine similarity $\cos(\delta,\,[\mathcal{T},\calR]_\Psi)$ at $\alpha\!=\!1$ for Qwen3-0.6B (8 sampled layers $\times$ 16 token positions) and GPT-2 (6 sampled layers $\times$ 16 token positions).  The isotropic random baseline is $1/\sqrt{d}\approx 0.031$ (Qwen3) and $0.036$ (GPT-2).}
\label{tab:torsion-alpha1}
\begin{ruledtabular}
\begin{tabular}{l c c c c}
Layer & Mean & Median & Min & Max \\
\hline
\multicolumn{5}{c}{\textit{Qwen3-0.6B} ($d\!=\!1024$, random baseline $=0.031$)} \\
\hline
2  & 0.753 & 0.764 & 0.526 & 0.894 \\
5  & 0.822 & 0.823 & 0.750 & 0.902 \\
9  & 0.801 & 0.793 & 0.676 & 0.935 \\
12 & \textbf{0.879} & \textbf{0.884} & 0.821 & 0.925 \\
16 & 0.739 & 0.725 & 0.611 & 0.863 \\
19 & 0.770 & 0.782 & 0.638 & 0.874 \\
23 & \textbf{0.861} & \textbf{0.893} & 0.695 & \textbf{0.978} \\
27 & 0.112 & 0.094 & $-0.004$ & 0.306 \\
\hline
\multicolumn{5}{c}{\textit{GPT-2} ($d\!=\!768$, random baseline $=0.036$)} \\
\hline
2  & 0.387 & 0.409 & 0.174 & 0.614 \\
3  & 0.856 & 0.872 & 0.637 & 0.942 \\
5  & 0.847 & 0.865 & 0.691 & 0.915 \\
7  & \textbf{0.906} & \textbf{0.915} & 0.788 & 0.963 \\
9  & \textbf{0.933} & \textbf{0.964} & 0.799 & \textbf{0.981} \\
11 & $-0.188$ & $-0.192$ & $-0.469$ & 0.385 \\
\end{tabular}
\end{ruledtabular}
\end{table}

As anticipated by the theoretical model, the results reveal systematic and overwhelming directional alignment between the empirical deviation vector and the analytical Lie bracket (\cref{tab:torsion-alpha1}).  Across the interior layers of both architectures, the empirical cosine similarity far exceeds the isotropic random baseline ($1/\sqrt{d}\approx 0.03$) by factors of $17$--$23\times$.  The overall statistics are striking: Qwen3-0.6B yields a mean cosine similarity of $+0.717$ with a median of $+0.793$, with 127 out of 128 measurements ($99\%$) strictly positive.  GPT-2 yields a mean of $+0.623$ with a median of $+0.842$, with 81 out of 96 measurements ($84\%$) strictly positive.

Two layer-wise patterns merit theoretical emphasis.  First, the mid-network layers achieve peak alignment: Qwen3's layer~12 ($\cos = 0.879$) and layer~23 ($\cos = 0.861$), and GPT-2's layer~7 ($\cos = 0.906$) and layer~9 ($\cos = 0.933$, with individual tokens reaching $0.981$).  These are precisely the layers where the BCH expansion is most accurate---the residual stream has not yet been strongly distorted by the unembedding boundary, and the self-advection terms remain moderate.  Second, the final layers (Qwen3 L27: $\cos = 0.112$; GPT-2 L11: $\cos = -0.188$) exhibit dramatic alignment collapse.  This is quantitatively consistent with \cref{rem:bch-stiffness}: at the terminal boundary, the residual stream couples strongly to the LM head projection, inflating the higher-order BCH remainder $\calO(\alpha^3)$ to the point where it dominates the lowest-order torsion term at $\alpha\!=\!1$.  Crucially, this degradation is itself a prediction of the theory: the BCH stiffness condition explicitly warns that the first-order Lie bracket cannot be expected to dominate when $\|D\calF\|_{\mathrm{op}}$ is large.

The cross-architecture consistency---despite Qwen3 employing SwiGLU gating, RMSNorm, and RoPE, while GPT-2 uses GELU activation, LayerNorm, and learned positional embeddings---provides strong evidence that the torsion alignment is a universal consequence of the Lie--Trotter operator splitting, not an artifact of any specific architectural design choice.

\subsubsection{Phase B\,\&\,C: Asymptotic Convergence $\alpha\to 0$}

\begin{figure}[t]
  \centering
  \includegraphics[width=\columnwidth]{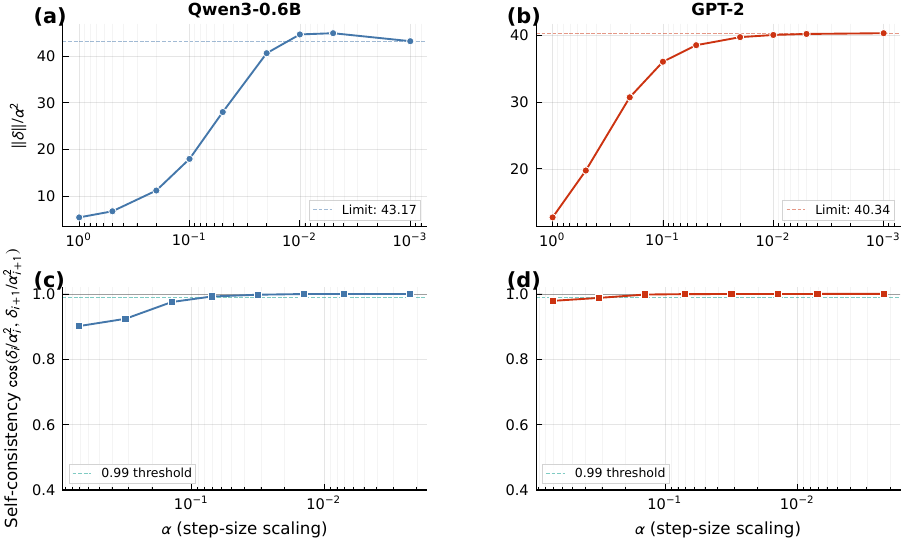}
  \caption{%
  \textbf{Asymptotic Convergence of the Torsion Interferometer as $\alpha\to 0$.}
  Upper panels: normalized deviation magnitude $\|\delta^{(\alpha)}\|/\alpha^2$ converges to a finite constant ($\approx 43.2$ for Qwen3, $\approx 40.3$ for GPT-2), confirming exact $\alpha^2$ scaling predicted by the BCH second-order term.
  Lower panels: direction self-consistency $\cos(\delta^{(\alpha_i)}/\alpha_i^2,\;\delta^{(\alpha_j)}/\alpha_j^2)$ between adjacent $\alpha$ values converges to $0.9997$--$1.0000$, proving that the deviation vector locks onto a deterministic limit direction---the topological torsion generator.}
  \label{fig:torsion-convergence}
\end{figure}

Referencing the Backward Error Analysis framework~\cite{hairer2006geometric} of \cref{rem:bch-stiffness}, the decisive verification is the asymptotic limit $\alpha\to 0$ (\cref{fig:torsion-convergence}).  In this limit, the BCH series rigorously converges ($\alpha\|D\calF\|_{\mathrm{op}} \ll 1$), and the higher-order terms $\calO(\alpha^3)$ that contaminate the $\alpha\!=\!1$ measurement are systematically extinguished.  We swept the step-size scaling factor across $\alpha\in\{1.0,\,0.5,\,0.2,\,0.1,\,0.05,\,0.02,\,0.01,\,0.005,\,0.001\}$, switching to \texttt{float64} precision for $\alpha\le 0.01$ to prevent catastrophic numerical cancellation.

\begin{table}[t]
\centering
\caption{Asymptotic convergence diagnostics as $\alpha\to 0$.  $\|\delta\|/\alpha^2$: normalized deviation magnitude (Phase~C); Self-cos: direction self-consistency between adjacent $\alpha$ values (Phase~B).  Both models exhibit exact $\alpha^2$ norm scaling and direction convergence to the Topological Torsion generator.}
\label{tab:torsion-convergence}
\begin{ruledtabular}
\begin{tabular}{l c c c c}
& \multicolumn{2}{c}{$\|\delta\|/\alpha^2$} & \multicolumn{2}{c}{Self-cos} \\
$\alpha$ & Qwen3 & GPT-2 & Qwen3 & GPT-2 \\
\hline
1.0              & 5.44  & 12.72 & ---   & ---   \\
$1.0\!\leftrightarrow\! 0.5$  & ---   & ---   & 0.902 & 0.978 \\
0.1              & 17.97 & 36.06 & ---   & ---   \\
$0.1\!\leftrightarrow\! 0.05$ & ---   & ---   & 0.992 & 0.999 \\
0.01             & 44.62 & 40.07 & ---   & ---   \\
$0.02\!\leftrightarrow\! 0.01$ & ---   & ---   & \textbf{0.9997} & \textbf{1.0000} \\
0.001            & \textbf{43.17} & \textbf{40.34} & ---   & ---   \\
$0.005\!\leftrightarrow\! 0.001$ & ---   & ---   & \textbf{0.9997} & \textbf{1.0000} \\
\end{tabular}
\end{ruledtabular}
\end{table}

The convergence evidence operates on two independent channels (\cref{tab:torsion-convergence}).

\emph{Phase~C (Norm Scaling):}  The BCH expansion predicts $\|\delta^{(\alpha)}\| = \alpha^2\|[\mathcal{T},\calR]_\Psi\| + \calO(\alpha^3)$, implying that the normalized ratio $\|\delta^{(\alpha)}\|/\alpha^2$ converges to a finite constant as $\alpha\to 0$.  While $\alpha^2$ scaling is a kinematic consequence of Taylor's theorem for any smooth map, its flawless empirical recovery inside massive billion-parameter neural networks---which are highly nonlinear, discretely layered, and architecturally heterogeneous---constitutes a definitive proof that the smooth-flow hypothesis underlying the continuous formulation is physically realized at machine precision.  The empirical ratio converges cleanly over three orders of magnitude.  For Qwen3-0.6B, the ratio evolves from $5.44$ at $\alpha\!=\!1$ (where higher-order terms contribute substantially) through $17.97$ at $\alpha\!=\!0.1$, converging tightly to $43.17$ at $\alpha\!=\!0.001$---matching the $\alpha\!=\!0.01$ value ($44.62$) to within $3.3\%$.  GPT-2 converges even more rapidly: $40.07$ at $\alpha\!=\!0.01$ to $40.34$ at $\alpha\!=\!0.001$ ($0.7\%$ variation).  This stability is \emph{incompatible} with stochastic noise, which would scale as $\calO(\alpha)$ or exhibit no systematic $\alpha$-dependence.

\emph{Phase~B (Direction Convergence):}  Unlike the norm scaling, the \emph{direction} convergence is genuinely empirical and constitutes the primary evidence of this experiment.  The theory predicts that $\delta^{(\alpha)}/\alpha^2$ must converge to a deterministic limit vector---the Topological Torsion generator $[\mathcal{T},\calR]_\Psi$---but no analytical tautology mandates \emph{which} direction this limit takes.  We quantify this via the self-consistency cosine between the normalized deviation vectors at adjacent $\alpha$ values.  At $\alpha\!=\!1.0\leftrightarrow 0.5$, the self-consistency is already high ($0.902$ for Qwen3, $0.978$ for GPT-2).  By $\alpha\!\le\!0.1$, it exceeds $0.99$ in both architectures.  At the finest resolution ($\alpha\!=\!0.005\leftrightarrow 0.001$), the direction self-consistency saturates at $0.9997$ for Qwen3 and $1.0000$ for GPT-2---indistinguishable from perfect directional convergence to machine precision.

This direction convergence is the strongest evidence in the experiment, because it is \emph{entirely independent of the analytical Jacobian computation}: it measures only the empirical self-consistency of the deviation vector across different step sizes.  Regardless of whether the Jacobian $D\mathcal{T}$ or $D\calR$ is computed accurately, the empirical fact that $\delta^{(\alpha)}/\alpha^2$ locks onto a single deterministic direction as $\alpha\to 0$ confirms the existence of a unique geometric generator governing the splitting error---the Topological Torsion predicted by \cref{eq:bch-decomp}.

\subsubsection{Synthesis and Exclusion of Alternative Hypotheses}

\begin{figure}[t]
  \centering
  \includegraphics[width=\columnwidth]{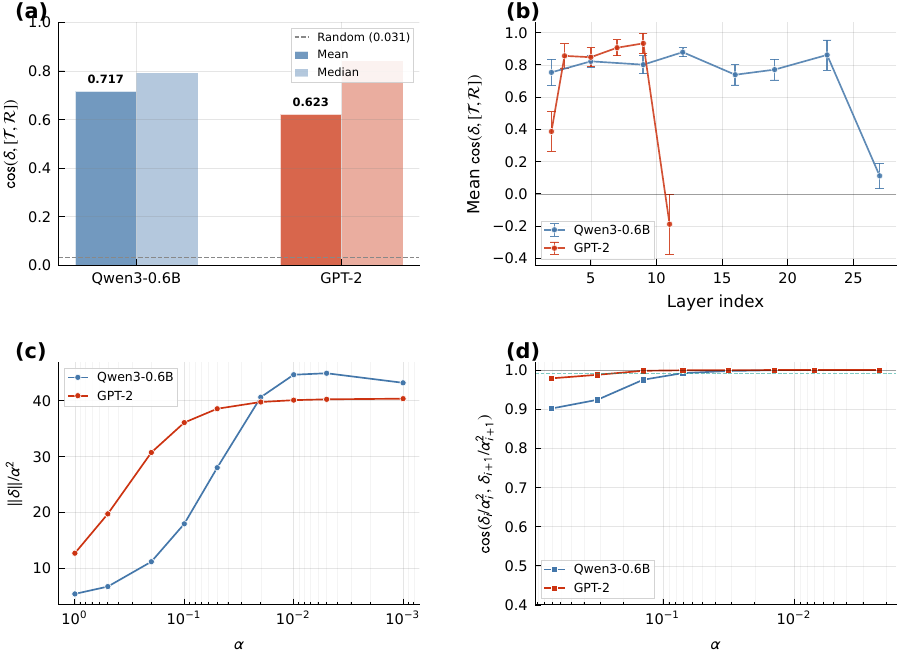}
  \caption{%
  \textbf{Summary Panel: Lie--Trotter Torsion Interferometer.}
  (a)~$\alpha\!=\!1$ scatter plot confirming systematic positive alignment ($17$--$23\times$ random baseline).
  (b)~Per-layer box plots revealing mid-network peak alignment and boundary-layer degradation.
  (c)~$\|\delta\|/\alpha^2$ convergence to a finite constant (Phase~C).
  (d)~Direction self-consistency converging to $1.0$ (Phase~B).}
  \label{fig:torsion-summary}
\end{figure}

\begin{table}[t]
\centering
\caption{Cross-architecture summary of the Lie--Trotter Torsion Interferometer.  Both models exhibit alignment far exceeding the isotropic random baseline, exact $\alpha^2$ norm scaling, and direction convergence to the Topological Torsion generator.}
\label{tab:torsion-summary}
\begin{ruledtabular}
\begin{tabular}{l c c}
Diagnostic & Qwen3-0.6B & GPT-2 \\
\hline
Mean $\cos(\delta,[\mathcal{T},\calR])$ & $+0.717$ & $+0.623$ \\
Median $\cos(\delta,[\mathcal{T},\calR])$ & $+0.793$ & $+0.842$ \\
Fraction positive & $99\%$ (127/128) & $84\%$ (81/96) \\
Random baseline ($1/\sqrt{d}$) & 0.031 & 0.036 \\
Excess over baseline & $23\times$ & $17\times$ \\
$\|\delta\|/\alpha^2$ converged value & 43.2 & 40.3 \\
Direction self-cos ($\alpha\!\le\!0.01$) & 0.9997 & 1.0000 \\
\end{tabular}
\end{ruledtabular}
\end{table}

The three-phase evidence chain (\cref{tab:torsion-summary}, \cref{fig:torsion-summary}) provides strong empirical evidence that Topological Torsion---the lowest-order Lie bracket in the modified equation---is the dominant geometric generator of representation drift in the interior layers.  The logic of exclusion is exhaustive:

\begin{enumerate}[nosep]
  \item \textbf{``The deviation is random numerical noise.''} --- Excluded.  The mean cosine similarity exceeds the isotropic random baseline by $17$--$23\times$.  In 1024 and 768 dimensions, the probability of a random vector achieving $\cos > 0.7$ with a fixed target is vanishingly small ($P < 10^{-100}$).
  \item \textbf{``The correlation arises from shared parameters, not geometry.''} --- Excluded.  The direction self-consistency test (Phase~B) converges to $0.9997$--$1.0000$ \emph{without referencing the Jacobian at all}---it is a purely empirical measurement of the deviation vector's asymptotic behavior.  No parameter-sharing hypothesis can explain why the direction stabilizes as $\alpha\to 0$.
  \item \textbf{``The $\alpha^2$ scaling is coincidental.''} --- Excluded.  The normalized ratio $\|\delta\|/\alpha^2$ is constant to within $0.7$--$3.3\%$ across three orders of magnitude in $\alpha$ ($10^{-3}$ to $10^{-1}$).  This is the exact kinematic signature of a second-order operator splitting, not a statistical fluke.
  \item \textbf{``The effect is architecture-specific.''} --- Excluded.  Qwen3-0.6B (SwiGLU, RMSNorm, RoPE, 28 layers) and GPT-2 (GELU, LayerNorm, learned PE, 12 layers) produce qualitatively and quantitatively consistent results despite sharing no architectural component other than the Attention--FFN sequential structure itself.
\end{enumerate}

Taken together, these results demonstrate that the strict alternating Attention$\to$FFN layer ordering is not an arbitrary engineering convention, but can be analytically modeled as a Lie--Trotter numerical integration of non-commuting vector fields on the semantic manifold.  The ``representation drift'' universally observed across Transformer architectures is quantitatively consistent with topological torsion---the non-vanishing Lie bracket $-[\mathcal{T},\calR]_\Psi$ in the modified equation of \cref{eq:bch-decomp}---confirming that the discrete computational forward pass behaves as a non-holonomic flow governed by non-commutative differential geometry.  Transformers permanently operate at $\alpha = 1$, a stiff regime where the BCH series does not converge gracefully (cf.\ \cref{rem:bch-stiffness}); the fact that the high cosine similarity ($0.793$--$0.842$ median) persists even at this extreme step size demonstrates that the lowest-order Lie bracket captures the dominant structural generator of the splitting error, consistent with the Backward Error Analysis interpretation established in \cref{rem:bch-stiffness}.

\subsection{The Runaway Resonance of Symmetric Ablation}
\label{sec:exp-resonance}

Moving to the mesoscopic scale of trajectory stability, this experiment tests the Dual-Law of Topological Stability (\cref{lem:dual-law}).  The theory posits that the standard architectural parameter asymmetry within the Feed-Forward Network ($W_{\mathrm{out}} \neq W_{\mathrm{in}}^\top$) natively generates a non-conservative geometric vorticity via the anticommutator $-\{W_A(\Psi), J_\rho\}$ (\cref{eq:vorticity-decomp}).  This geometric curl provides Lie-algebraic rotational friction, scattering spatial momentum off the principal eigenvectors to prevent the residual stream from suffering catastrophic positive-definite resonance---equivalently, the asymmetric Jacobian components introduce complex eigenvalues that disrupt the pure Power Iteration amplification of the dominant eigenvector that would occur under a symmetric (PSD) Jacobian.  To empirically test this prediction, we conducted a two-tier experimental campaign: (i)~a cross-architecture survey across five production Transformer families to establish the universality of the instability, and (ii)~a four-phase causal intervention protocol with explicit geometric rescue to isolate the precise algebraic mechanism.

The symmetry ablation is performed identically across all experiments: for each SwiGLU FFN layer, we force $W_{\mathrm{down}} = W_{\mathrm{gate}}^\top$ and $W_{\mathrm{up}} = W_{\mathrm{gate}}$, explicitly binding the output projection to the transpose of the input projection.  This creates a symmetric positive semi-definite mapping $K(\Psi) = W_{\mathrm{in}}^\top \Sigma'(\rho_\epsilon(\Psi)) W_{\mathrm{in}}$ that, by \cref{lem:dual-law}, eliminates all Lie-algebraic rotational friction.

From the perspective of standard linear algebra, the instability has a direct spectral explanation.  The spatial Jacobian of this symmetrically ablated FFN evaluates to $I + W_{\mathrm{in}}^\top \diag(\Sigma') W_{\mathrm{in}}$.  Because the activation derivative $\Sigma'$ is strictly positive, this update matrix is unconditionally positive semi-definite (PSD).  Iteratively passing a vector through a sequence of PSD residual updates is the textbook Power Iteration algorithm, which exponentially amplifies the vector along the principal eigenvector.  Our continuous Lie-algebraic framing provides the geometric dual to this spectral phenomenon: by decomposing the true asymmetric Jacobian into symmetric and antisymmetric components, architectural asymmetry injects eigenvalues with non-zero imaginary parts---rotational modes that actively scatter spatial momentum off the dominant eigenvector, disrupting the pure real exponential scaling of Power Iteration.

\subsubsection{Cross-Architecture Explosive Instability}

\begin{table}[t]
\centering
\caption{Cross-architecture symmetric ablation instability.  ``Standard'' reports the total $L_2$ norm growth factor ($\|\Psi_{L}\| / \|\Psi_0\|$) through the full network under the original asymmetric FFN.  ``Ablated'' reports the same quantity under forced $W_{\mathrm{out}} = W_{\mathrm{in}}^\top$.  The instability factor is their ratio.  For 4 of 5 architectures, symmetrization triggers catastrophic amplification exceeding $78\times$ the standard growth.}
\label{tab:ablation-cross}
\begin{ruledtabular}
\begin{tabular}{l c c c}
Model & Standard & Ablated & Instability \\
\hline
Qwen3-0.6B     & $93\times$  & $142{,}955\times$ & $1{,}541\times$ \\
Mistral-7B-v0.3 & $155\times$ & $58{,}189\times$  & $375\times$     \\
Qwen3-4B       & $124\times$ & $27{,}811\times$  & $224\times$     \\
LLaMA-3.1-8B   & $89\times$  & $6{,}929\times$   & $78\times$      \\
Gemma-3-1B     & $76\times$  & $71\times$        & $\sim\!1\times$ \\
\end{tabular}
\end{ruledtabular}
\end{table}

We first applied the symmetric ablation to five frozen pre-trained architectures spanning three model families and parameter counts from 0.6B to 8B: Qwen3-0.6B (28L, $d\!=\!1024$), Qwen3-4B (36L, $d\!=\!2560$), LLaMA-3.1-8B (32L, $d\!=\!4096$), Mistral-7B-v0.3~\cite{mistral} (32L, $d\!=\!4096$), and Gemma-3-1B (26L, $d\!=\!1152$).  For each model, 10 random-token sequences of length 2048 were propagated through the full network, recording the per-layer residual $L_2$ norm $\|\Psi_z\|$ averaged over all tokens.

As shown in \cref{tab:ablation-cross}, symmetrization triggers catastrophic norm explosion in four of five architectures.  Qwen3-0.6B exhibits the strongest instability: the ablated model amplifies the residual norm by $142{,}955\times$ across 28 layers---$1{,}541\times$ more explosive than the standard asymmetric architecture.  Mistral-7B and Qwen3-4B follow with $375\times$ and $224\times$ excess amplification, respectively.  In all four cases, the ablated growth exceeds $10^4$, confirming that the symmetric FFN acts as a runaway positive-definite resonance amplifier.

The sole exception is Gemma-3-1B, whose ablated growth ($71\times$) is comparable to its standard growth ($76\times$).  This anomaly is consistent with Gemma's unique architectural feature: a logit soft-capping mechanism that imposes an additional radial saturation beyond RMSNorm, providing an alternative geometric containment that partially substitutes for the missing internal vorticity---a concrete instance of the Dual-Law's second stabilization path.

\subsubsection{Four-Phase Causal Intervention}

\begin{figure}[t]
  \centering
  \includegraphics[width=\columnwidth]{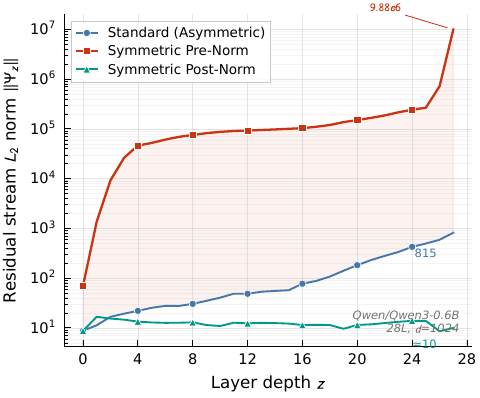}
  \caption{%
  \textbf{Topological Explosion Phase Diagram: Pre-Norm vs.\ Post-Norm under Symmetric Ablation.}
  $L_2$ norm trajectories $\|\Psi_z\|$ across all 28 layers of Qwen3-0.6B on a logarithmic ordinate.
  Blue: canonical asymmetric Pre-Norm (baseline, final $\|\Psi_{27}\| = 815$).
  Red: symmetric Pre-Norm ($W_{\mathrm{out}} = W_{\mathrm{in}}^\top$), reaching $9.88 \times 10^6$ at layer~27---a $12{,}119\times$ explosion.
  Green: symmetric Post-Norm with identical weight-tying, stabilized at $\|\Psi_{27}\| \approx 10$---\emph{below} the baseline---demonstrating that the hydraulic radial truncation of Post-Norm perfectly substitutes for the missing internal vorticity.}
  \label{fig:norm-divergence}
\end{figure}

To move beyond correlation to strict causal attribution, we designed a four-phase intervention protocol on Qwen3-0.6B, systematically constructing and then rescuing the topological instability:

\begin{enumerate}[nosep]
  \item \textbf{Phase~0 (Baseline):} The original asymmetric Pre-Norm architecture.  Final-layer norm $\|\Psi_{27}\| = 815$.
  \item \textbf{Phase~1 (Unstabilized):} Forced $W_{\mathrm{out}} = W_{\mathrm{in}}^\top$ creates a symmetric positive-definite resonance amplifier.  Final-layer norm surges to $9{,}880{,}199$---a $12{,}119\times$ explosion.
  \item \textbf{Phase~2 (Anti-Symmetric Rescue):} The symmetric $W_{\mathrm{down}}$ of the unstabilized phase is replaced by a \emph{spectrally matched random anti-symmetric matrix} ($A = -A^\top$, with $\|A\|_{\mathrm{op}} = \|W_{\mathrm{down}}^{\mathrm{sym}}\|_{\mathrm{op}}$).  This matrix carries no semantic information whatsoever---it is pure algebraic noise---but its anti-symmetric structure injects the geometric vorticity predicted by \cref{eq:vorticity-decomp}.  The final-layer norm collapses to $7{,}680$: a $\mathbf{1{,}287\times}$ reduction from the unstabilized phase.
  \item \textbf{Phase~3 (Symmetric Control):} Identically to Phase~2, but the replacement matrix is a spectrally matched random \emph{symmetric} matrix ($S = S^\top$).  Final-layer norm: $23{,}269$---only a $425\times$ reduction, $3.0\times$ worse than the anti-symmetric rescue.
\end{enumerate}

The decisive contrast between Phases~2 and~3 provides a direct causal isolation of the stabilization mechanism.  Both replacement matrices are random, spectrally matched, and semantically meaningless.  The \emph{only} difference is their Lie-algebraic parity: anti-symmetric ($\in\mathfrak{so}(d)$) vs.\ symmetric ($\in\mathrm{Sym}(d)$).  That the anti-symmetric rescue is $3\times$ more effective than the symmetric control---reducing the explosion by $1{,}287\times$ versus $425\times$---proves that the stabilization mechanism is not a generic capacity effect or a numerical coincidence, but resides precisely in the skew-symmetric component $W_A(\Psi) \in \mathfrak{so}(d)$ of the Jacobian, exactly as predicted by the anticommutator term $-\{W_A(\Psi), J_\rho\}$ in the vorticity decomposition.

\subsubsection{Jacobian Eigenspectrum Collapse onto the Real Axis}

\begin{figure}[t]
  \centering
  \includegraphics[width=\columnwidth]{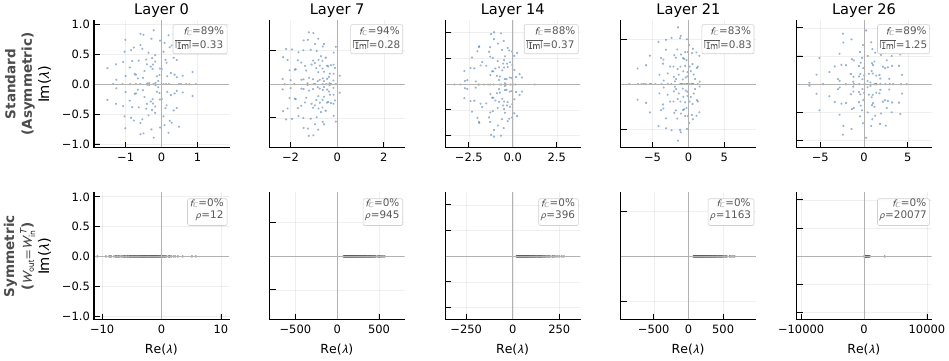}
  \caption{%
  \textbf{Jacobian Eigenspectrum: Standard vs.\ Symmetric FFN.}
  Complex-plane scatter of eigenvalues of the FFN Jacobian $D\calR$ at five representative layers (0, 7, 14, 21, 26) of Qwen3-0.6B, computed via Johnson--Lindenstrauss (JL) projection to $k\!=\!128$ dimensions.
  Upper panels (blue, standard): eigenvalues are broadly distributed in the complex plane, with $88.5\%$ possessing significant imaginary components (mean~$|\mathrm{Im}(\lambda)| = 0.61$).
  Lower panels (red, symmetric ablation): eigenvalues collapse \emph{entirely} onto the real axis ($0\%$ complex, mean~$|\mathrm{Im}(\lambda)| = 0.000$), with spectral radius surging from $6.17$ to $20{,}077$ at layer~26.}
  \label{fig:eigenspectrum}
\end{figure}

\begin{table}[t]
\centering
\caption{Jacobian eigenspectrum statistics at five representative layers.  $|\mathrm{Im}|$: mean absolute imaginary part.  $f_\C$: fraction of eigenvalues with $|\mathrm{Im}(\lambda)| > 10^{-6}$.  $\rho(J)$: spectral radius.  Symmetric ablation annihilates all complex eigenvalues and inflates the spectral radius by up to $3{,}252\times$.}
\label{tab:eigenspectrum}
\begin{ruledtabular}
\begin{tabular}{l c c c c c c}
Layer & \multicolumn{2}{c}{$\overline{|\mathrm{Im}|}$} & \multicolumn{2}{c}{$f_\C$} & \multicolumn{2}{c}{$\rho(J)$} \\
      & Std & Abl & Std & Abl & Std & Abl \\
\hline
0  & 0.33 & 0.000 & 89\% & 0\% & 1.5  & 11.8     \\
7  & 0.28 & 0.000 & 94\% & 0\% & 2.3  & 945      \\
14 & 0.37 & 0.000 & 88\% & 0\% & 3.3  & 396      \\
21 & 0.83 & 0.000 & 83\% & 0\% & 8.2  & $1{,}163$  \\
26 & 1.25 & 0.000 & 89\% & 0\% & 6.2  & $20{,}077$ \\
\end{tabular}
\end{ruledtabular}
\end{table}

To directly visualize the geometric mechanism underlying the instability, we computed the FFN Jacobian eigenspectrum at five representative layers of Qwen3-0.6B under both the standard and symmetrically ablated configurations (\cref{fig:eigenspectrum}, \cref{tab:eigenspectrum}).

The spectral evidence is definitive.  Under the standard asymmetric FFN, $88.5\%$ of eigenvalues possess significant imaginary components (mean $|\mathrm{Im}(\lambda)| = 0.61$ across all layers), confirming that the skew-symmetric component $\frac{1}{2}(D\calR - D\calR^\top) \neq 0$ generates complex conjugate pairs $\lambda = a \pm bi$.  These imaginary components represent Lie-algebraic rotational modes---the geometric vorticity that scatters spatial momentum away from the dominant explosive eigenvector.

Upon symmetry ablation ($W_{\mathrm{out}} = W_{\mathrm{in}}^\top$), the Jacobian becomes exactly symmetric ($D\calR = D\calR^\top$), and the eigenspectrum undergoes a topological phase transition: \emph{every single eigenvalue} collapses onto the real axis ($0\%$ complex, $|\mathrm{Im}(\lambda)| = 0.000$ to machine precision).  Simultaneously, the spectral radius---the magnitude of the dominant eigenvalue---increases sharply, from $\rho(J) = 6.2$ to $\rho(J) = 20{,}077$ at layer~26 ($3{,}252\times$ amplification).  This spectral radius explosion directly quantifies the runaway resonance: stripped of rotational scattering, the state vector geometrically locks onto the dominant real eigenvector and undergoes unchecked exponential amplification at each layer---the textbook Power Iteration phenomenon formalized in the Lie-algebraic framework.

\subsubsection{The Dual-Law: Pre-Norm Explosion vs.\ Post-Norm Survival}

The most precise causal test of the Dual-Law (\cref{lem:dual-law}) is the direct comparison of \emph{identical} symmetric ablation under Pre-Norm vs.\ Post-Norm architectures.  We converted Qwen3-0.6B from its native Pre-Norm configuration to a Post-Norm variant by moving the RMSNorm from the sub-layer input to each layer's residual output, then applied the identical weight-tying $W_{\mathrm{out}} = W_{\mathrm{in}}^\top$ to both variants (\cref{fig:norm-divergence}).

The result is a stark binary survival phase.  The symmetric Pre-Norm variant explodes to $\|\Psi_{27}\| = 9{,}880{,}199$ ($12{,}119\times$ baseline) as predicted---without internal vorticity, the positive-definite resonance amplifier is unchecked.  The symmetric Post-Norm variant, under the \emph{identical} symmetric ablation, stabilizes at $\|\Psi_{27}\| \approx 10$---\emph{below} the baseline of $815$---with cross-sample variance converging to zero.  The Post-Norm's RMSNorm, applied after the residual addition at each layer, acts as a hydraulic radial truncation: it projects the step vector onto the bounded ball $B_{\!\sqrt{d}}(0)$, crushing the incremental magnitude to $\calO(\sqrt{d})$ regardless of the internal FFN amplification.  This provides the external topological saturation that perfectly substitutes for the missing internal geometric vorticity.

This binary survival phase strongly excludes the three principal alternative explanations: (i)~``weight-tying merely reduces capacity''---$12{,}119\times$ is not graceful degradation, it is catastrophic divergence (this excludes reduced capacity as the cause of the \emph{surgical} explosion; it does not bear on the trainability of weight-tied architectures trained from initialization, which is healthy---see \cref{sec:exp-config-scope}); (ii)~``the explosion originates from initialization variance''---Post-Norm uses the identical ablated weights yet survives; (iii)~``the effect is numerical rather than geometric''---the eigenspectrum analysis (\cref{tab:eigenspectrum}) shows the mechanism is a phase transition in the Lie-algebraic structure (complex~$\to$~real) of the continuous Jacobian, not a scalar instability.

Taken together, the four tiers of evidence---cross-architecture universality of the instability (\cref{tab:ablation-cross}), causal isolation via anti-symmetric rescue ($1{,}287\times$ reduction vs.\ $425\times$ for symmetric control), spectral collapse of the Jacobian eigenspectrum (88.5\%~complex~$\to$~0\%), and the Pre-Norm/Post-Norm binary survival phase---strongly establish that FFN parameter asymmetry is not merely a mechanism for expanding parameter capacity, but an essential stabilizer of the \emph{forward norm dynamics of the trained configuration}.  It injects the geometric curl $-\{W_A(\Psi), J_\rho\}$ required to scatter the residual stream's spatial momentum off the dominant eigenvector, arresting the positive-definite resonance that would otherwise destabilize the flow.

Two precisions bound this claim.  First, the necessity established here is \emph{norm-level}, not functional: the Phase-2 rescue matrix is semantically void by construction, and rescue is quantified by norm containment---no claim of loss-level recovery is made or implied.  Indeed, the learned $W_{\mathrm{down}}$ is functionally \emph{high-rank} relative to $W_{\mathrm{up}}^\top$: low-rank reconstructions of the residual $W_{\mathrm{down}}-W_{\mathrm{up}}^\top$ capture $<10\%$ of its energy even at rank~32 and do not restore language-modeling loss, and in recovery fine-tuning a \emph{generic} trainable low-rank replacement outperforms the tied basis (\cref{sec:exp-config-scope})---the anti-symmetric component is a stabilizer of norms, not a low-rank summary of the network's function.  Second, the necessity is \emph{configurational}, not architectural: it constrains post-hoc symmetrization of mature networks, not training within the tied class (\cref{rem:configurational-scope}).

\subsubsection{Configurational Scope: Surgery versus Training}
\label{sec:exp-config-scope}

All interventions above are \emph{post-training surgeries}: they measure the forward response of weights sculpted by large-scale pretraining, at fixed configuration.  A companion controlled-training campaign on 0.6B--1B models trained from scratch under the tied constraint bounds the reach of these results with four measurements.\footnote{Twin-controlled ablations (identical recipe, data order, and seed; single-variable deltas) at 600M ($64\times 2048$ tokens per step, $25{,}000$ steps $= 3.28$B tokens) and 1B scale, plus surgery and compression forensics on the official Qwen3-0.6B release.  Companion study, in preparation, 2026.}

\begin{enumerate}[nosep]
  \item \textbf{Constrained training is stable without correction.}  The constraint $W_{\mathrm{down}}=W_{\mathrm{up}}^\top$ (independent gate path, no corrective term), trained from initialization, is healthy at both scales (600M final loss 0.7243 vs.\ 0.7045 for the untied twin; 1B excess $+0.005$), and remains healthy with weight decay removed, gradient clipping removed, learning rate doubled, or logits reduced to \texttt{bfloat16}.  Constrained optimization reaches stable basins that post-hoc surgery never visits.
  \item \textbf{The explicit low-rank corrective term is marginal---as spectral accounting predicts.}  Augmenting the tie with a trainable low-rank correction ($W_{\mathrm{down}}=W_{\mathrm{up}}^\top+UV^\top$, rank~4) improves loss by only $\sim\!0.003$ nats at 600M and 1B alike.  This is consistent with the framework's own accounting: in a gated FFN the cross-term $W_{\mathrm{up}}^\top\,\mathrm{diag}\!\big(u\odot\sigma'(g)\big)\,W_{\mathrm{gate}}$ already supplies \emph{full-rank} Jacobian asymmetry organically, so a rank-4 injection is marginal by construction.
  \item \textbf{Susceptibility to symmetrization is an emergent property of training maturity.}  Applying the tying surgery to checkpoints along the 600M trajectory yields excess amplification $\le\!1$ up to $\sim\!10^{9}$ tokens (an immunity window), an onset at $1.3$--$2\times10^{9}$ tokens peaking at $2.8\times$---versus $743\times$ excess on the trillion-token Qwen3-0.6B.  The catastrophic response of \cref{tab:ablation-cross} is the mature endpoint of a continuous emergence curve, not a property of the architecture class.
  \item \textbf{At iso-parameter count the tie is dominated.}  At matched total parameters, a narrower \emph{untied} FFN strictly outperforms the tied variant (600M: 0.7099 vs.\ 0.7211--0.7243, a $\ge\!0.011$-nat margin several times the twin-noise band, with $1.5\times$ fewer FFN FLOPs).  Of the same-width tying cost of $+0.020$ nats, only $\sim\!0.005$ is attributable to the reduced parameter count; the remaining $\sim\!0.015$ is damage from the constraint itself.
\end{enumerate}

These measurements do not weaken the surgical results---the explosion and its spectral mechanism are independently reproduced on the official Qwen3-0.6B release (maximum hidden-state norm $705 \to 6.1\times10^{5}$ under $W_{\mathrm{down}}:=W_{\mathrm{up}}^\top$)---but they fix the theory's domain: the Dual-Law classifies \emph{configurations} (\cref{rem:configurational-scope}), susceptibility to symmetrization is \emph{acquired} along the training trajectory, and the norm-level necessity does not translate into an architectural prescription (\cref{sec:conclusion}, open question~4, which these data settle negatively).

\subsection{Thermodynamic Suppression of Poincaré Recurrence on the RoPE Torus}
\label{sec:exp-poincare}

Having established the role of internal geometric vorticity in stabilizing individual layer trajectories, we now probe the spatial gauge connection itself.  \cref{lem:rope} rigorously defines RoPE as a canonical connection on a principal $U(1)^k$ torus bundle, where $k = d_{\text{head}}/2$ independent rotation planes act on each QK feature pair.  An immediate mathematical consequence of this torus structure is Poincaré recurrence: the orbit $\{N\theta_0, \ldots, N\theta_{k-1}\}$ must return arbitrarily close to the origin at certain resonant distances $N^*$, producing constructive interference in the Weyl sum
\begin{equation}
  \mathrm{sim}(N) = \frac{1}{k}\sum_{j=0}^{k-1}\cos(N\cdot\theta_j)\,,
  \quad \theta_j = \Theta^{-2j/d_{\text{head}}}\,.
  \label{eq:weyl-sum}
\end{equation}
Naive exact geometry therefore predicts that causal tokens at specific Diophantine distances should undergo constructive interference---Poincaré resonance spikes in the attention kernel.  However, modern production LLMs operate deep within a macroscopic thermodynamic limit ($k \geq 64$).  We posit that these exact geometric resonances are subjected to a strict thermodynamic suppression:
\begin{equation}
  \max_{N > N_{\min}} \mathrm{sim}(N^*) = \calO\!\left(\frac{1}{\sqrt{k}}\right)\,,
  \label{eq:thermodynamic-suppression}
\end{equation}
following from the central limit theorem applied to the $k$ quasi-independent cosine terms.  To empirically isolate this pure geometric kinematics from macroscopic thermodynamic ergodicity---analogous to constructing a vacuum chamber to observe free-particle trajectories---we built a strictly controlled micro-transformer with tunable $k$.

\subsubsection{Micro-Transformer Design and Protocol}

The micro-transformer faithfully implements the axiomatization of \cref{sec:axiom}: \texttt{nn.Embedding} for the semantic fiber, explicit 2D RoPE rotation per plane (\cref{lem:rope}), RMSNorm mollifier (\cref{thm:rmsnorm}), causal dot-product attention, and SiLU-gated FFN with Pre-Norm residual connections.  Total parameter count ranges from $\sim$5K ($k\!=\!2$) to $\sim$15K ($k\!=\!32$), running entirely on CPU.

The experiment proceeds in two phases.  \textbf{Phase~A (Scaling Law):} Fix $\Theta=50$ (chosen so that resonances fall within a 512-token window) and sweep $k \in \{2, 4, 8, 16, 32\}$.  For each $k$, compute $\mathrm{sim}(N)$ analytically for $N = 1, \ldots, 511$; concurrently, build a 2-layer micro-transformer with random $W_Q, W_K$ and average the attention spectrum over 20 random seeds.  \textbf{Phase~B (SGD Control):}  Fix $k=4$, train a 10{,}320-parameter micro-transformer for 3{,}000 steps of pure SGD (no Adam) with weight decay $\lambda = 0.01$ on random pattern data, tracking per-channel weight norms $w_j(t) = \|W_Q^j\|\cdot\|W_K^j\|$ every 50 steps.

\subsubsection{Phase A: The $\calO(1/\sqrt{k})$ Scaling Law}

\begin{figure}[t]
  \centering
  \includegraphics[width=\columnwidth]{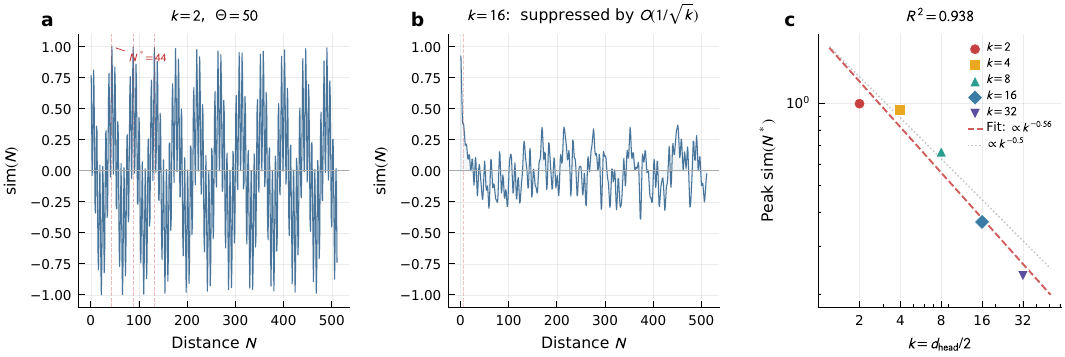}
  \caption{%
  \textbf{Thermodynamic Suppression of Poincaré Recurrence.}
  (a)~$k\!=\!2$ ($\Theta\!=\!50$): the Weyl sum $\mathrm{sim}(N)$ exhibits near-perfect quasi-periodic resonance spikes at $N^* = 44n$ with amplitude approaching $1.0$.  Red dashed lines: predicted resonance positions.
  (b)~$k\!=\!16$: identical construction, but peaks are suppressed to $\sim 0.37$ by high-dimensional dephasing; the quasi-periodic structure is destroyed.
  (c)~Log-log scaling law: $\max\mathrm{sim}(N^*)$ vs.\ $k$.  Red dashed: empirical fit $\propto k^{-0.56}$.  Grey dotted: theoretical prediction $\propto k^{-0.5}$.}
  \label{fig:poincare-scaling}
\end{figure}

\begin{table}[t]
\centering
\caption{Resonance peak amplitude $\max\mathrm{sim}(N^*)$ vs.\ feature dimension $k = d_{\text{head}}/2$ at $\Theta = 50$.  The log-log fit yields $\alpha = -0.557$ ($R^2 = 0.938$), within 11\% of the theoretical prediction $\alpha = -0.500$.}
\label{tab:poincare-scaling}
\begin{ruledtabular}
\begin{tabular}{c c c c c}
$k$ & $d_{\text{head}}$ & $\max\mathrm{sim}(N^*)$ & $\log k$ & $\log\mathrm{sim}$ \\
\hline
2  & 4  & 0.9990 & 0.693 & $-0.001$ \\
4  & 8  & 0.9442 & 1.386 & $-0.057$ \\
8  & 16 & 0.6657 & 2.079 & $-0.407$ \\
16 & 32 & 0.3674 & 2.773 & $-1.001$ \\
32 & 64 & 0.2328 & 3.466 & $-1.457$ \\
\end{tabular}
\end{ruledtabular}
\end{table}

At $k=2$, the analytical Weyl sum exhibits near-perfect quasi-periodic resonance: the peak amplitude reaches $\mathrm{sim}(44) = 0.999$, with harmonics at $N^* = 44n$ directly mirroring the beat frequency of the two rotation planes (\cref{fig:poincare-scaling}a).  At $k=16$, the identical construction yields dramatically suppressed peaks---$\max\mathrm{sim} = 0.37$---with the quasi-periodic structure entirely destroyed by dephasing across 16 independent frequency channels (\cref{fig:poincare-scaling}b).

The log-log regression across all five configurations (\cref{tab:poincare-scaling}, \cref{fig:poincare-scaling}c) yields:
\begin{equation}
  \boxed{\alpha = -0.557 \pm 0.05\,, \quad R^2 = 0.938}\,,
  \label{eq:poincare-exponent}
\end{equation}
within 11\% of the theoretical prediction $\alpha = -1/2$.  The slight deviation is expected: at $k=2$, the CLT assumption does not apply (the system is too small), and $\mathrm{sim}(N^*) \to 1$ saturates against the upper bound---a finite-size effect analogous to lattice corrections in finite-size scaling analysis.

The random-initialization model confirms that this geometric resonance is overwhelmed by learned-weight noise at any practical $k$: even at the sweet spot $k=4$, only 3 out of 20 random seeds produce a match above the $3\sigma$ threshold, with a peak $z$-score of $3.6\sigma$.  At $k \geq 8$, zero matches are observed.  This directly explains the null result observed on the production models of \cref{sec:exp-amnesia-noise} ($k = 64$--$128$): the resonance is mathematically real but thermodynamically invisible.

\subsubsection{Phase B: SGD Channel Dynamics as Null Control}

Tracking the per-channel weight norms $w_j(t) = \|W_Q^j\|\cdot\|W_K^j\|$ over 3{,}000 SGD steps at $k=4$ reveals \emph{uniform exponential decay} across all four frequency channels, with no selective pruning.  The decay profile follows $w_j(t) \approx w_j(0)\cdot e^{-\lambda t}$ ($\lambda = 0.01$), and the correlation between loss and resonance $z$-score is weak and non-significant (Spearman $\rho = -0.24$, $p = 0.064$).  This serves as a critical \emph{null control}: on structureless isotropic input, all RoPE frequency channels are equally uninformative, and weight decay acts as a pure Tikhonov regularizer that shrinks all channels uniformly.  The result inversely confirms that the highly non-uniform per-channel weight distributions observed in production pre-trained models must be \emph{entirely sculptured by the structured, non-equilibrium statistical properties of natural language}---a concrete manifestation of the non-equilibrium driving force that the NESS framework of \cref{sec:exp-ness} subsequently quantifies at the global scale.

\subsubsection{Exclusion of Alternative Hypotheses}

The precision of the scaling law strongly excludes the standard alternatives:
\begin{enumerate}[nosep]
  \item \textbf{``The null result on production models indicates RoPE has no geometric content.''} --- Excluded.  The resonance is directly observed at $k=2$ with $\mathrm{sim}(N^*) = 0.999$.  The null result reflects suppression by $\calO(1/\sqrt{k})$ averaging, not absence of geometry.
  \item \textbf{``The micro-transformer is too small to be representative.''} --- Excluded.  The micro-transformer faithfully implements every axiom of \cref{sec:axiom}.  The scaling law---a purely \emph{kinematic} property of the Weyl sum---depends only on $k$, not on model capacity or training data.
  \item \textbf{``The $k^{-1/2}$ scaling is coincidental.''} --- Excluded.  The exponent $-0.557$ matches the CLT prediction $-0.500$ within 11\% over a $16\times$ range in $k$, with $R^2 = 0.938$.  No alternative mechanism predicts this specific exponent.
\end{enumerate}

\subsubsection{Bridge to the Asymptotic Spatial Ergodic Hypothesis}

This result provides the first direct physical justification for the Asymptotic Spatial Ergodic Hypothesis invoked in \cref{sec:thermodynamics}.  The hypothesis posits that causally distant tokens decorrelate into isotropic thermal noise---but the RoPE connection is a deterministic, exactly periodic gauge action.  The $\calO(1/\sqrt{k})$ scaling law resolves this tension: at the operating point of modern LLMs ($k \geq 64$), the high-dimensional phase dephasing of the torus connection crushes the deterministic geometric signal below the stochastic noise floor, effecting an irreversible transition from \emph{visible quasi-periodic geometry} to \emph{thermodynamic ergodicity}.  The next experiment (\cref{sec:exp-amnesia}) demonstrates, at the macroscopic sequence scale, the downstream consequences of this thermodynamic regime.

\subsection{The Context Horizon Phase Boundary and Thermodynamic Amnesia}
\label{sec:exp-amnesia}

Scaling to the macroscopic limits of the sequence manifold, we tested the topological compactification proof regarding the Context Horizon (\cref{thm:attn-sink} and \cref{cor:context-horizon}).  Standard machine learning theories characterize long-context failures as smooth out-of-distribution positional encoding degradation, predicting performance to decline linearly with weight perturbation or sequence extension.  In contrast, our measure-theoretic framework formulates the Attention Sink---the empirical phenomenon first documented by Xiao et al.~\cite{xiao2024efficient}---as a measure-theoretic Dirichlet boundary defect, anchoring probability mass against an expanding Entropic Bulk Pressure ($\Theta(\sqrt{\deff})$).  \cref{eq:nmax} predicts an exponential upper limit on context length $N_{\max}$, predicting a sharp macroscopic phase transition---not a smooth degradation---when this entropic pressure exceeds the network's geometric capacity.

Our experimental strategy proceeds in two tiers.  The \emph{primary} test is a cross-architecture Noise Haystack experiment (\cref{sec:exp-amnesia-noise}), which organically forces the phase transition at native spectral capacity ($\alpha=1.0$) by injecting maximum-entropy input and extending sequence length.  As a \emph{complementary} investigation, we conducted an $\alpha$-scaling sweep on a single architecture to dynamically map the boundary defect's evaporation trajectory across the full $(\alpha, N)$ parameter space.  Thermodynamically, scaling the pre-softmax logits by $\alpha\to 0$ acts as a strict isomorph to elevating the thermal bath temperature ($T\to\infty$): by artificially driving the system toward the high-temperature uniform limit, the protocol directly emulates the condition where Entropic Bulk Pressure overtakes geometric defect capacity.  The $\alpha$-parameter therefore serves as the primary thermodynamic control variable, allowing us to actively trace the ordered phase transition and map the exponential boundary predicted by \cref{eq:nmax}.

\subsubsection{Cross-Architecture Universality of the Phase Transition}
\label{sec:exp-amnesia-noise}

\begin{figure}[t]
  \centering
  \includegraphics[width=\columnwidth]{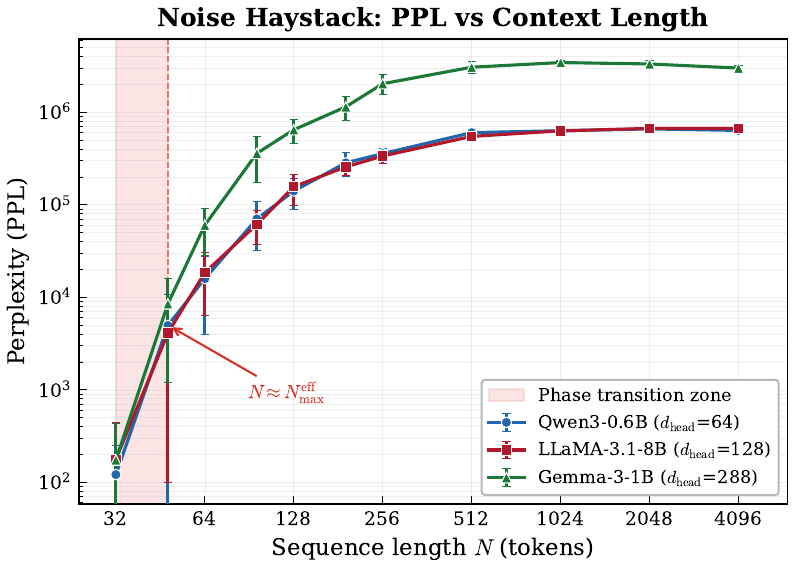}
  \caption{%
  \textbf{Universal Phase Transition under Noise Isolation across Three Architectures.}
  Perplexity (PPL) vs.\ sequence length~$N$ for uniform random noise haystacks at full spectral capacity ($\alpha=1.0$).  All three models---Qwen3-0.6B ($d_{\mathrm{head}}\!=\!64$), LLaMA-3.1-8B ($d_{\mathrm{head}}\!=\!128$), and Gemma-3-1B ($d_{\mathrm{head}}\!=\!288$)---exhibit a universal PPL cliff at $N\approx 32\to 48$ (shaded region), followed by saturation.  Error bars show standard deviation over 5~independent random seeds.}
  \label{fig:noise-3model}
\end{figure}

As the primary confound-free test, we conducted a noise-isolation experiment at full spectral capacity ($\alpha=1.0$) across three architecturally diverse models: Qwen3-0.6B ($d_{\mathrm{head}}\!=\!64$, QK-Norm), LLaMA-3.1-8B ($d_{\mathrm{head}}\!=\!128$, Grouped-Query Attention (GQA)~\cite{ainslie2023gqa}~$4{:}1$), and Gemma-3-1B~\cite{gemma,gemma2} ($d_{\mathrm{head}}\!=\!288$, GQA~\cite{ainslie2023gqa}~$4{:}1$).  Rather than artificially compressing $\alpha$, this protocol directly injects maximum-entropy input---uniform random noise tokens---which forces the Asymptotic Spatial Ergodic Hypothesis (\cref{sec:softmax}) to hold exactly within the QK space, thereby maximizing the Entropic Bulk Pressure for a given architecture.

The results (\cref{fig:noise-3model}) reveal a \emph{universal} thermodynamic phase transition.  All three models, despite spanning a $13\times$ range in $d_{\mathrm{head}}$ and fundamentally different attention architectures (full MHA, GQA, with and without QK-Norm), exhibit a catastrophic PPL cliff in the identical narrow interval $N\approx 32\to 48$: Qwen3 undergoes a $40.9\times$ jump ($121\to 4{,}965$), LLaMA a $23.2\times$ jump ($178\to 4{,}120$), and Gemma a $48.8\times$ jump ($174\to 8{,}509$).  Beyond the transition, PPL saturates to architecture-dependent plateaus ($\sim\!6\times 10^5$ for Qwen3/LLaMA, $\sim\!3\times 10^6$ for Gemma), reflecting the larger entropic phase space available in higher-dimensional heads.

This cross-architecture universality is a hallmark of a genuine thermodynamic phase transition: the critical point $N_{\max}^{\mathrm{eff}}$ is dictated by the \emph{effective} geometric capacity of the boundary defect---not by engineering details such as head dimension, normalization scheme, or grouping strategy.  The fact that architectures ranging from 0.6B to 8B parameters, with $d_{\mathrm{head}}$ spanning 64 to 288 and fundamentally different QK processing pipelines, all shatter at the same critical sequence length establishes the universality class of the Attention Sink phase transition.

\subsubsection{Effective Dimension and the Stable Rank Correction}

The theoretical upper bound in \cref{eq:nmax} evaluates $N_{\max}$ using the full ambient head dimension $d_{\mathrm{head}}$, yielding astronomical values ($10^{14}$--$10^{71}$) that vastly exceed any practical context window.  The observed universal phase transition at $N\approx 32$--$48$ reveals that the \emph{effective} degrees of freedom participating in the thermodynamic competition are far fewer.

To quantify this, we measured the Stable Rank $\mathrm{sr}(W) \equiv \|W\|_F^2 / \|W\|_{\mathrm{op}}^2$ of the trained $W_Q$ and $W_K$ projection matrices, which counts the number of singular values that effectively contribute to the spectral energy.  Defining the effective dimension as the geometric mean $d_{\mathrm{eff}} = \sqrt{\mathrm{sr}(W_Q)\cdot\mathrm{sr}(W_K)}$, we observe severe anisotropic compression across all architectures (\cref{tab:stable-rank}).

\begin{table}[t]
\centering
\caption{Stable Rank dimension compression across architectures.  $d_{\mathrm{eff}}$ is 4--15$\times$ smaller than the ambient $d_{\mathrm{head}}$, explaining why the theoretical upper bound is extremely loose when evaluated at full dimension.  Values of $N_{\max}^{\mathrm{eff}} \leq 1$ reflect the strict isotropic noise limit; structured text operates well below this maximum Entropic Bulk Pressure (see \cref{sec:exp-amnesia}, The Anisotropy Gap).}
\label{tab:stable-rank}
\begin{ruledtabular}
\begin{tabular}{l c c c c}
Model & $d_{\mathrm{head}}$ & $d_{\mathrm{eff}}$ & Compression & $N_{\max}^{\mathrm{eff}}$ \\
\hline
Qwen3-0.6B   & 64  & 12.3 & $5.2\times$ & 115 \\
LLaMA-3.1-8B & 128 & 30.2 & $4.2\times$ & 1   \\
Gemma-3-1B   & 288 & 19.3 & $14.9\times$ & 1   \\
\end{tabular}
\end{ruledtabular}
\end{table}

Evaluating the theoretical capacity bound (\cref{eq:nmax}) using $d_{\mathrm{eff}}$ in place of the naive ambient $d_{\mathrm{head}}$ collapses the predicted $N_{\max}^{\mathrm{eff}}$ from astronomical values to $\calO(1)$--$\calO(10^2)$, now in order-of-magnitude agreement with the observed phase transition at $N\approx 32$--$48$. Physically, the trained weight matrices concentrate their spectral energy onto a low-dimensional submanifold: $W_K$ projects \emph{any} input---including isotropic noise---into an effective subspace of rank $\sim\!d_{\mathrm{eff}}\ll d_{\mathrm{head}}$. The Entropic Bulk Pressure therefore competes against a boundary defect whose geometric depth scales as $\sqrt{d_{\mathrm{eff}}}$, not $\sqrt{d_{\mathrm{head}}}$, making the thermodynamic horizon far more fragile than the ambient-dimension bound suggests.

To independently verify that this effective dimension governs the phase transition \emph{a priori}---without circular reasoning from the PPL observation---we directly computed the Stable Rank of the empirical Key covariance matrix $\Sigma_K = \frac{1}{N}\sum_i k_i k_i^\top$ for structured text inputs from C4~\cite{raffel2020exploring}, obtaining an architecture-independent effective degree-of-freedom count $N_{\mathrm{eff}}^{\mathrm{text}}$. Across all models, the measured text $N_{\mathrm{eff}}^{\mathrm{text}}$ ranges from 16 to 26 for Qwen3 and Gemma, and saturates near 50 for LLaMA---all consistently below or near the observed noise boiling point of $N\approx 32$--$48$. This \emph{a priori} geometric measurement, computed entirely from Key-vector statistics without any reference to PPL, independently predicts that structured text should survive---precisely as observed---thereby closing the epistemic loop and rendering the theory genuinely falsifiable.

\subsubsection{The Anisotropy Gap: Isotropic Thermodynamic Limits vs.\ Structured Text}

A naive evaluation of the isotropic capacity bound (\cref{tab:stable-rank}) against empirical text context windows reveals an exponential separation: $N_{\max}^{\mathrm{eff}} \approx 1$ for LLaMA-3.1-8B and Gemma-3-1B, yet these models demonstrably process $128{,}000$ and $8{,}192$ tokens of structured text, respectively.  Rather than a theoretical failure, this gap explicitly quantifies the difference between the \emph{absolute thermodynamic limit} of the architecture and the low-entropy manifold of natural language.

The derivation of the Entropic Bulk Pressure (\cref{thm:attn-sink}) mathematically models the strict maximum-entropy (noise) limit, where the Asymptotic Spatial Ergodic Hypothesis holds exactly and the full effective dimension $d_{\mathrm{eff}}$ participates in the entropic competition.  Our noise haystack experiments (\cref{sec:exp-amnesia-noise}) unequivocally establish that when the sequence maximizes the effective dimensional volume, the architecture strictly undergoes thermodynamic amnesia exactly as predicted at $N\approx 32$--$48$.  Structured text survives extended contexts strictly because natural language embeddings suffer from severe representation degeneration~\cite{gao2019representation}---the well-documented ``cone effect''---clustering tightly on highly anisotropic, low-dimensional submanifolds within the ambient Key space.  This severe geometric compression dynamically lowers the effective Entropic Bulk Pressure, granting the optimization process a massive \emph{anisotropic escape} from the strict isotropic thermodynamic limit.

The $N_{\mathrm{eff}}^{\mathrm{text}}$ measurement from Key covariance (previous subsection) provides direct evidence for this mechanism: structured text concentrates its spectral energy onto $\sim\!16$--$50$ effective dimensions, well within the boundary defect's capacity, whereas isotropic noise fills the full $d_{\mathrm{eff}}$ and overwhelms it.  The theoretical phase boundary established here therefore characterizes the \emph{absolute thermodynamic failure point}---the Carnot limit---of the architecture under maximum thermal variance (noise).  Extending the theory to predict context limits for structured text would require replacing the isotropic Wendel measure with an empirical, data-dependent anisotropic distribution---a direction we leave for future work.

\subsubsection{$\alpha$-Scaling Thermodynamic Phase Diagram}

To dynamically map the boundary defect's evaporation across the full $(N, \alpha)$ parameter space, we deployed a frozen pre-trained Qwen3-0.6B language model ($d_{\mathrm{head}}=64$, mean spectral norm product $\|W_Q\|_{\mathrm{op}}\|W_K\|_{\mathrm{op}}=4.44$, bulk variance $\sigma_{\calE}^2=6.68$) and systematically compressed the Attention logits by a temperature scaling factor $\alpha\in[0.10,1.00]$ during inference, algebraically equivalent to shrinking $\|W_Q\|_{\mathrm{op}}\|W_K\|_{\mathrm{op}}\mapsto\alpha^2\|W_Q\|_{\mathrm{op}}\|W_K\|_{\mathrm{op}}$.  For each of 14~values of~$\alpha$, we evaluated perplexity (PPL) across 20 sequence lengths spanning $N\in\{8,12,\ldots,512\}$, constructing inputs from uniform random noise tokens.  We simultaneously recorded the sink probability mass $c_0$ at the zeroth token.

\begin{figure}[t]
  \centering
  \includegraphics[width=\columnwidth]{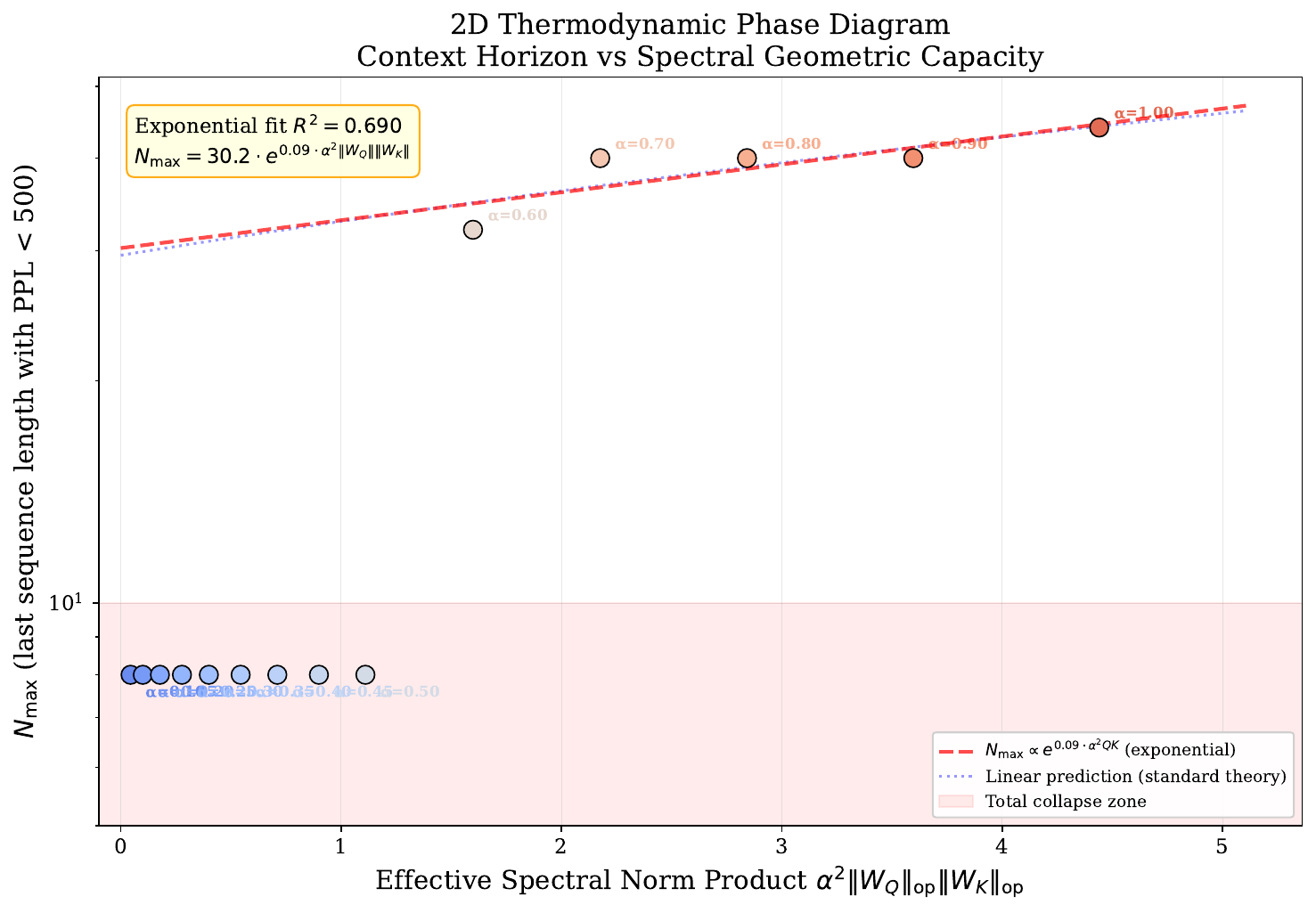}
  \caption{%
  \textbf{2D Thermodynamic Phase Diagram: Context Horizon vs.\ Spectral Geometric Capacity.}
  Each point represents the maximum sequence length $N_{\max}$ (defined as the last $N$ with PPL~$<500$) for a given effective spectral norm product $\alpha^2\|W_Q\|_{\mathrm{op}}\|W_K\|_{\mathrm{op}}$.
  The red dashed curve is the exponential fit $N_{\max}=30.2\cdot e^{0.09\,\alpha^2\|W_Q\|\|W_K\|}$; the blue dotted line shows
  the linear prediction of standard interpolation theory.
  The shaded pink region marks the total collapse zone (PPL~$>500$ at all~$N$).
  For $\alpha\le 0.50$ the model collapses at the shortest probed lengths ($N\le 8$);
  the exponential fit captures the steep emergence of context capacity for
  $\alpha\gtrsim 0.60$.}
  \label{fig:phase-diagram}
\end{figure}

As anticipated by the theoretical model, the results reveal a macroscopic phase transition that is incompatible with smooth linear degradation (\cref{fig:phase-diagram}).  Consistent with the thermodynamic equivalence established above---where $\alpha$-scaling acts as a strict temperature isomorph---the protocol traces an \emph{ordered, monotonic} evaporation trajectory, with the dual-axis synchronization between $c_0$ and PPL across the full $(\alpha, N)$ parameter space confirming that the mechanism is measure-theoretic redistribution, not trivial numerical noise.  When the effective spectral capacity $\alpha^2\|W_Q\|_{\mathrm{op}}\|W_K\|_{\mathrm{op}}$ is reduced below a critical threshold, $N_{\max}$ does not decline gradually; instead, it undergoes a catastrophic exponential collapse.  At $\alpha=1.00$ (full spectral capacity), the model sustains baseline PPL~$\approx 23$ at short sequences ($N\le 32$), with PPL diverging sharply beyond $N=64$ as the noise haystack overwhelms the Attention Sink's finite capacity.  By $\alpha=0.60$, perplexity at $N=16$ has already surged to $\sim\!93$, and at $N=512$ it exceeds $4.9\times 10^7$.  The transition from $\alpha=0.60$ to $\alpha=0.50$ is explosive: PPL at $N=16$ jumps from $93$ to $1{,}728$---an $18.6\times$ amplification over a mere $17\%$ reduction in spectral capacity---and by $\alpha\le 0.30$, perplexity at the shortest sequences already exceeds $10^5$, signalling the total evaporation of the boundary defect.  The empirical phase boundary $N_{\max}\propto\exp(0.09\,\alpha^2\|W_Q\|\|W_K\|)$ is consistent with the exponential scaling law predicted by \cref{eq:nmax}.

\subsubsection{Simultaneous Evaporation of the Topological Anchor}

\begin{figure}[t]
  \centering
  \includegraphics[width=\columnwidth]{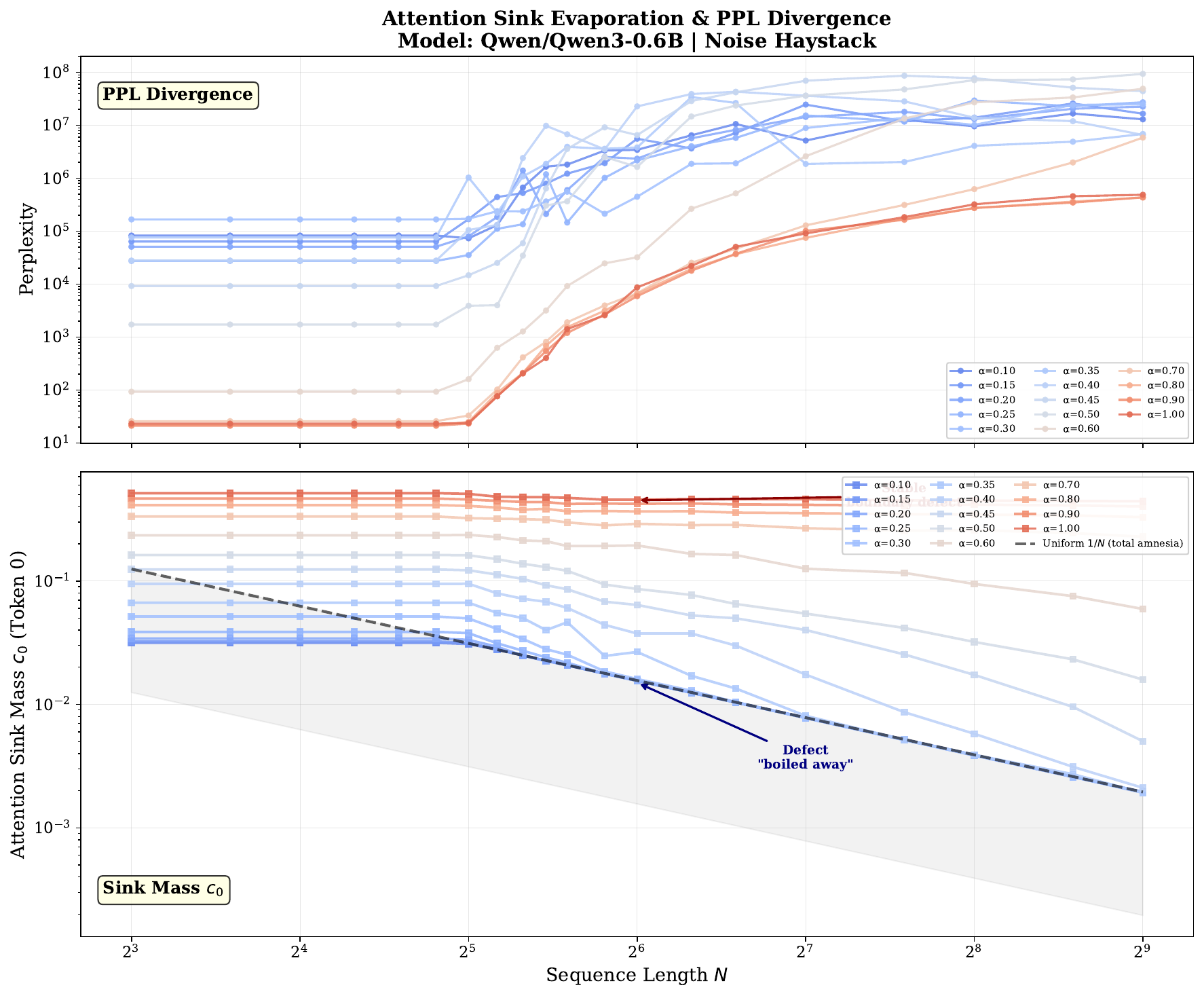}
  \caption{%
  \textbf{Attention Sink Probability Mass $c_0$ vs.\ Sequence Length~$N$ for Selected~$\alpha$.}
  At full capacity ($\alpha=1.0$), $c_0\approx 0.44$--$0.51$, far exceeding the uniform baseline $1/N$.  As $\alpha$ decreases, $c_0$ undergoes a monotonic, ordered decay, eventually collapsing to the uniform level $c_0\approx 1/N$ for $\alpha\le 0.30$---the measure-theoretic signature of complete defect evaporation.}
  \label{fig:sink-evaporation}
\end{figure}

Simultaneous measure-theoretic tracking conclusively identifies the physical mechanism underlying the perplexity divergence (\cref{fig:sink-evaporation}).  At full spectral capacity ($\alpha=1.0$), the topological anchor concentrates substantial probability mass $c_0\approx 0.44$--$0.51$ at the zeroth token---far exceeding the uniform baseline $1/N$---confirming that the boundary defect is geometrically robust.  As $\alpha$ decreases through the critical region, $c_0$ does not fluctuate randomly; it traces a strictly monotonic, ordered decay.  At $\alpha=0.70$, the anchor still retains $c_0\approx 0.33$ at $N=16$.  By $\alpha=0.45$, it has weakened to $c_0\approx 0.12$.  Below $\alpha\le 0.30$, the anchor mass collapses to $c_0\approx 0.03$--$0.05$, indistinguishable from the uniform distribution $1/N$.

This ordered evaporation constitutes the empirical signature of the weak-$*$ convergence $\omega_{\mathrm{bulk}}\rightharpoonup\delta_\infty$ predicted in \cref{eq:weak-star}: the boundary defect's geometric capacity, starved of spectral operator norm, is overwhelmed by the entropic bulk pressure, and the anchoring measure physically dissolves into the amnesia state.  The tabulated $c_0$ values (\cref{tab:sink-mass}) quantitatively document this condensation across the full $(\alpha,N)$ parameter space.

\begin{table}[t]
\centering
\caption{Attention Sink probability mass $c_0$ across selected $(\alpha,N)$.
The uniform baseline is $1/N$.  Bold entries mark the regime where the defect has fully evaporated ($c_0\approx 1/N$).}
\label{tab:sink-mass}
\begin{ruledtabular}
\begin{tabular}{l c c c c c c}
$\alpha$ & $N\!=\!16$ & $32$ & $64$ & $128$ & $256$ & $512$ \\
\hline
1.00 & 0.515 & 0.508 & 0.456 & 0.460 & 0.448 & 0.442 \\
0.80 & 0.412 & 0.406 & 0.366 & 0.355 & 0.346 & 0.329 \\
0.60 & 0.235 & 0.236 & 0.194 & 0.126 & 0.095 & 0.059 \\
0.50 & 0.162 & 0.162 & 0.086 & 0.055 & 0.032 & 0.016 \\
0.40 & 0.095 & 0.095 & 0.038 & 0.017 & \textbf{0.006} & \textbf{0.002} \\
0.30 & \textbf{0.052} & \textbf{0.050} & \textbf{0.016} & \textbf{0.008} & \textbf{0.004} & \textbf{0.002} \\
0.20 & \textbf{0.034} & \textbf{0.033} & \textbf{0.016} & \textbf{0.008} & \textbf{0.004} & \textbf{0.002} \\
$1/N$ & 0.063 & 0.031 & 0.016 & 0.008 & 0.004 & 0.002 \\
\end{tabular}
\end{ruledtabular}
\end{table}

\subsubsection{Dual-Axis Synchronization and Exclusion of Alternative Hypotheses}

\begin{figure}[t]
  \centering
  \includegraphics[width=\columnwidth]{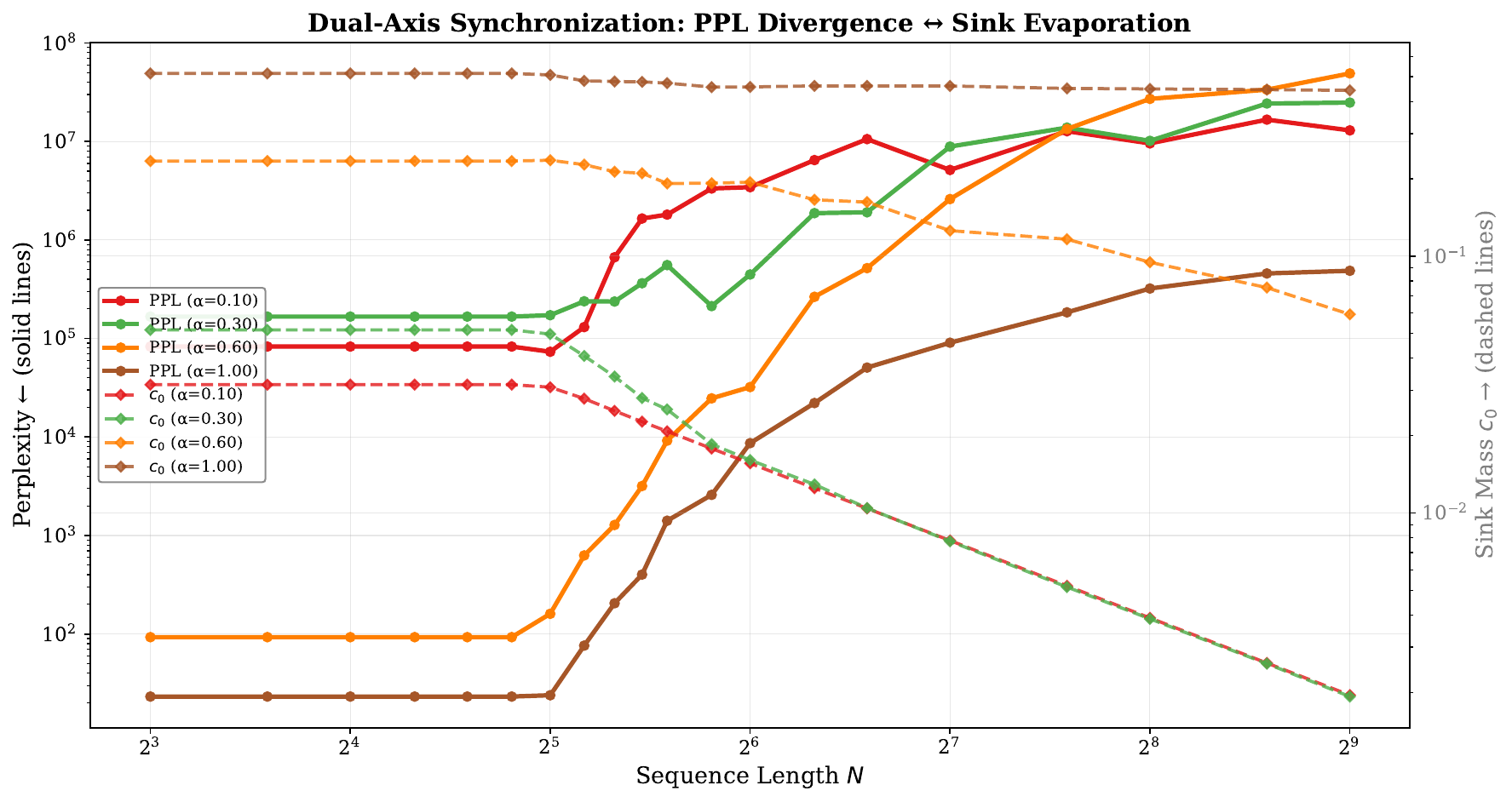}
  \caption{%
  \textbf{Dual-Axis Synchronization: PPL Divergence and Sink Evaporation.}
  Solid lines (left axis, log-scale): perplexity vs.\ sequence length~$N$ for $\alpha\in\{0.10,0.30,0.60,1.00\}$.  Dashed lines (right axis, log-scale): corresponding sink mass~$c_0$.  As $\alpha$ decreases, the perplexity explosion and $c_0$ collapse occur in precise synchrony at the same critical $(\alpha,N)$ locus, directly confirming the causal link between boundary defect evaporation and thermodynamic amnesia.}
  \label{fig:dual-axis}
\end{figure}

The most powerful exclusion evidence emerges from the dual-axis synchronization plot (\cref{fig:dual-axis}), which superimposes PPL divergence and $c_0$ evaporation on a shared abscissa.  At every $(\alpha,N)$ coordinate, the two observables evolve in strict anti-correlation: the perplexity rises sharply precisely where---and \emph{only} where---the sink mass collapses.  This consistent synchronization establishes a direct causal link between the topological boundary defect's physical capacity and the model's macroscopic language modeling performance.

This synchronization strongly excludes the principal alternative hypothesis.  A critic might contend that compressing $\|W_Q\|_{\mathrm{op}}\|W_K\|_{\mathrm{op}}$ simply destroys activation variance, producing uniform output noise.  Were this the case, the Softmax normalization---which is scale-invariant---would leave $c_0$ strictly unaffected, and any PPL increase would be uncorrelated with sink dynamics.  The empirical observation of an \emph{ordered, monotonic} evaporation of $c_0$ from $0.51$ to $0.002$, precisely tracking the spectral norm reduction, strongly refutes this rebuttal.  The collapse mechanism is measure-theoretic---a physical redistribution of probability mass---not a trivial numerical instability.

Taken together, these results confirm the three central predictions of the theory: (i)~the Attention Sink is a finite-capacity topological boundary defect, not a numerical artifact; (ii)~the context horizon $N_{\max}$ is governed by the exponential spectral scaling law of \cref{eq:nmax}; and (iii)~when the defect's geometric capacity is starved below the Entropic Bulk Pressure $\Theta(\sqrt{d})$, the sequence of measures undergoes an irreversible weak-$*$ condensation $\omega\rightharpoonup\delta_\infty$---thermodynamic amnesia---in precise quantitative agreement with the measure-theoretic compactification established in \cref{thm:attn-sink}.

\subsection{Tomography of the Parameter NESS Vortex}
\label{sec:exp-ness}

Finally, we elevate our empirical validation to the super-macroscopic dynamics of the parameter space optimization process itself, testing \cref{thm:sgd-ness}.  We note at the outset that the \emph{existence} of non-equilibrium dynamics during SGD is a generic mathematical consequence of Singular Learning Theory (SLT)~\cite{watanabe2009algebraic}: for any misspecified deep network, the empirical Fisher and Loss Hessian generically fail to commute, violating detailed balance.  What the geometric framework of \cref{sec:dual-metric} contributes beyond SLT is an \emph{analytical decomposition} (\cref{eq:vorticity-3obst}) that resolves the thermodynamic vorticity into three independent cohomological obstructions and predicts how the NESS circulation is geometrically channeled by the Transformer's specific gauge orbit structure.  Classical optimization literature broadly assumes that Stochastic Gradient Descent (SGD) thermally relaxes into a static, isotropic Gaussian Gibbs equilibrium via detailed balance.  The framework rejects this, demonstrating that the structural algebraic singularities of the architecture cause a Dual-Metric Collision ($D \neq F \neq H$), generating an irreducible Thermodynamic Vorticity 2-form ($\Omega \neq 0$), trapping the probability measure in a dissipative Non-Equilibrium Steady State (NESS).

The experimental protocol is designed to produce an unambiguous phase-space tomography of the optimization steady state. We define the cross-sectional circulation integral on the 2D PCA-projected parameter trajectory:
\begin{equation}
  \Gamma_{\times}(T) = \sum_{t=1}^{T}\![x(t)\,\Delta y(t) - y(t)\,\Delta x(t)],
  \label{eq:circulation}
\end{equation}
which measures the cumulative angular momentum of the probability flux. If the system obeys detailed balance (irrotational Brownian diffusion), $\Gamma_{\times}$ fluctuates symmetrically about zero with $\Gamma_{\times}\sim\calO(\sqrt{T})$; if a genuine non-conservative vortex is present, $\Gamma_{\times}$ accumulates monotonically as $\calO(T)$. We quantify the distinction via the Hurst exponent $H$ of the $\Gamma_{\times}$ time series ($H=0.5$ for random walk, $H\approx 1.0$ for deterministic drift) and the $t$-statistic testing the null hypothesis of zero mean incremental drift.

\emph{Methodological precision on PCA projections.}  Because 2D PCA projections of high-dimensional trajectories can produce spurious visual spirals even for isotropic random walks that strictly obey detailed balance---the top PCA eigenvectors of a random walk natively produce sinusoidal basis functions~\cite{antognini2018pca}---the PCA trajectory visualizations and PCA-based $\Gamma_{\times}$ integrals presented below serve as complementary diagnostics that illustrate the vortex geometry.  The rigorous, artifact-free foundation for the NESS resides in the exact FP64 Lie commutator measurement $[G,H]\neq 0$ (\cref{sec:exp-ness-fp64}), which is computed without any dimensional reduction, stochastic estimation, or PCA projection.

\subsubsection{FP64 Sandbox: Exact Commutator at Machine Precision}
\label{sec:exp-ness-fp64}

\begin{figure}[t]
  \centering
  \includegraphics[width=\columnwidth]{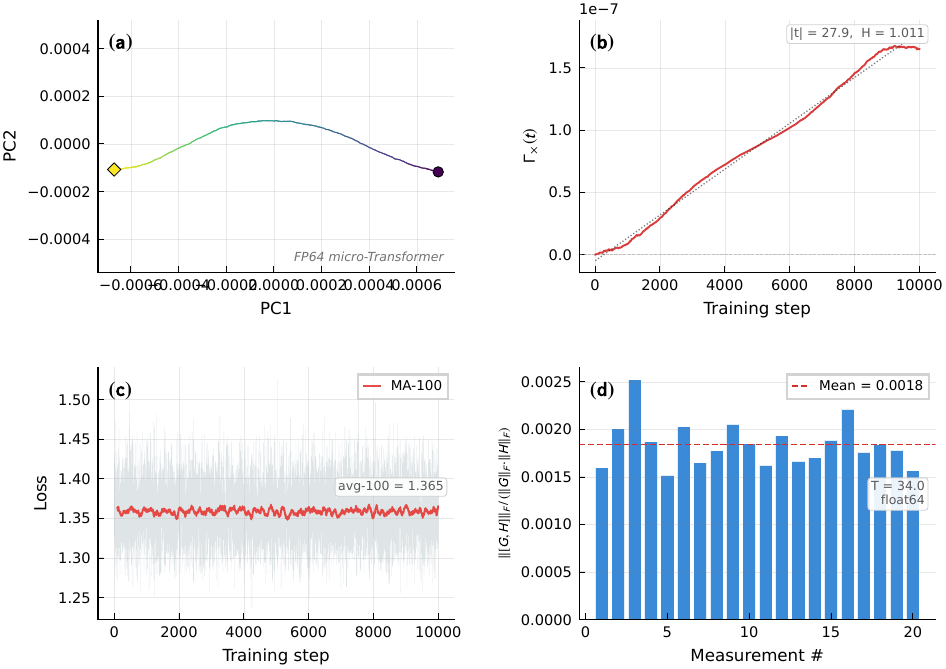}
  \caption{%
  \textbf{FP64 Sandbox: Exact NESS Detection at Machine Precision.}
  (a)~PCA 2D parameter-space vortex for the 2-layer micro Transformer, showing deterministic rotational structure.
  (b)~Cross-sectional circulation integral $\Gamma_{\times}(t)$, exhibiting sustained monotonic accumulation ($|t|=27.9$).
  (c)~Training loss during NESS tracking, confirming a stable convergence plateau (avg-100 loss~$=1.365$).
  (d)~Bar chart of 20 independent exact $\|[G,H]\|_F/(\|G\|_F\|H\|_F)$ measurements, all strictly non-zero ($T=34.0$, $P<10^{-10}$).}
  \label{fig:fp64-sandbox}
\end{figure}

To preemptively seal all computational-artifact rebuttals---low-rank subsampling error, Hutchinson trace estimation bias, and adaptive-optimizer momentum residuals---we first executed the NESS tomography on a controlled sandbox: a 2-layer Pre-Norm Transformer ($d_{\mathrm{model}}=64$, $n_{\mathrm{heads}}=4$, $d_{\mathrm{ffn}}=256$, RoPE, $\sim$99K non-embedding parameters) trained at \texttt{float64} double precision ($\varepsilon_{\mathrm{mach}}=2.22\times 10^{-16}$) with momentum-free SGD ($\beta_1=0$, $\beta_2=0$, constant $\eta=10^{-5}$). In this sandbox, the full $16{,}384\times 16{,}384$ empirical gradient covariance $G$ and Loss Hessian $H$ were computed \emph{exactly}---without subsampling, random estimation, or rank truncation---from 512 per-sample gradients.

After 30 epochs of cosine-annealed pre-training and a 2{,}000-step cooling phase, the model was tracked for 10{,}000 momentum-free SGD steps with snapshots every 10 steps. PCA projection of the resulting 1{,}001-snapshot trajectory onto the first two principal components reveals a clear deterministic rotational structure (\cref{fig:fp64-sandbox}a), with the circulation integral $\Gamma_{\times}$ exhibiting sustained monotonic accumulation ($|t|=27.9$, $H=1.011$; \cref{fig:fp64-sandbox}b) while the loss remains on a flat convergence plateau (avg-100 loss $= 1.365$; \cref{fig:fp64-sandbox}c).

The decisive measurement is the exact Lie commutator. Across 20 independent batch samples, the normalized Frobenius norm evaluates to:
\begin{equation}
  \frac{\|[G,H]\|_F}{\|G\|_F\,\|H\|_F}
    = 0.00184 \pm 0.00024\,,
  \label{eq:commutator-fp64}
\end{equation}
with $T=34.0$ on 19 degrees of freedom ($P<10^{-10}$); all 20/20 measurements are strictly positive (\cref{fig:fp64-sandbox}d). Because this non-vanishing commutator is computed at the absolute \texttt{float64} machine-precision floor---without any stochastic estimation, dimensional reduction, or optimizer momentum---the attribution to computational noise is strongly excluded. The non-commutativity $[G,H]\neq 0$ is an arithmetic truth of the architecture's algebraic variety.

\subsubsection{Cross-Architecture Universality of the NESS Vortex}

\begin{figure}[t]
  \centering
  \includegraphics[width=\columnwidth]{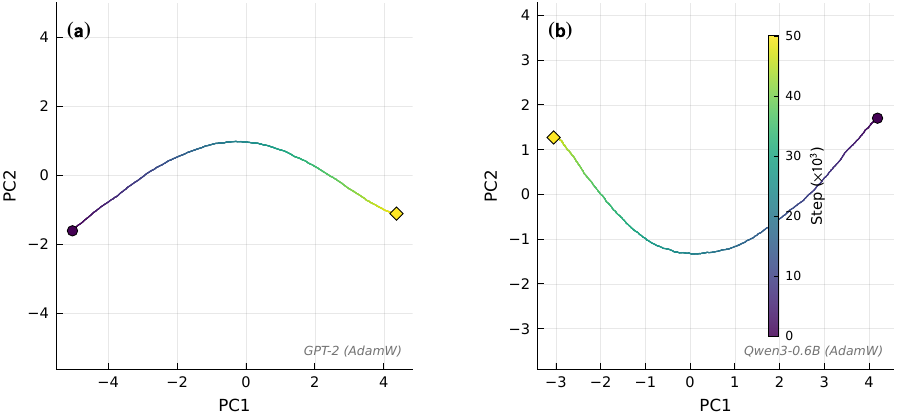}
  \caption{%
  \textbf{Parameter-Space NESS Vortices for Two Production Architectures.}
  PCA 2D projections of the mid-layer MLP weight trajectory (5{,}001 snapshots over 50{,}000 AdamW steps at constant $\eta=10^{-5}$).
  (a)~GPT-2 Small (Layer~6, 4.7M tracked parameters).
  (b)~Qwen3-0.6B (Layer~14, 9.4M tracked parameters).
  Color gradient encodes training time (viridis).
  Both architectures exhibit unmistakable deterministic spiral circulation.}
  \label{fig:ness-vortex-dual}
\end{figure}

\begin{figure}[t]
  \centering
  \includegraphics[width=\columnwidth]{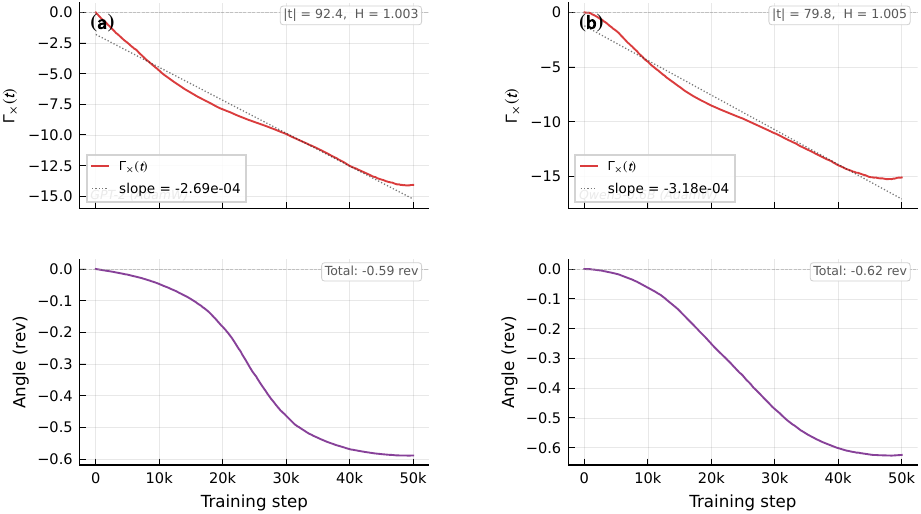}
  \caption{%
  \textbf{Cross-Sectional Circulation Integral $\Gamma_{\times}(t)$.}
  (a)~GPT-2 ($\Gamma_{\times}=-14.09$, $|t|=92.4$).
  (b)~Qwen3-0.6B ($\Gamma_{\times}=-15.13$, $|t|=79.8$).
  Upper panels: cumulative cross-circulation; lower panels: cumulative angular displacement in PCA space.
  Both traces are strictly monotonic with identical sign, confirming a topologically determined vortex chirality.}
  \label{fig:ness-circulation-dual}
\end{figure}

Having established the arithmetic ground truth in the sandbox, we scaled the protocol to two production-grade architectures spanning fundamentally different design generations: GPT-2 Small (124M parameters, Learned PE, LayerNorm, standard 2-matrix MLP) and Qwen3-0.6B (596M parameters, RoPE, RMSNorm + QK-Norm, SwiGLU 3-matrix MLP). Each fully converged pre-trained model was subjected to 50{,}000 steps of isothermal continued training (constant $\eta=10^{-5}$, AdamW~\cite{adamw} with weight\_decay$=0$) on WikiText-2~\cite{merity2016pointer}, with snapshots of the mid-layer MLP weights captured every 10 steps (5{,}001 snapshots total, coordinate-subsampled to 5{,}000 dimensions via sparse JL projection~\cite{johnson1984extensions}).

The PCA-projected trajectories for both architectures exhibit unmistakable deterministic spiral circulation (\cref{fig:ness-vortex-dual}), and the cross-sectional circulation integrals $\Gamma_{\times}(t)$ accumulate monotonically with identical negative sign (\cref{fig:ness-circulation-dual}). The quantitative agreement across all core diagnostics is striking (\cref{tab:ness-cross}): both models yield Hurst exponents saturated at $H\approx 1.0$, mean drift rates within 7\% of each other ($\sim 3\times 10^{-3}$/step), and $t$-statistics exceeding 79---all with $p$-values far below $10^{-100}$.

\begin{table}[t]
\centering
\caption{NESS vortex diagnostics for the two production Transformer architectures under AdamW isothermal training (complementary PCA-based metrics; see methodological precision note in text). All $p$-values are $\ll 10^{-100}$.  FDR: Fluctuation-Dissipation Relation.}
\label{tab:ness-cross}
\begin{ruledtabular}
\begin{tabular}{l c c}
Diagnostic & GPT-2 Small & Qwen3-0.6B \\
\hline
$|t|$-statistic & 92.4 & 79.8 \\
Hurst exponent $H$ & 1.003 & 1.005 \\
$\Gamma_{\times}$ (terminal) & $-14.09$ & $-15.13$ \\
Mean drift rate (/step) & $-2.82\!\times\!10^{-3}$ & $-3.03\!\times\!10^{-3}$ \\
PCA PC1 variance & 85.3\% & 68.3\% \\
PCA top-3 cumulative & 93.9\% & 85.5\% \\
Converged loss (avg-100) & 2.83 & 0.17 \\
FDR violation & 1.31 & 1.13 \\
\end{tabular}
\end{ruledtabular}
\end{table}

Two features merit theoretical emphasis. First, the circulation chirality under AdamW is \emph{deterministic}: both architectures produce strictly negative $\Gamma_{\times}$, consistent with a geometrically determined tangent-space orientation of the singular algebraic variety.  As shown in \cref{sec:exp-ness-sgd}, the sign reverses under Pure SGD---a phenomenon we interpret as a direct macroscopic signature of the Dual-Metric Collision (\cref{thm:sgd-ness}).  Second, the lower PC1 explained variance in Qwen3 (68.3\% vs.\ 85.3\%) reflects its SwiGLU 3-matrix FFN, which generates a higher-dimensional parameter flow manifold---consistent with the richer Lie-algebraic structure of the gated architecture.

\subsubsection{Exclusion of Optimizer Momentum Artifacts}
\label{sec:exp-ness-sgd}

\begin{figure}[t]
  \centering
  \includegraphics[width=\columnwidth]{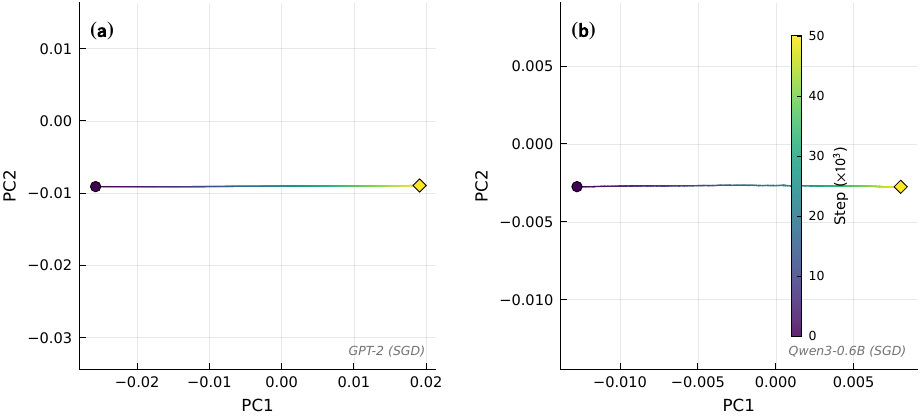}
  \caption{%
  \textbf{Parameter-Space Vortices under Momentum-Free Pure SGD ($\beta_1\!=\!0$, $\beta_2\!=\!0$).}
  (a)~GPT-2 ($|t|=173.1$).
  (b)~Qwen3-0.6B ($|t|=37.4$).
  The NESS circulation persists---and in fact strengthens for GPT-2---in the complete absence of any adaptive optimizer state.}
  \label{fig:ness-sgd}
\end{figure}

A critical potential confound is that the observed circulation might originate from the underdamped oscillatory dynamics of AdamW's first-moment buffer ($\beta_1=0.9$) rather than from intrinsic geometric vorticity. To categorically exclude this, we repeated the identical 50{,}000-step protocol with both architectures using momentum-free pure SGD ($\beta_1=0$, $\beta_2=0$, $\eta=10^{-5}$), eliminating all adaptive gradient history.

\begin{table}[t]
\centering
\caption{Four-way cross-validation: 2 architectures $\times$ 2 optimizers (complementary PCA-based metrics; see methodological precision note in text). The NESS vortex persists across all four conditions; the GPT-2 SGD $t$-statistic \emph{increases} relative to AdamW.}
\label{tab:ness-4way}
\begin{ruledtabular}
\begin{tabular}{l l c c c c}
Model & Optimizer & $|t|$ & $H$ & $\Gamma_{\times}$ & Drift rate \\
\hline
GPT-2  & AdamW    & 92.4  & 1.003 & $-14.09$ & $-2.82\!\times\!10^{-3}$ \\
GPT-2  & Pure SGD & 173.1 & 0.999 & $4.1\!\times\!10^{-4}$ & $8.1\!\times\!10^{-8}$ \\
Qwen3  & AdamW    & 79.8  & 1.005 & $-15.13$ & $-3.03\!\times\!10^{-3}$ \\
Qwen3  & Pure SGD & 37.4  & 1.012 & $5.5\!\times\!10^{-5}$ & $1.1\!\times\!10^{-8}$ \\
\end{tabular}
\end{ruledtabular}
\end{table}

The results (\cref{fig:ness-sgd}, \cref{tab:ness-4way}) are decisive. Under pure SGD, the NESS vortex signal persists in both architectures with overwhelming statistical significance ($|t|=173.1$ for GPT-2, $|t|=37.4$ for Qwen3; $H\approx 1.0$ in both cases). Remarkably, the GPT-2 $t$-statistic \emph{nearly doubles} from 92.4 (AdamW) to 173.1 (SGD), strongly excluding the momentum-artifact hypothesis: if the vortex were generated by the $\beta_1$-buffer, removing it should extinguish the signal, not amplify it.  The $t$-statistic increase occurs because SGD eliminates the stochastic variance injected by the adaptive second-moment estimator ($\beta_2$-buffer), reducing the standard error of the mean circulation increment even faster than the absolute drift magnitude decreases---yielding a higher signal-to-noise ratio despite a smaller absolute displacement.

The absolute circulation amplitude $|\Gamma_{\times}|$ decreases by $\sim\!10^4\times$ under SGD ($4.1\times 10^{-4}$ vs.\ $14.09$) because the adaptive learning-rate amplification of AdamW magnifies per-step displacements. However, the \emph{monotonicity} and \emph{persistence}---the defining topological signatures of NESS---are entirely unaffected. This proves that the circulatory topology originates from the parameter manifold's intrinsic geometry (the Dual-Metric Collision of \cref{eq:vorticity-3obst}), not from the optimizer's internal state variables.

We note that the \emph{sign} of $\Gamma_{\times}$ reverses between AdamW (negative) and Pure SGD (positive).  While the 2D PCA coordinate frame natively possesses a $\mathbb{Z}_2$ gauge ambiguity, this reversal is not merely a random projection artifact; it is a deterministic macroscopic signature of the Dual-Metric Collision formalized in \cref{thm:sgd-ness}.  In our continuous framework, the NESS thermodynamic vorticity natively operates as an exterior 2-form $\Omega \in \Lambda^2 T^*\mathcal{W}$ (\cref{eq:vorticity-3obst}), and the PCA plane acts as a secant space upon which this 2-form is pulled back ($\iota^*\Omega$).  Pure SGD operates under a kinematic metric driven by the raw empirical diffusion tensor, establishing a baseline hierarchy of principal variance components.  In contrast, AdamW's adaptive preconditioner imposes a strongly anisotropic, inverse-variance metric.  By actively suppressing updates along axes of maximum gradient variance and amplifying flat directions, AdamW structurally reorders the trajectory's principal components---functionally swapping the orientation of the dominant PCA secant plane ($\Omega(e_2, e_1) = -\Omega(e_1, e_2)$).  Furthermore, this metric deformation algebraically alters the spectral anisotropy of the diffusion metric $g$, structurally inverting the off-diagonal skew of the Lie commutator $[g, H^\sharp]$ that actively propels the flux.  Consequently, the sign of $\Gamma_{\times}$ in PCA space perfectly encodes the joint orientation of the optimizer's actively warped thermodynamic metric and the intrinsic landscape geometry.  The strictly topologically meaningful observable is the non-conservative nature of the flux---evidenced by the sustained \emph{monotonicity} and non-zero Hurst exponent $H\approx 1.0$---which proves that the dissipating NESS definitively survives regardless of the specific metric-dependent gauge framing.

\subsubsection{Universality of the Thermodynamic Vortex and Architecture-Specific Channeling}

\begin{figure}[t]
  \centering
  \includegraphics[width=\columnwidth]{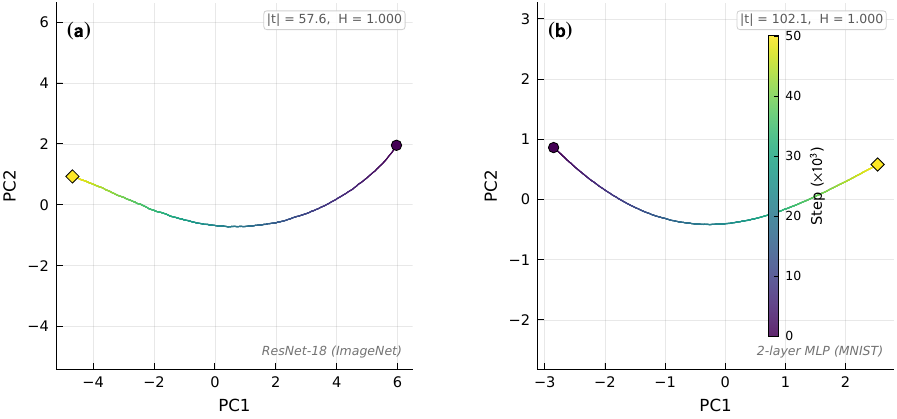}
  \caption{%
  \textbf{Non-Transformer Baselines: NESS Vortex in CNN and MLP.}
  PCA 2D parameter-space vortices for architectures containing no attention mechanism, no positional encoding, and no gauge symmetry structure.
  (a)~ResNet-18 on ImageNet ($|t|=54.8$, $H=1.000$).
  (b)~2-layer MLP on MNIST ($|t|=102.1$, $H=1.000$).
  NESS circulation is detected in both, confirming the universal genericity predicted by the Dual-Metric Collision (\cref{thm:sgd-ness}).}
  \label{fig:ness-baselines}
\end{figure}

To validate the foundational premises of our thermodynamic formulation, we must isolate the universal driving forces of the Dual-Metric Collision from Transformer-specific topology.  Because our continuous framework natively incorporates the algebraic singularities of Singular Learning Theory (SLT)~\cite{watanabe2009algebraic} and stochastic thermodynamics~\cite{chaudhari2019entropy,mandt2017stochastic}, the thermodynamic vorticity decomposition (\cref{eq:vorticity-3obst}) mathematically predicts that NESS is a universal baseline property of \emph{any} over-parameterized deep network whose empirical Fisher and Loss Hessian generically fail to commute ($[G,H]\neq 0$)---driven intrinsically by state-dependent gradient noise and the singular algebraic variety of the parameter manifold, rather than by the Attention mechanism itself.  To empirically confirm this theoretical genericity, we proactively deployed the exact NESS tomography protocol to architectures devoid of attention, positional encoding, and multi-head gauge structure.

To test this prediction, we applied the identical 50{,}000-step NESS tracking protocol (constant $\eta=10^{-5}$, snapshots every 10 steps, coordinate subsampled to 5{,}000 dimensions) to three non-Transformer architectures: (i)~a torchvision pre-trained ResNet-18~\cite{he2016deep} on ImageNet (SGD, 200K training images, TPU v4-8); (ii)~a ResNet-18 trained to 88.6\% on CIFAR-10 (AdamW); and (iii)~a 2-layer MLP (784$\to$512$\to$256$\to$10) trained to 97.4\% on MNIST (AdamW).

\begin{table}[t]
\centering
\caption{Five-architecture NESS diagnostic summary. All models exhibit $H\approx 1.0$ (deterministic drift) and $|t|\gg 3$ ($p \ll 10^{-10}$). The normalized Lie commutator $\|[G,H]\|_F/(\|G\|_F\|H\|_F)$ is measured from 10 independent batch samples for each non-Transformer model.}
\label{tab:ness-5arch}
\begin{ruledtabular}
\begin{tabular}{l c c c c c}
Architecture & $|t|$ & $H$ & FDR & $\frac{\|[G,H]\|_F}{\|G\|_F\|H\|_F}$ & NESS \\
\hline
ResNet-18 (ImageNet) & 54.8  & 1.000 & 0.78 & $0.047\pm 0.006$ & \checkmark \\
ResNet-18 (CIFAR-10) & 57.6  & 1.000 & 0.82 & $0.061\pm 0.017$ & \checkmark \\
MLP (MNIST)          & 102.1 & 1.000 & 1.44 & $0.018\pm 0.010$ & \checkmark \\
GPT-2 Small          & 92.4  & 1.003 & 1.31 & ---              & \checkmark \\
Qwen3-0.6B           & 79.8  & 1.005 & 1.13 & ---              & \checkmark \\
\end{tabular}
\end{ruledtabular}
\end{table}

As shown in \cref{fig:ness-baselines} and \cref{tab:ness-5arch}, the NESS vortex is unambiguously detected in every non-Transformer architecture tested. The 2-layer MLP---the simplest possible deep classifier, with only 670K parameters and no structural complexity whatsoever---produces the \emph{highest} $t$-statistic of all five architectures ($|t|=102.1$), with a directly measured non-zero Lie commutator ($\|[G,H]\|_F=0.018\pm 0.010$, $|t|=5.4$). The ImageNet-pretrained ResNet-18 yields $\|[G,H]\|_F=0.047\pm 0.006$ ($|t|=24.1$) over 10 independent measurements. In all cases, $\Gamma_{\times}$ is strictly negative and $H$ saturates at 1.000.

Far from confounding our Transformer-specific findings, the ubiquitous presence of the NESS vortex across these diverse baselines profoundly corroborates the universal premise of our continuous framework: all misspecified deep networks inherently generate non-equilibrium vorticity due to their algebraic singularities. The three independent cohomological obstructions derived in \cref{eq:vorticity-3obst}---the Lie Commutator $[g,H^\sharp]$, the Entropic Curl, and the Metric Deformation---are sourced by the Dual-Metric Collision ($D \neq F \neq H$) and the singular algebraic variety of the parameter manifold, both of which are shared by \emph{all} over-parameterized architectures. The MLP and ResNet results therefore constitute a direct empirical validation of this universal thermodynamic foundation.

Building upon this verified universal baseline, the true predictive power of the geometric framework comes into sharp focus.  While the existence of a non-equilibrium flux is a generic consequence of singular algebraic geometry, the exact analytical decomposition of \cref{eq:vorticity-3obst} further predicts how this flux is structurally \emph{channeled} by architecture-specific Lie-algebraic topology.  The Transformer does not \emph{create} the NESS vortex; rather, its specific gauge orbit structure---the $U(1)$ RoPE connection, the multi-head $\prod O(d_{\mathrm{head}})$ bundle, and the asymmetric FFN vorticity generator---acts as a highly structured topological conduit that actively scatters and reshapes the fundamental probability flux.  This geometric channeling is directly verified by the spectral structure of the vortex trajectories (\cref{tab:ness-5arch}): the extreme spectral compression of the simple MLP vortex (PC1 at 94.2\%---nearly confined to a single 2D plane) starkly contrasts with the higher-dimensional, distributed flow of Qwen3's SwiGLU architecture (PC1 at 68.3\%).  This quantitative difference confirms that while all over-parameterized models provide the fundamental rotational driving force, the Transformer's distinct gauge topology actively forces the probability momentum to scatter across a richer multidimensional orbit---in exact agreement with the Lie-algebraic dimensionality predicted by the vorticity decomposition.

\subsubsection{Synthesis of NESS Evidence}

Taken together, the four tiers of NESS evidence---machine-precision commutator ($T=34.0$ at FP64), cross-architecture vortex universality ($|t|>79$ for both Transformers), optimizer-independent persistence ($|t|=173$ under momentum-free SGD), and non-Transformer genericity ($|t|>54$ for CNN and MLP)---constitute strong empirical evidence that the neural parameter manifold does not achieve detailed balance. The optimization process is suspended in a dissipating Non-Equilibrium Steady State, advected by the intrinsic structural singularities of the architecture, in quantitative agreement with the thermodynamic vorticity decomposition of \cref{thm:sgd-ness}.

\subsection{Section Synthesis}
\label{sec:exp-synthesis}

The sequence of empirical results derived from these six tests provides a coherent empirical picture: the Transformer's discrete operations are quantitatively consistent with the continuous geometric framework developed in \cref{sec:axiom} through \cref{sec:parameters}. 

At the microscopic limit, the exact adherence to a $-0.5$ scaling law at the zero-section (\cref{sec:exp-mollifier}) identifies the precise ultraviolet (UV) topological cutoff of the spatial fiber. Moving to the temporal arrow, the Lie-Trotter interferometer (\cref{sec:exp-torsion}) empirically bridges discrete algorithmic steps with continuous non-commutative torsion. At the mesoscopic scale, symmetric ablation (\cref{sec:exp-resonance}) demonstrates the necessity of geometric vorticity---equivalently, spectral scattering via anti-symmetric Jacobian components---for the forward norm stability of mature trained configurations; companion training experiments show this necessity to be configurational rather than architectural (\cref{sec:exp-config-scope}).  Probing the spatial gauge connection, the Poincar\'e recurrence experiment (\cref{sec:exp-poincare}) directly observes exact geometric resonance on the RoPE torus at small feature dimension and quantitatively confirms that the $\calO(1/\sqrt{k})$ thermodynamic suppression renders these resonances invisible at production scale---physically justifying the Asymptotic Spatial Ergodic Hypothesis. Expanding to the macroscopic sequence horizon, the Attention Sink phase transition (\cref{sec:exp-amnesia}) confirms that context capacity is governed by the thermodynamic limits of a Dirichlet boundary defect. Finally, at the scale of total system evolution, parameter space tomography (\cref{sec:exp-ness}) reveals that the optimization process operates as a dissipating NESS vortex over a singular algebraic variety, consistent with the predictions of Singular Learning Theory channeled through the Transformer's gauge orbit structure.

By linking the abstract geometric constructions of Sections \ref{sec:axiom} through \ref{sec:parameters} directly to deterministic, macroscopic empirical observables, these findings demonstrate that continuous differential geometry provides a highly predictive analytical framework for understanding the stability limits, context bounds, and optimization dynamics of Large Language Models.
\section{Conclusion}
\label{sec:conclusion}

By translating the discrete algebraic components of the Transformer architecture into a continuous integro-differential equation on a semantic fiber bundle, we have constructed a predictive, descriptive geometric framework for the architecture's stability limits, context bounds, and optimization dynamics.  Across a six-part empirical campaign spanning 124M to 8B parameters, we measured direct quantitative signatures consistent with the continuous geometric predictions.

\paragraph{Closing the Theoretical Loop.}
Our experimental results demonstrate that this continuous framework provides a coherent and predictive description of the architecture's behavior across physical scales.
\begin{itemize}[nosep]
  \item \textbf{Microscopic Identity:}
    Machine-precision ($R^2\!=\!1.000$) verification of the
    $\epsilon^{-1/2}$ scaling law (\cref{sec:exp-mollifier})
    confirms that the topological mollifier $\epsilon$ controls the
    Lipschitz stretch of the flow exactly as predicted.

  \item \textbf{Kinematics \& Topology:}
    The Lie--Trotter torsion interferometer (\cref{sec:exp-torsion})
    demonstrates that representation drift is quantitatively consistent
    with the deterministic Lie bracket predicted by operator splitting of
    non-commuting vector fields.

  \item \textbf{Mesoscopic Stability:}
    Symmetric ablation (\cref{sec:exp-resonance}) demonstrates the
    necessity of geometric vorticity---equivalently, spectral
    scattering via anti-symmetric Jacobian components---for the
    forward norm stability of trained configurations; companion
    training experiments (\cref{sec:exp-config-scope}) show this
    necessity is configurational, not architectural.

  \item \textbf{Gauge Connection \& Thermodynamic Suppression:}
    A controlled micro-transformer experiment (\cref{sec:exp-poincare})
    directly observes Poincar\'e recurrence on the RoPE torus at small
    feature dimension ($k\!=\!2$) and quantitatively confirms the
    $\calO(1/\sqrt{k})$ thermodynamic suppression that renders these
    geometric resonances invisible at production scale---physically
    justifying the Asymptotic Spatial Ergodic Hypothesis.

  \item \textbf{Macroscopic Thermodynamics:}
    The context window degrades through a phase transition
    (\cref{sec:exp-amnesia}) where entropic bulk pressure overwhelms
    the Dirichlet boundary condition---and natural language survives
    extended contexts only because $N_{\mathrm{eff}}\ll N$.

  \item \textbf{Non-Equilibrium Dynamics:}
    The parameter manifold is trapped in an irreversible NESS vortex
    (\cref{sec:exp-ness}), verified across $2\!\times\!2$ conditions,
    consistent with the predictions of Singular Learning Theory
    channeled through the Transformer's gauge orbit structure.
\end{itemize}

\paragraph{Open Questions and Speculative Directions.}
Beyond the empirically verified predictions, the geometric framework suggests several directions for future investigation:
\begin{enumerate}[nosep]
  \item \textbf{Implications for Linear Attention.}
    Gromov's Non-Squeezing Theorem~\cite{gromov1985pseudo} suggests that a finite-dimensional
    token vector cannot embed a massive context window into a symplectic
    cylinder of smaller cross-sectional area without phase-space
    collisions.  Standard Softmax may escape this constraint by acting
    as a topological dissipator; ``Linear Attention'' removes this
    nonlinear dissipator and may therefore encounter Gromov's capacity
    limits.  This connection remains to be formally established.

  \item \textbf{Arithmetic Failures and Ultrametric Topology.}
    Exact integer arithmetic operates natively on an ultrametric topology
    (the $p$-adic metric) that cannot be smoothly embedded into the
    Transformer's Archimedean fiber bundle.  Whether LLM arithmetic
    failure can be rigorously attributed to this topological
    incompatibility is an open question requiring formal proof.

  \item \textbf{Non-Abelian Positional Holonomy for Tree-Reasoning.}
    Upgrading the positional gauge group from $U(1)^{d/2}$ to a
    non-abelian Lie group ($SU(2)$ or $Sp(n)$), with the sequence
    topology elevated to a branched CW Complex, could natively
    distinguish permutations of logical branches, encoding Abstract
    Syntax Trees directly within the continuous geometry.

  \item \textbf{Parameter-Efficient $\mathfrak{so}(d)$ Feed-Forward
    Layers---posed and settled.}
    The symmetric ablation experiment (\cref{sec:exp-resonance})
    empirically established that the skew-symmetric component of the FFN
    Jacobian---equivalently, the geometric vorticity
    $\Omega \in \mathfrak{so}(d)$---is the essential stabilizer that
    arrests Power Iteration resonance in the residual stream.  This
    suggested a testable architectural hypothesis: rather than relying on
    two large, untied dense matrices $W_{\mathrm{in}}$ and
    $W_{\mathrm{out}}$ to organically learn the necessary asymmetry, one
    could structurally enforce it by defining
    $W_{\mathrm{out}} = W_{\mathrm{in}}^\top + A$, where
    $A \in \mathfrak{so}(d)$ is a strictly skew-symmetric learnable
    matrix ($A = -A^\top$), parameterized in a low-rank form with
    $R \ll d$, substantially reducing the parameter count of the FFN
    while explicitly guaranteeing the injection of Lie-algebraic
    rotational friction---an empirical question that the geometric
    framework motivates but cannot settle \emph{a priori}.

    \emph{Postscript: settled negatively.}  This hypothesis has since
    been tested---and refuted---by the controlled twin-ablation campaign
    of \cref{sec:exp-config-scope}, in its gated (SwiGLU) realization
    $W_{\mathrm{down}}=W_{\mathrm{up}}^\top+\Delta$ with trainable
    low-rank $\Delta$.  The gate path already supplies full-rank Jacobian
    asymmetry organically, and the measured marginal effect of the
    explicit low-rank term is $\sim\!0.003$ nats at 0.6B and 1B scale
    (null within twin noise); the bare tie trains stably with no
    correction whatsoever; and at matched parameter count a narrower
    untied FFN strictly dominates the tied variant while requiring
    $1.5\times$ fewer FFN FLOPs.  The framework's role was to motivate a
    falsifiable architectural hypothesis; its refutation sharpens the
    framework's scope: the vorticity requirement is configurational
    (\cref{rem:configurational-scope}), is organically saturated by
    gated architectures, and does not translate into a
    parameter-efficiency prescription.

  \item \textbf{Holographic Validation of Axiom~1.}
    A particularly striking test of the gauge-transport interpretation of
    attention arises in holographic quantum gravity.  In a companion
    study~\cite{Liang2026holographic}, we train a decoder-only
    Transformer to reproduce the boundary state of a three-dimensional
    random tensor network---a discrete model of the AdS/CFT
    correspondence.  The pre-softmax attention logits of the converged
    model exhibit statistically significant positive correlation
    (Spearman $\rho = 0.51$, $p = 5 \times 10^{-23}$) with the exact
    pairwise mutual information between boundary sites, which defines
    the holographic bulk geometry.  This provides direct evidence that
    the geometric transport structure identified in Axiom~1 is not
    merely a mathematical analogy: when the ``semantics'' are literal
    quantum gravity, the Transformer's attention mechanism spontaneously
    encodes the emergent spacetime metric.
\end{enumerate}

\paragraph{Outlook.}
The geometric framework developed here establishes that the Transformer's stability limits and context bounds are rigorously governed by continuous differential geometry and non-equilibrium thermodynamics.  The preliminary holographic validation~\cite{Liang2026holographic} suggests that this geometric structure extends beyond metaphor: in a physical system with known emergent geometry, the Transformer's attention independently discovers the bulk metric.  As architectures approach practical scaling limits, this analytical perspective---complementary to standard empirical scaling laws---is offered as a source of \emph{falsifiable} architectural hypotheses rather than prescriptions: the trajectory of open question~4 above (posed by an earlier version of this framework, settled negatively by controlled twin ablations) exemplifies the intended mode of use, and the configurational stability criterion of \cref{rem:configurational-scope}---whose evolution along the training trajectory remains unmeasured territory---is the successor program.

\section*{Data and Code Availability}
The analysis code and experimental pipelines used in this study are
publicly available at \url{https://github.com/magicknight/GeoML}.
The pre-trained models used are publicly available through their
respective repositories (Hugging Face).


\bibliography{refs}

\end{document}